\documentclass{article} 
\usepackage{fullpage}

\usepackage[ruled,linesnumbered]{algorithm2e}
\usepackage{verbatim}
\usepackage{xcolor}
\usepackage{amsmath,amssymb,amsthm}
\usepackage{hyperref}
\hypersetup{
    colorlinks,
    citecolor=black,
    filecolor=black,
    linkcolor=black,
    urlcolor=black
}
\usepackage{rotating,caption}
\usepackage{authblk}

\def\HiLi{\leavevmode\rlap{\hbox to \hsize{\color{gray!30}\leaders\hrule height .8\baselineskip depth .5ex\hfill}}}

\newtheorem{theorem}{Theorem}
\newtheorem{observation}[theorem]{Observation}

\newtheorem{problem}[theorem]{Problem}
\newtheorem{definition}{Definition}
\newtheorem{lemma}{Lemma}

\newcommand\SP[1]{\mathtt{sp}(#1)}
\newcommand\EP[1]{\mathtt{ep}(#1)} 
\newcommand\INTERVAL[1]{\mathtt{range}(#1)}
\newcommand\REPR[1]{\mathtt{repr}(#1)}
\newcommand\REPRPRIME[1]{\mathtt{repr}'(#1)}

\newcommand{\ST}{\ensuremath{\mathsf{ST}}}
\newcommand{\SA}{\ensuremath{\mathsf{SA}}}
\newcommand{\LCP}{\ensuremath{\mathsf{LCP}}}
\newcommand{\PLCP}{\ensuremath{\mathsf{PLCP}}}
\newcommand{\MS}{\ensuremath{\mathsf{MS}}}
\newcommand{\DS}{\ensuremath{\mathsf{DS}}}
\newcommand{\BWT}{\ensuremath{\mathsf{BWT}}}
\newcommand{\REV}[1]{\ensuremath{\underline{#1}}}

\newcommand{\SLT}{\ensuremath{\mathsf{SLT}}}

\newcommand{\INTERVALFUNCTION}{\ensuremath{\mathbb{I}}}

\newcommand{\LF}{\ensuremath{\mathsf{LF}}}

\newcommand{\No}[1]{}

\newcommand{\ltdots}{..}

\newcommand{\Sec}[1]{Section~\ref{#1}}

\newcommand{\Thm}[1]{Theorem~\ref{#1}}
\newcommand{\Lemma}[1]{Lemma~\ref{#1}}
\newcommand{\Algo}[1]{Algorithm~\ref{#1}}

\newcommand{\Obs}[1]{Observation~\ref{#1}}

\newcommand{\Or}{~\textbf{or}~}
\newcommand{\Not}{~\textbf{not}~}

\newcommand{\st}{\:|\:}
    
\SetAlFnt{\small}
\SetAlCapFnt{\small}
\SetAlCapNameFnt{\small}
\SetAlCapHSkip{0pt}
\IncMargin{-\parindent}

\begin{document}


\title{Linear-time string indexing and analysis in small space}
\author{Djamal Belazzougui}
\author{Fabio Cunial}
\author{Juha K\"{a}rkk\"{a}inen}
\author{Veli M\"{a}kinen}
\affil{Helsinki Institute for Information Technology}
\maketitle
\begin{abstract}
The field of \emph{succinct data structures} has flourished over the last 16 years. Starting from the compressed suffix array by Grossi and Vitter (STOC 2000) and the FM-index by Ferragina and Manzini (FOCS 2000), a number of generalizations and applications of string indexes based on the Burrows-Wheeler transform (BWT) have been developed, all taking an amount of space that is close to the input size in bits. In many large-scale applications, the construction of the index and its usage need to be considered as one unit of computation. For example, one can compare two genomes by building a common index for their concatenation, and by detecting common substructures by querying the index. Efficient string indexing and analysis in small space lies also at the core of a number of primitives in the data-intensive field of high-throughput DNA sequencing.

We report the following advances in string indexing and analysis. We show that the BWT of a string $T\in \{1,\ldots,\sigma\}^n$ can be built in deterministic $O(n)$ time using just $O(n\log{\sigma})$ bits of space, where $\sigma \leq n$. Deterministic linear time is achieved by exploiting a new 
\emph{partial rank} data structure that supports queries in constant time, and that might have independent interest. Within the same time and space budget, we can build an index based on the BWT that allows one to enumerate all the internal nodes of the suffix tree of $T$. Many fundamental string analysis problems, such as maximal repeats, maximal unique matches, and string kernels, can be mapped to such enumeration, and can thus be solved in deterministic $O(n)$ time and in $O(n\log{\sigma})$ bits of space from the input string, by tailoring the enumeration algorithm to some problem-specific computations.

We also show how to build many of the existing indexes based on the BWT, such as the \emph{compressed suffix array}, the \emph{compressed suffix tree}, and the \emph{bidirectional BWT index}, in \emph{randomized} $O(n)$ time and in $O(n\log{\sigma})$ bits of space. The previously fastest construction algorithms for BWT, compressed suffix array and compressed suffix tree, which used $O(n\log{\sigma})$ bits of space, took $O(n\log{\log{\sigma}})$ time for the first two structures, and $O(n\log^{\epsilon}n)$ time for the third, where $\epsilon$ is any positive constant smaller than one. Contrary to the state of the art, our bidirectional BWT index supports every operation in constant time per element in its output.
\end{abstract}

\vfill
\noindent\rule{\textwidth}{0.5pt}
\\[0.2\baselineskip]
{
\footnotesize
\indent This work was partially supported by Academy of Finland under grants 250345 and 284598 (CoECGR).
\\
\indent This work extends results originally presented in ESA 2013 (all authors) and STOC 2014 (Belazzougui). 
\\
\indent Author's address: (Helsinki Institute for Information Technology), Department of Computer Science, P.O. Box 68 (Gustaf H\"{a}llstr\"{o}min katu 2b), FIN-00014, University of Helsinki, Finland.
\\
\indent Djamal Belazzougui is currently with the Centre de Reserche sur L'Information Scientifique et Technique, Algeria, and Fabio Cunial is with the Max-Planck Institute of Molecular Cell Biology and Genetics, Germany.
}

\newpage

\tableofcontents
\newpage
\section{Introduction}

The suffix tree~\cite{Wei73} is a fundamental text indexing data structure that has been used for solving a large number of string processing problems over the last 40 years~\cite{Ap85,Gu97}. The suffix array~\cite{MM93} is another widely popular data structure in text indexing, and although not as versatile as the suffix tree, its space usage is bounded by a smaller constant: specifically, given a string of length $n$ over an alphabet of size $\sigma$, a suffix tree occupies $O(n\log n)$ bits of space, while a suffix array takes exactly $n \lceil \log{n} \rceil$ bits\footnote{In this paper $\log n$ stands for $\log_2 n$.}.

The last decade has witnessed the rise of \emph{compressed} versions of the suffix array~\cite{GV05,FM05} and of the suffix tree~\cite{Sa07a}. In contrast to their plain versions, they occupy just $O(n\log\sigma)$ bits of space: this shaves a $\Theta(\log_\sigma n)$ factor, thus space becomes just a constant times larger than the original text, which is encoded in exactly $n\log\sigma$ bits. Any operation that can be implemented on a suffix tree (and thus any algorithm or data structure that uses the suffix tree) can be implemented on the \emph{compressed suffix tree} (henceforth denoted by CST) as well, at the price of a slowdown that ranges from $O(1)$ to $O(\log^{\epsilon}{n})$ depending on the operation. Building a CST, however, suffers from a large slowdown if we are restricted to use an amount of space that is only a constant factor away from the space taken by the CST itself. More precisely, a CST can be built in deterministic $O(n\log^\epsilon n)$ time (where $\epsilon$ is any constant such that $0<\epsilon<1$) and $O(n\log\sigma)$ bits of space~\cite{HSS09}, or alternatively in deterministic $O(n)$ time and $O(n\log n)$ bits by first employing a linear-time deterministic suffix tree construction algorithm to build the plain suffix tree~\cite{Fa97}, and then compressing the resulting representation. 
It can also be built in deterministic $O(n\log{\log{n}})$ time and $O(n\log{\sigma}\log{\log{n}})$ bits of space (by combining ~\cite{HSS09} with \cite{hon02space}). 

The compressed version of the suffix array (denoted by CSA in what follows) does not suffer from the same slowdown in construction as the compressed suffix tree, since it can be built in deterministic $O(n\log{\log{\sigma}})$ time\footnote{This bound should actually read as $O(n\cdot\max(1,\log{\log{\sigma}}))$.} and $O(n\log\sigma)$ bits of space~\cite{HSS09}, or alternatively in deterministic $O(n)$ time and in $O(n\log{\sigma} \log{\log{n}})$ bits of space~\cite{OS09}.

In this paper we show that both the CST and the CSA can be built in \emph{randomized} $O(n)$ time using $O(n\log{\sigma})$ bits of space, where randomization comes from the use of \emph{monotone minimal perfect hash functions}\footnote{Monotone minimal perfect hash functions are defined in Section \ref{sec:build_mmphf}.}. This seems in contrast to the plain suffix array and suffix tree, which can be built in \emph{deterministic} $O(n)$ time. However, hashing is also necessary to build a representation of the plain suffix tree that supports the fundamental child operation in constant time\footnote{The constant-time child operation enables e.g. matching a pattern of length $m$ against the suffix tree in $O(m)$ time.}: building such a plain representation of the suffix tree takes itself randomized $O(n)$ time. If one insists on achieving deterministic linear construction time, then the fastest bound known so far for the child operation is $O(\log{\log{\sigma}})$.

We also show that the key ingredient of compressed text indexes, namely the \emph{Burrows-Wheeler transform} (BWT) of a string \cite{BW94}, can be built in deterministic $O(n)$ time and $O(n\log{\sigma})$ bits of space. Such construction rests on the following results, which we believe have independent technical interest and wide applicability to string processing and biological sequence analysis problems. The first result is a data structure that takes at most $n\log{\sigma}+O(n)$ bits of space, and that supports access and partial rank\footnote{Access, partial rank and select queries are defined in Section \ref{sect:definitions}.} queries in constant time, and a related data structure that takes $n\log{\sigma}(1+1/k)+O(n)$ bits of space for any positive integer $k$, and that supports either access and partial rank queries in constant time and select queries in $O(k)$ time, or select queries in constant time and access and partial rank queries in $O(k)$ time (Lemma \ref{lemma:rank_select_access}). Both such data structures can be built in deterministic $O(n)$ time and $o(n)$ bits of space. 

In turn, the latter data structure enables an index that takes $n\log{\sigma}+O(n)$ bits of space, and that allows one to enumerate a rich representation of all the internal nodes of a suffix tree, in overall $O(n)$ time and in $O(\sigma^{2}\log^{2}{n})$ bits of additional space (Lemmas \ref{lemma:rangedistinct} and \ref{lemma:iterator}). Such index is our second result of independent interest: we call it the \emph{unidirectional BWT index}. Our enumeration algorithm is easy to implement, to parallelize, and to apply to multiple strings, and it performs a depth-first traversal of the \emph{suffix-link tree}\footnote{The suffix-link tree is defined in Section \ref{sec:suffixTree}.} using a stack that contains at every time at most $\sigma\log{n}$ nodes. A similar enumeration algorithm, which performs however a breadth-first traversal of the suffix-link tree, was described in \cite{BGOS13}: such algorithm uses a queue that takes $\Theta(n)$ bits of space, and that contains $\Theta(n)$ nodes in the worst case. This number of nodes might be too much for applications that require storing e.g. a real number per node, like weighted string kernels (see \cite{belazzougui2015framework} and references therein). 

We also show that many fundamental operations in string analysis and comparison, with a number of applications to genomics and high-throughput sequencing, can all be performed by enumerating the internal nodes of a suffix tree \emph{regardless of their order}. This allows one to implement all such operations in deterministic $O(n)$ time on top of the unidirectional index, and thus in deterministic $O(n)$ time and $O(n\log{\sigma})$ bits of space \emph{directly from the input string}: this is our third result of independent interest. 
Implementing such string analysis procedures on top of our enumeration algorithm is also practical, as it amounts to few lines of code invoked by a callback function. Using again the enumeration procedure, we give a practical algorithm for building the BWT of the \emph{reverse} of a string given the BWT of the string. Contrary to \cite{OhlebuschBA14}, our algorithm does not need the suffix array and the original string in addition to the BWT.

To build the CST we make use of the \emph{bidirectional BWT index}, a data structure consisting of two BWTs that has a number of applications in high-throughput sequencing \cite{Bidirectional_search_in_a_string_with_wavelet_trees,Bidirectional_search_in_a_string_with_wavelet_trees_and_bidirectional_matching_statistics,LLTWWY09,SOAP2_An_improved_ultrafast_tool_for_short_read_alignment}. Our fourth result of independent interest consists in showing that, in randomized $O(n)$ time and in $O(n\log{\sigma})$ bits of space, we can build a bidirectional BWT index that takes $O(n\log{\sigma})$ bits of space and that supports every operation in constant time per element in the output (Theorem \ref{thm:bidirectionalIndexConstruction}). This is in contrast to the $O(\sigma)$ or $O(\log{\sigma})$ time per element in the output required by existing bidirectional indexes in some key operations. Our fifth result of independent interest is an algorithm that builds the \emph{permuted LCP array} (a key component of the CST), as well as the \emph{matching statistics array}\footnote{The permuted LCP array is defined in Section \ref{sec:bwt}. The matching statistics array and related notions are defined in Section \ref{sec:stringAnalysis}.}, given a constant-time bidirectional BWT index, in $O(n)$ time and $O(\log{n})$ bits of space (Lemmas \ref{lemma:plcpConstruction} and \ref{lemma:ms}). Both such algorithms are practical.

The paper consists of a number of other intermediate results, whose logical dependencies are summarized in Figure~\ref{fig:architecture}. We suggest to keep this figure at hand in particular while reading Section \ref{sec:bwtIndexesConstruction}.

\begin{sidewaysfigure}[t!]
\begin{center}
\includegraphics[width=1.02\textwidth]{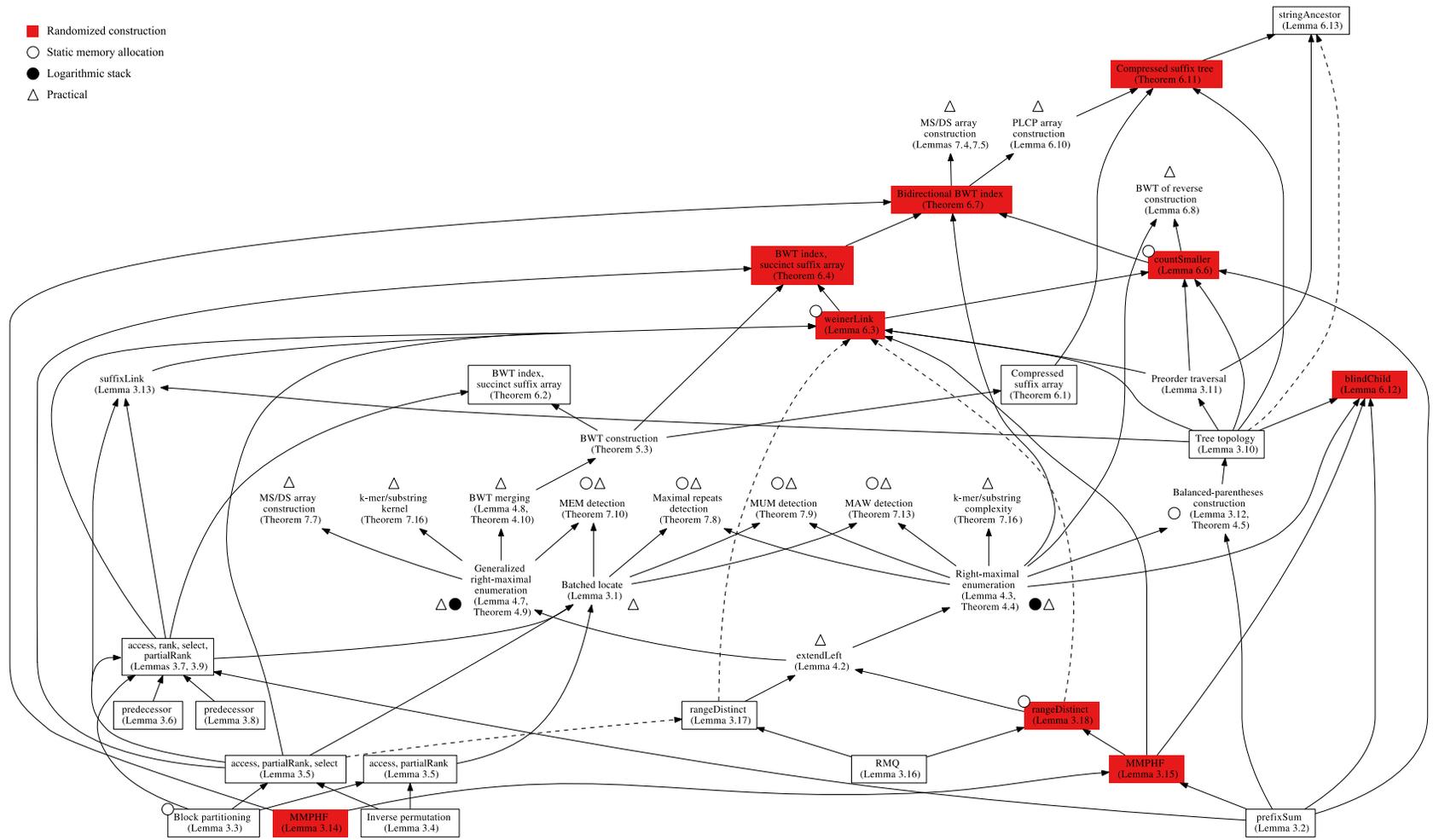}
\caption{\footnotesize
Map of the main data structures (rectangles) and algorithms described in the paper. Data structures whose construction algorithm works in randomized time are highlighted in red. Algorithms that use the static allocation strategy are marked with a white circle (see Section \ref{sec:staticAllocation}). Algorithms that use the logarithmic stack technique described in the proof of Lemma \ref{lemma:iterator} are marked with a black circle. Algorithms that are easy to implement in practice are highlighted with a triangle. Arcs indicate logical dependencies. A dashed arc $(v,w)$ means that data structure $v$ is used to build data structure $w$, but some of the components of $v$ are discarded after the construction of $w$.}
\label{fig:architecture}
\end{center}
\end{sidewaysfigure}


\section{Definitions and preliminaries} \label{sect:definitions}

We work in the RAM model, and we index arrays starting from one. We denote by $i \; (\mbox{mod}_1\; n)$ the function that returns $n$ if $i=0$, that returns $i$ if $i \in [1 \ltdots n]$, and that returns $1$ if $i=n+1$.

\subsection{Temporary space and working space}

We call \emph{temporary space} the size of any region of memory that: (1) is given in input to an algorithm, initialized to a specific state; (2) is read and written (possibly only in part) by the algorithm during its execution; (3) is restored to the original state by the algorithm before it terminates. We call \emph{working space} the maximum amount of memory that an algorithm uses in addition to its input, its output, and its temporary space (if any). The temporary space of an algorithm can be bigger than its working space.

\subsection{Strings\label{sect:definitions_strings}} 

A \emph{string} $T$ of length $n$ is a sequence of symbols from the compact alphabet $\Sigma=[1..\sigma]$, i.e. $T \in \Sigma^n$. We assume $\sigma \in o(\sqrt{n}/\log{n})$, since for larger alphabets there already exist algorithms for building the data structures described in this paper, in linear time and in $O(n\log{\sigma}) = O(n\log{n})$ bits of working space (for example the linear-time suffix array construction algorithms in \cite{KA03,KSPP05,KSB06}). The reason behind our choice of $o(\sqrt{n}/\log{n})$ will become apparent in \Sec{sec:enumeration}. We also assume $\#$ to be a separator that does not belong to $[1..\sigma]$, and specifically we set $\#=0$. In some cases we use multiple distinct separators, denoted by $\#_i=-i+1$ for integers $i>0$.

Given a string $T \in [1..\sigma]^n$, we denote by $T[i..j]$ (with $i$ and $j$ in $[1..n]$) a \emph{substring} of $T$, with the convention that $T[i..j]$ equals the empty string if $i>j$. As customary, we denote by $V \cdot W$ the concatenation of two strings $V$ and $W$. We call $T[1..i]$ (with $i \in [1..n]$) a \emph{prefix} of $T$, and $T[j..n]$ (with $j \in [1..n]$) a \emph{suffix} of $T$. 

A \emph{rotation} of $T$ is a string $T[i..n] \cdot T[1..i-1]$ for $i \in [1..n]$. We denote by $\mathcal{R}(T)$ the set of all \emph{lexicographically distinct} rotations of $T$. Note that $|\mathcal{R}(T)|$ can be smaller than $n$, since some rotations of $T$ can be lexicographically identical: this happens if and only if $T=W^k$ for some $W \in [1..\sigma]^+$ and $k>1$. We are interested only in strings for which all rotations are lexicographically distinct: we often enforce this property by terminating a string with $\#$. We denote by $\mathcal{S}(T)$ the set of all distinct, not necessarily proper, prefixes of rotations of $T$. In what follows we will use rotations to define a set of notions (like maximal repeats, suffix tree, suffix array, longest common prefix array) that are typically defined in terms of the \emph{suffixes} of a string terminated by $\#$. We do so to highlight the connection between such notions and the \emph{Burrows-Wheeler transform}, one of the key tools used in the following sections, which is defined in terms of rotations. Note that there is a one-to-one correspondence between the $i$-th rotation of $T\#$ in lexicographic order and the $i$-th suffix of $T\#$ in lexicographic order.

A \emph{repeat} of $T$ is a string $W \in \mathcal{S}(T)$ such that there are two rotations $T^1 = T[i_{1}..n] \cdot T[1..i_{1}-1]$ and $T^2 = T[i_{2}..n] \cdot T[1..i_{2}-1]$, with $i_1 \neq i_2$, such that $T^{1}[1..|W|]=T^{2}[1..|W|]=W$. Repeats are substrings of $T$ if $T \in [1..\sigma]^{n-1}\#$. A repeat $W$ is \emph{right-maximal} if $|W|<n$ and there are two rotations $T^1 = T[i_{1}..n] \cdot T[1..i_{1}-1]$ and $T^2 = T[i_{2}..n] \cdot T[1..i_{2}-1]$, with $i_1 \neq i_2$, such that $T^{1}[1..|W|]=T^{2}[1..|W|]=W$ and $T^{1}[|W|+1] \neq T^{2}[|W|+1]$. A repeat $W$ is \emph{left-maximal} if $|W|<n$ and there are two rotations $T^1 = T[i_{1}..n] \cdot T[1..i_{1}-1]$ and $T^2 = T[i_{2}..n] \cdot T[1..i_{2}-1]$, with $i_1 \neq i_2$, such that $T^{1}[2..|W|+1]=T^{2}[2..|W|+1]=W$ and $T^{1}[1] \neq T^{2}[1]$. 
Intuitively, a right-maximal (respectively, left-maximal) repeat cannot be extended to the right (respectively, to the left) by a single character, without losing at least one of its occurrences in $T$. A repeat is \emph{maximal} if it is both left- and right-maximal. 
If $T \in [1..\sigma]^{n-1}\#$, repeats are substrings of $T$, and we use the terms left- (respectively, right-) maximal \emph{substring}. Given a string $W \in \mathcal{S}(T)$, we call $\mu(W)$ the number of (not necessarily proper) suffixes of $W$ that are maximal repeats of $T$, and we set $\mu_{T}=\max\{\mu(T') : T' \in \mathcal{R}(T)\}$. We say that a repeat $W$ of $T$ is \emph{strongly left-maximal} if there are at least two distinct characters $a$ and $b$ in $[1..\sigma]$ such that both $aW$ and $bW$ are right-maximal repeats of $T$. Since only a right-maximal repeat $W$ of $T$ can be strongly left-maximal, the set of strongly left-maximal repeats of $T$ is a subset of the maximal repeats of $T$. Let $W \in \mathcal{S}(T)$, and let $\lambda(W)$ be the number of (not necessarily proper) suffixes of $W$ that are strongly left-maximal repeats of $T$. We set $\lambda_{T}=\max\{\lambda(T') : T' \in \mathcal{R}(T)\}$. Note that $\lambda_T \leq \mu_T$. 
Other types of repeat will be described in \Sec{sec:stringAnalysis}.

\subsection{Suffix tree} \label{sec:suffixTree}

Let $\mathcal{T}=\{T^1,T^2,\dots,T^m\}$ be a set of strings on alphabet $[1..\sigma]$. The \emph{trie of} $\mathcal{T}$ is the tree $G=(V,E,\ell)$, with set of nodes $V$, set of edges $E$, and labeling function $\ell$, defined as follows: (1) every edge $e \in E$ is labeled by exactly one character $\ell(e) \in [1..\sigma]$; (2) the edges that connect a node to its children have distinct labels; (3) the children of a node are sorted lexicographically according to the labels of the corresponding edges; (4) there is a one-to-one correspondence between $V$ and the set of distinct prefixes of strings in $\mathcal{T}$. Note that, if no string in $\mathcal{T}$ is a prefix of another string in $\mathcal{T}$, there is a one-to-one correspondence between the elements of $\mathcal{T}$ and the leaves of the trie of $\mathcal{T}$.

Given a trie, we call \emph{unary path} a maximal sequence $v_1,v_2,\dots,v_k$ such that $v_i \in V$ and $v_i$ has exactly one child, for all $i \in [1..k]$. By \emph{collapsing a unary path} we mean transforming $G=(V,E,\ell)$ into a tree $G'=(V \setminus \{v_1,\dots,v_k\}, (E \setminus \{(v_0,v_1),(v_1,v_2),\dots,(v_k,v_{k+1})\}) \cup \{(v_0,v_{k+1})\}, \ell')$, where $v_0$ is the parent of $v_1$ in $G$, $v_{k+1}$ is the only child of $v_k$ in $G$, $\ell'(e)=\ell(e)$ for all $e \in E \cap E'$, and $\ell'((v_0,v_{k+1}))$ is the concatenation $\ell(v_0,v_1) \cdot \ell(v_1,v_2) \cdot \cdots \cdot \ell(v_k,v_{k+1})$. Note that $\ell'$ labels the edges of $G'$ with strings rather than with single characters. Given a trie, we call \emph{compact trie} the labeled tree obtained by collapsing all unary paths in the trie. Every node of a compact trie has either zero or at least two children.

\begin{definition}[\cite{Wei73}] \label{def:suffixtree}
Let $T \in [1..\sigma]^{n}$ be a string such that $|\mathcal{R}(T)|=n$. The suffix tree $\ST_T=(V,E,\ell)$ of $T$ is the compact trie of $\mathcal{R}(T)$.
\end{definition}

Note that $\ST_T$ is not defined if some rotations of $T$ are lexicographically identical, and that there is a one-to-one correspondence between the leaves of the suffix tree of $T$ and the elements of $\mathcal{R}(T)$. Since the suffix tree of $T$ has precisely $n$ leaves, and since every internal node is branching, there are at most $n-1$ internal nodes. We denote by $\SP{v}$, $\EP{v}$, and $\INTERVAL{v}$ the left-most leaf, the right-most leaf, and the set of all leaves in the subtree of an internal node $v$, respectively. We denote by $\ell(e)$ the label of an edge $e \in E$, and by $\ell(v)$ the string $\ell(r,v_1) \cdot \ell(v_1,v_2)\cdot  \cdots \cdot \ell(v_{k-1},v)$, where $r \in V$ is the root of the tree, and $r,v_1,v_2,\dots,v_{k-1},v$ is the path of $v \in V$ in the tree. We say that node $v$ has \emph{string depth} $|\ell(v)|$. We call $w$ the \emph{proper locus} of string $W$ if the search for $W$ starting from the root of $\ST_T$ ends at an edge $(v,w) \in E$. Note that there is a one-to-one correspondence between the set of internal nodes of $\ST_T$ and the set of right-maximal repeats of $T$. Moreover, the set of all left-maximal repeats of $T$ enjoys the \emph{prefix closure} property, in the sense that if a repeat is left-maximal, so is any of its prefixes. It follows that the maximal repeats of $T$ form an induced subgraph of the suffix tree of $T$, rooted at $r$.

Given strings $T^1,T^2,\dots,T^m$ with $T^i \in [1..\sigma]^{n_i}$ for $i \in [1..m]$, assume that $|\mathcal{R}(T^i)|=n_i$ for all $i \in [1..m]$, and that $\mathcal{R}(T^i) \cap \mathcal{R}(T^j) = \emptyset$ for all $i \neq j$ in $[1..m]$. We call \emph{generalized suffix tree} the compact trie of $\mathcal{R}(T^1) \cup \mathcal{R}(T^2) \cup \dots \cup \mathcal{R}(T^m)$. Note that, if string $W$ labels an internal node of the suffix tree of a string $T^i$, then it also labels an internal node of the generalized suffix tree. However, there could be an internal node $v$ in the generalized suffix tree $G=(V,E,\ell)$ such that $\ell(v)$ does not label an internal node in any $T^i$. This means that: (1) if $\ell(v) \in \mathcal{S}(T^i)$, then it is always followed by the same character $a_i$ in every rotation of $T^i$; (2) there are at least two strings $T^i$ and $T^j$, with $i \neq j$, such that $a_i \neq a_j$. A node $v$ in the generalized suffix tree could be such that all leaves in the subtree rooted at $v$ are rotations of the same string $T^i$: we call such a node \emph{pure}, and we call it \emph{impure} otherwise.

Let the label $\ell(v)$ of an internal node $v$ of $\ST_{T}=(V,E,\ell)$ be $aW$, with $a \in \Sigma$ and $W \in \Sigma^*$. Since $W$ occurs at all positions where $aW$ occurs, there must be a node $w \in V$ with $\ell(w)=W$, otherwise $v$ would not be a node of the suffix tree. We say that there is a \emph{suffix link from $v$ to $w$ labelled by $a$}, and we write $\mathtt{suffixLink}(v)=w$. More generally, we say that the set of labels of internal nodes of $\ST_T$ enjoys the \emph{suffix closure} property, in the sense that if a string $W$ belongs to this set, so does every one of its suffixes. If $T \in [1..\sigma]^{n-1}\#$, we define $\mathtt{suffixLink}(v)$ for leaves $v$ of $\ST_T$ as well: the suffix link from a leaf leads either to another leaf, or to the root of $\ST_T$. The graph that consists of the set of internal nodes of $\ST_T$ and of the set of suffix links, is a trie rooted at the same root node as $\ST_T$: we call such trie the \emph{suffix-link tree} $\SLT_T$ of $T$. Note that the suffix-link tree might contain unary paths, and that traversing the suffix-link tree allows one to enumerate all nodes of the suffix tree. Note also that extending to the left a repeat that is not right-maximal does not lead to a right-maximal repeat. We exploit this property in Section \ref{sec:enumeration} to enumerate all nodes of the suffix tree in small space, storing neither the suffix tree nor the suffix-link tree explicitly. Note that every leaf of the suffix-link tree has more than one Weiner link, or its label has length $n-1$. Thus, the set of all maximal repeats of $T$ coincides with the set of all the internal nodes of the suffix-link tree with at least two (implicit or explicit) Weiner links, and with a subset of all the leaves of the suffix-link tree.

Inverting the direction of all suffix links yields the so-called \emph{explicit Weiner links}. 
Given a node $v \in V$ and a character $a \in \Sigma$, it might happen that string $a\ell(v) \in \mathcal{S}(T)$, but that it does not label any internal node of $\ST_T$: we call all such extensions of internal nodes \emph{implicit Weiner links}. An internal node might have multiple outgoing Weiner links (possibly both explicit and implicit), and all such Weiner links have distinct labels. The constructions described in this paper rest on the fact that the total number of explicit and implicit Weiner links is small:

\begin{observation} \label{obs:suffixtree}
Let $T \in [1..\sigma]^{n}$ be a string such that $|\mathcal{R}(T)|=n$. The number of suffix links, explicit Weiner links, and implicit Weiner links in the suffix tree of $T$ are upper bounded by $n-2$, $n-2$, and $3n-3$, respectively.
\end{observation}
\begin{proof}
Each of the at most $n-2$ internal nodes of the suffix tree (other than the root) has a suffix link. Each explicit Weiner link is the inverse of a suffix link, so their total number is also at most $n-2$.

Consider an internal node $v$ with only one implicit Weiner link $e=(\ell(v),a\ell(v))$. The number of such nodes, and thus the number of such implicit Weiner links, is bounded by $n-1$. Call these the implicit Weiner links of type I, and the remaining the implicit Weiner links of type II. Consider an internal node $v$ with two or more implicit Weiner links, and let $\Sigma_v$ be the set of labels of all Weiner links from $v$. Since $|\Sigma_v|>1$, there is an internal node $w$ in the suffix tree $\ST_{\REV{T}}$ of the \emph{reverse} $\REV{T}$ of $T$, labeled by the reverse $\REV{\ell(v)}$ of $\ell(v)$: every $c \in \Sigma_v$ can be mapped to a distinct edge of $\ST_{\REV{T}}$ connecting $w$ to one of its children. This is an injective mapping from all type II implicit Weiner links to the at most $2n-2$ edges of the suffix tree of $\REV{T}$. The sum of type I and type II Weiner links, i.e. the number of all implicit Weiner links, is hence bounded by $3n-3$.
\end{proof}

Slightly more involved arguments push the upper bound on the number of implicit Weiner links down to $n$.

\subsection{Rank and select\label{sect:rankandselect}} 

Given a string $S \in [1..\sigma]^n$, we denote by $\mathtt{rank}_c(S,i)$ the number of occurrences of character $c \in [1..\sigma]$ in $S[1..i]$, and we denote by $\mathtt{select}_c(S,j)$ the position $i$ of the $j$-th occurrence of $c$ in $S$, i.e. $j=\mathtt{rank}_c(S,\mathtt{select}_c(S,j))$. We use $\mathtt{partialRank}(S,i)$ as a shorthand for $\mathtt{rank}_{S[i]}(S,i)$. Data structures to support such operations efficiently will be described in the sequel. Here we just recall that it is possible to represent a bitvector of length $n$ using $n+o(n)$ bits of space, such that rank and select queries can be supported in constant time (see e.g. \cite{Cla96,Mun96}). This representation can be built in $O(n)$ time and in $o(n)$ bits of working space. Rank and select data structures can be used to implement a \emph{representation} of a string $S$ that supports operation $\mathtt{access}(S,i)=S[i]$ without storing $S$ itself.

\subsection{String indexes}\label{sec:bwt}

Sorting the set of rotations of a string yields an index that can be used for supporting pattern matching by binary search:
 
\begin{definition}[\cite{MM93}] \label{def:sa}
Let $T \in [1..\sigma]^{n}$ be a string such that $|\mathcal{R}(T)|=n$. The \emph{suffix array} $\SA_{T}[1..n]$ of $T$ is the permutation of $[1..n]$ such that $\SA_{T}[i]=j$ iff rotation $T[j..n] \cdot T[1..j-1]$ has rank $i$ in the list of all rotations of $T$ taken in lexicographic order.
\end{definition}

Note that $\SA_T$ is not defined if some rotations of $T$ are lexicographically identical.
We denote by $\INTERVAL{W} = [\SP{W} .. \EP{W}]$ the maximal interval of $\SA_{T}$ whose rotations are prefixed by $W$. Note that $\INTERVAL{W}$, $\SP{W}$ and $\EP{W}$ are in one-to-one correspondence with $\INTERVAL{v}$, $\SP{v}$ and $\EP{v}$ of a node $v$ of the suffix tree of $T$ such that $\ell(v)=W$. We will often use such notions interchangeably. 

The \emph{longest common prefix array} stores the length of the longest common prefix between every two consecutive rotations in the suffix array:

\begin{definition}[\cite{MM93}]
Let $T \in [1..\sigma]^n$ be a string such that $|\mathcal{R}(T)|=n$, and let $p(i,j)$ be the function that returns the longest common prefix between the rotation that starts at position $\SA_{T}[i]$ in $T$ and the rotation that starts at position $\SA_{T}[j]$ in $T$. The \emph{longest common prefix array} of $T$, denoted by $\LCP_{T}[1..n]$, is defined as follows: $\LCP_{T}[1]=0$ and $\LCP_{T}[i] = p(i,i-1)$ for all $i \in [2..n]$. The \emph{permuted longest common prefix array} of $T$, denoted by $\PLCP_{T}[1..n]$, is the permutation of $\LCP_{T}$ in string order, i.e. $\PLCP_{T}[\SA_{T}[i]]=\LCP_{T}[i]$ for all $i \in [1..n]$.
\end{definition}

The main tool that we use in this paper for obtaining space-efficient index structures is a permutation of $T$ induced by its suffix array:

\begin{definition}[\cite{BW94}] \label{def:bwt}
Let $T \in [1..\sigma]^{n}$ be a string such that $|\mathcal{R}(T)|=n$. The \emph{Burrows-Wheeler transform} of $T$, denoted by $\BWT_{T}$, is the permutation $L[1 \ltdots n]$ of $T$ such that $L[i]=T[\SA_{T}[i]-1 \; (\mbox{mod}_{1} \; n)]$ for all $i \in [1..n]$.
\end{definition}


Like $\SA_T$, $\BWT_T$ cannot be uniquely defined if some rotations of $T$ are lexicographically identical. Given two strings $S$ and $T$ such that $\mathcal{R}(S) \cap \mathcal{R}(T) = \emptyset$, we say that the BWT of $\mathcal{R}(S) \cup \mathcal{R}(T)$ is the string obtained by sorting $\mathcal{R}(S) \cup \mathcal{R}(T)$ lexicographically, and by printing the character that precedes the starting position of each rotation. Note that either $\mathcal{R}(S) \cap \mathcal{R}(T) = \emptyset$, or $\mathcal{R}(S)=\mathcal{R}(T)$.

A key feature of the BWT is that it is \emph{reversible}: given $\BWT_T=L$, one can reconstruct the unique $T$ of which $L$ is the Burrows-Wheeler transform. Indeed, let $V$ and $W$ be two rotations of $T$ such that $V$ is lexicographically smaller than $W$, and assume that both $V$ and $W$ are preceded by character $a$ in $T$. It follows that rotation $aV$ is lexicographically smaller than rotation $aW$, thus there is a bijection between rotations \emph{preceded by $a$} and rotations that \emph{start with $a$} that preserves the relative order among such rotations. Consider thus the rotation that starts at position $i$ in $T$, and assume that it corresponds to position $p_i$ in $\SA_T$ (i.e $\SA_T[p_i]=i$). If $L[p_i]=a$ is the $k$-th occurrence of $a$ in $L$, then the rotation that starts at position $i-1$ in $T$ must be the $k$-th rotation that starts with $a$ in $\SA_T$, and its position $p_{i-1}$ in $\SA_T$ must belong to the compact interval $\INTERVAL{a}$ that contains all rotations that start with $a$. For historical reasons, the function that projects the position $p_i$ in $\SA_T$ of a rotation that starts at position $i$, to the position $p_{i-1}$ in $\SA_T$ of the rotation that starts at position $i-1 \; (\mbox{mod}_1 \; n)$, is called $\LF$ (or \emph{last-to-first}) \emph{mapping} \cite{FM00,FM05}, and it is defined as $\LF(i)=j$, where $\SA[j]=\SA[i]-1 \; (\mbox{mod}_1\; n)$. Note that reconstructing $T$ from its BWT requires to know the starting position in $T$ of its lexicographically smallest rotation.

Let again $L$ be the Burrows-Wheeler transform of a string $T \in [1..\sigma]^n$, and assume that we have an array $C[1 \ltdots \sigma]$ that stores in $C[c]$ the number of occurrences in $T$ of all characters strictly smaller than $c$, that is the sum of the frequency of all characters in $[1..c-1]$. Note that $C[1]=0$, and that $C[c]+1$ is the position in $\SA_T$ of the first rotation that starts with character $c$. It follows that $\LF(i)=C[L[i]]+\mathtt{rank}_{L[i]}(L,i)$. 

Function $\LF$ can be extended to a \emph{backward search} algorithm which counts the number of occurrences in $T$ of a string $W$, in $O(|W|)$ steps, considering iteratively suffixes $W[i..|W|]$ with $i$ that goes from $|W|$ to one \cite{FM00,FM05}. Given the interval $[i_1..j_1]$ that corresponds to a string $V$ and a character $c$, the interval $[i_2..j_2]$ that corresponds to string $cV$ can be computed as $i_2=\mathtt{rank}_c(i_1-1)+C[c]+1$ and $j_2=\mathtt{rank}_c(j_1)+C[c]$. If $i_2>j_2$, then $cV \notin \mathcal{S}(T)$. Note that, if $W$ is a right-maximal repeat of $T$, a step of backward search corresponds to taking an explicit or implicit Weiner link in $\ST_T$. The time for computing a backward step is dominated by the time needed to perform a $\mathtt{rank}$ query, which is typically $O(\log{\log{\sigma}})$ \cite{GMR06} or $O(\log{\sigma})$ \cite{GGV03}.

The inverse of function $\LF$ is called $\psi$ for historical reasons \cite{GV00,Sad00}, and it is defined as follows. Assume that position $i$ in $\SA_{T}$ 
corresponds to rotation $aW$ with $a \in [1..\sigma]$: since $a$ satisfies $C[a] < i \leq C[a+1]$, it can be computed from $i$ by performing $\mathtt{select}_{0}(C',i)-i+1$ on a bitvector $C'$ that represents $C$ with $\sigma-1$ ones and $n$ zeros, and that is built as follows: we append $C[i+1]-C[i]$ zeros followed by a one for all $i \in [1..\sigma-1]$, and we append $n-C[\sigma]$ zeros at the end. Function $\psi(i)$ returns the lexicographic rank of rotation $W$, given the lexicographic rank $i$ of rotation $aW$, as follows: $\psi(i)=\mathtt{select}_{a}(\BWT_{T},i-C[a])$. 

Combining the BWT and the $C$ array gives rise to the following index, which is known as \emph{FM-index} in the literature \cite{FM00,FM05}:

\begin{definition} \label{def:BWTindex}
Given a string $T \in [1 \ltdots \sigma]^n$, a \emph{BWT index} on $T$ is a data structure that consists of: 
\begin{itemize}
\item $\BWT_{T\#}$, with support for rank (and select) queries;
\item the integer array $C[0\ltdots\sigma]$, that stores in $C[c]$ the number of occurrences in $T\#$ of all characters strictly smaller than $c$.
\end{itemize}
\end{definition}

The following lemma derives immediately from function $\LF$:

\begin{lemma} \label{lemma:bwtinvert}
Given the BWT index of a string $T \in [1..\sigma]^{n-1}\#$, there is an algorithm that outputs the sequence $\SA_{T}^{-1}[n],\SA_{T}^{-1}[n-1],\ldots,\SA_{T}^{-1}[1]$, in $O(t)$ time per value in the output, in $O(nt)$ total time, and in $O(\log{n})$ bits of working space, where $t$ is the time for performing function $\LF$.
\end{lemma}


So far we have only described how to support \emph{counting} queries, and we are still not able to \emph{locate} the starting positions of a pattern $P$ in string $T$. One way of doing this is to \emph{sample} suffix array values, and to extract the missing values using the $\LF$ mapping. Adjusting the sampling rate $r$ gives different space/time tradeoffs. Specifically, we sample all the values of $\SA_{T\#}[i]$ that satisfy $\SA_{T\#}[i] = 1+rk$ for $0 \leq k < n/r$, and we store such samples consecutively, in the same order as in $\SA_{T\#}$, in array $\mathtt{samples}[1 \ltdots \lceil n/r \rceil]$. Note that this is equivalent to sampling every $r$ positions \emph{in string order}. We also mark in a bitvector $B[1 \ltdots n]$ the positions of the suffix array that have been sampled, that is we set $B[i]=1$ if $\SA_{T\#}[i]=1+rk$, and we set $B[i]=0$ otherwise. Combined with the $\LF$ mapping, this allows one to compute $\SA_{T\#}[i]$ in $O(rt)$ time, where $t$ is the time required for function $\LF$. One can set $r=\log^{1+\epsilon} n/\log\sigma$ for any given $\epsilon>0$ to have the samples fit in $(n/r)\log n=n\log\sigma/\log^\epsilon n = o(n \log \sigma)$ bits, which is asymptotically the same as the space required for supporting counting queries. This setting implies that the extraction of $\SA_{T\#}[i]$ takes $O(\log^{1+\epsilon} n t/\log\sigma)$ time. The resulting collection of data structures is called \emph{succinct suffix array} (see e.g. \cite{NM07}).

Succinct suffix arrays can be further extended into \emph{self-indexes}. 
A self-index is a succinct representation of a string $T$ that, in addition to supporting count and locate queries on arbitrary strings provided in input, allows one to access any substring of $T$ by specifying its starting and ending position. In other words, a self-index for $T$ completely replaces the original string $T$, which can be discarded. Recall that we can reconstruct the whole string $T\#$ from $\BWT_{T\#}$ by applying function $\LF$ iteratively. To reconstruct arbitrary substrings efficiently, it suffices to store, for every sampled position $1+ri$ in string $T\#$, the position of suffix $T[1+ri \ltdots n]$ in $\SA_{T\#}$: specifically, we use an additional array $\mathtt{pos2rank}[1 \ltdots \lceil n/r \rceil]$ such that $\mathtt{pos2rank}[i]=j$ if $\SA_{T\#}[j]=1+ri$ \cite{FM05}. Note that $\mathtt{pos2rank}$ can be seen itself as a sampling of the \emph{inverse suffix array} at positions $1+ri$, and that it takes the same amount of space as array $\mathtt{samples}$. Given an interval $[e \ltdots f]$ in string $T\#$, we can use $\mathtt{pos2rank}[k]$ to go to the position $i$ of suffix $T[1+rk \ltdots n]$ in $\SA_{T\#}$, where $k=\lceil (f-1)/r \rceil$ and $1+rk$ is the smallest sampled position greater than or equal to $f$ in $T\#$. We can then apply $\LF$ mapping $1+rk-e$ times starting from $i$: the result is the whole substring $T[e \ltdots 1+rk]$ printed from right to left, thus we can return its proper prefix $T[e \ltdots f]$. The running time of this procedure is $O((f-e+r)t)$.

Making a succinct suffix array a self-index does not increase its asymptotic space complexity. We can thus define the succinct suffix array as follows:

\begin{definition}\label{def:succinctSuffixArray}
Given a string $T \in [1 \ltdots \sigma]^n$, the \emph{succinct suffix array} of $T$ is a data structure that takes $n\log\sigma(1+o(1))+O((n/r)\log{n})$ bits of space, where $r$ is the sampling rate, and that supports the following queries: 
\begin{itemize}
\item $\mathtt{count}(P)$: returns the number of occurrences of string $P \in [1 \ltdots \sigma]^m$ in $T$.
\item $\mathtt{locate}(i)$: returns $\SA_{T\#}[i]$. 
\item $\mathtt{substring}(e,f)$: returns $T[e \ltdots f]$. 
\end{itemize}
\end{definition}

The following result is an immediate consequence of Lemma \ref{lemma:bwtinvert}:

\begin{lemma}\label{lemma:buildingSuccinctSA}
The succinct suffix array of a string $T \in [1 \ltdots \sigma]^n$ can be built from the BWT index of $T$ in $O(nt)$ time and in $O(\log{n})$ bits of working space, where $t$ is the time for performing function $\LF$.
\end{lemma}

In Section \ref{sec:bwtIndexesConstruction} we define additional string indexes used in this paper, like the \emph{compressed suffix array}, the \emph{compressed suffix tree}, and the \emph{bidirectional BWT index}.

\section{Building blocks and techniques}\label{sect:datastructures}


\subsection{Static memory allocation} \label{sec:staticAllocation}

Let $\mathcal{A}$ be an algorithm that builds a set of arrays by iteratively appending new elements to their end. In all cases described in this paper, the final size of all growing arrays built by $\mathcal{A}$ can be precomputed by running a slightly modified version $\mathcal{A}'$ of $\mathcal{A}$ that has the same time and space complexity as $\mathcal{A}$. Thus, we always restructure $\mathcal{A}$ as follows: first, we run $\mathcal{A}'$ to precompute the final size of all growing arrays built by $\mathcal{A}$; then, we allocate a single, contiguous region of memory that is large enough to contain all the arrays built by $\mathcal{A}$, and we compute the starting position of each array inside the region; finally, we run $\mathcal{A}$ using such positions. This strategy avoids memory fragmentation in practice, and in some cases, for example in Section \ref{sec:topology}, it even allows us to achieve better space bounds. See Figure \ref{fig:architecture} for a list of all algorithms in the paper that use this technique.

\subsection{Batched locate queries}

In this paper we will repeatedly need to resolve \emph{a batch} of queries $\mathtt{locate}(i)=\SA_{T\#}[i]$ issued on a set of distinct values of $i$ in $[1..n]$, where $T \in [1..\sigma]^{n-1}$. The following lemma describes how to answer such queries using just the BWT of $T$ and a data structure that supports function $\LF$:

\begin{lemma} \label{lemma:batch_extract}
Let $T \in [1..\sigma]^{n-1}$ be a string. Given the BWT of $T\#$, a data structure that supports function $\LF$, and a list $\mathtt{pairs}[1 \ltdots \mathtt{occ}]$ of pairs $(i_k,p_k)$, where $i_k \in [1 \ltdots n]$ is a position in $\SA_{T\#}$ and $p_k$ is an integer for all $k \in [1 \ltdots \mathtt{occ}]$, we can transform every pair $(i_k,p_k)$ in $\mathtt{pairs}$ into the corresponding pair $(\SA_{T\#}[i_k],p_k)$, possibly altering the order of list $\mathtt{pairs}$, in $O(nt + \mathtt{occ})$ time and in $O(\mathtt{occ} \cdot \log{n})$ bits of working space, where $t$ is the time taken to perform function $\LF$.
\end{lemma}
\begin{proof}
Assume that we could use a bitvector $\mathtt{marked}[1\ltdots n]$ such that $\mathtt{marked}[i_k]=1$ for all the distinct $i_k$ that appear in $\mathtt{pairs}$. Building $\mathtt{marked}$ from $\mathtt{pairs}$ takes $O(n+\mathtt{occ})$ time. Then, we invert $\BWT_{T\#}$ in $O(nt)$ time. During this process, whenever we are at a position $i$ in $\BWT_{T\#}$, we also know the corresponding position $\SA_{T\#}[i]$ in $T\#$: if $\mathtt{marked}[i]=1$, we append pair $(i,\SA_{T\#}[i])$ to a temporary array $\mathtt{translate}[1 \ltdots \mathtt{occ}]$. At the end of this process, the pairs in $\mathtt{translate}$ are in reverse string order: thus, we sort both $\mathtt{translate}$ and $\mathtt{pairs}$ in suffix array order. Finally, we perform a linear, simultaneous scan of the two sorted arrays, replacing $(i_k,p_k)$ in $\mathtt{pairs}$ with $(\SA_{T\#}[i_k],p_k)$ using the corresponding pair $(i_k,\SA_{T\#}[i_k])$ in $\mathtt{translate}$.

If $\mathtt{occ} \geq n/\log{n}$, $\mathtt{marked}$ fits in $O(\mathtt{occ} \cdot \log{n})$ bits. Otherwise, rather than storing $\mathtt{marked}[1..n]$, we use a smaller bitvector $\mathtt{marked}'[1..n/h]$ in which we set $\mathtt{marked}'[i]=1$ iff there is an $i_k \in [hi..h(i+1)-1]$. As we invert $\BWT_{T\#}$, we check whether the block $i/h$ that contains the current position $i$ in the BWT, is such that $\mathtt{marked}'[i/h]=1$. If this is the case, we binary search $i$ in $\mathtt{pairs}$. Every such binary search takes $O(\log{\mathtt{occ}})$ time, and we perform at most $h \cdot \mathtt{occ}$ binary searches in total. Setting $h=n/(\mathtt{occ} \cdot \log{n})$ makes $\mathtt{marked}'$ fit in $\mathtt{occ} \cdot \log{n}$ bits, and it makes the total time spent in binary searches $O(n/(\log{n}/\log{\mathtt{occ}})) \in O(n)$.

Now if $\mathtt{occ}\geq \sqrt{\log n}$, we sort array $\mathtt{pairs}[1 \ltdots \mathtt{occ}]$ using radix sort: specifically, we interpret each pair $(i_k,p_k)$ as a triple $(\mathtt{msb}(i_k),\mathtt{lsb}(i_k),p_k)$, where $\mathtt{msb}(x)$ is a function that returns the most significant $\lceil (\log n)/2 \rceil$ bits of $x$ and $\mathtt{lsb}(x)$ is a function that returns the least significant $\lfloor (\log n)/2 \rfloor$ bits of $x$. Since the resulting primary and secondary keys belong to the range $[1 \ltdots 2\sqrt{n}]$, sorting both $\mathtt{pairs}$ and $\mathtt{translate}$ takes $O(\sqrt{n}+\mathtt{occ})$ time and $O((\sqrt{n}+\mathtt{occ})\log n)\in O(\mathtt{occ}\log n)$ bits of working space.
If $\mathtt{occ}<\sqrt{\log n}$ we just sort array $\mathtt{pairs}[1 \ltdots \mathtt{occ}]$ using standard comparison sort, 
in time $O(\mathtt{occ}\log n)\in O(\sqrt{\log{n}}\log n)$. 
\end{proof}

\subsection{Data structures for prefix-sum queries}\label{sec:prefixSums}

A \emph{prefix-sum data structure} supports the following query on an array of numbers $A[1..n]$: given $i \in [1..n]$, return $\sum_{j=1}^{i}A[j]$. The following well-known result, which we will use extensively in what follows, derives from combining Elias-Fano coding~\cite{El74,Fa71} with bitvectors indexed to support the $\mathtt{select}$ operation in constant time:

\begin{lemma}[\cite{okanohara2007practical}]\label{lemma:prefixSums}
Given a representation of an array of integers $A[1..n]$ whose total sum is $U$, that allows one to access its entries from left to right, we can build in $O(n)$ time and in $O(\log{U})$ bits of working space a data structure that takes $n(2+\lceil \log(U/n) \rceil)+o(n)$ bits of space and that answers prefix-sum queries in constant time.
\end{lemma}

\subsection{Data structures for access, rank, and select queries}\label{sec:accessPartialrankSelect}

We conceptually split a string $S$ of length $n$ into $N=\lceil n/\sigma\rceil$ blocks of size $\sigma$ each, except possibly for the last block which might be smaller. Specifically, block number $i \in [1 .. N-1]$, denoted by $S^i$, covers substring $S[\sigma(i-1)+1 .. \sigma i]$, and the last block $S^N$ covers substring $S[\sigma(N-1)+1 .. n]$. The purpose of splitting $S$ into blocks consists in translating global operations on $S$ into local operations on a block: for example, $\mathtt{access}(i)$ can be implemented by issuing $\mathtt{access}(i-\sigma(b-1))$ on block $b = \lceil i/\sigma \rceil$. The construction we describe in this section largely overlaps with \cite{GMR06}.

We use $f(c)$ to denote the frequency of character $c$ in $S$, $f(c,b)$ to denote the frequency of character $c$ in $S^b$, and $C^{b}[c]$ as a shorthand for $\sum_{a=1}^{c-1}f(a,b)$. We encode the block structure of $S$ using bitvector $\mathtt{freq}=\mathtt{freq}_1 \mathtt{freq}_2 \cdots \mathtt{freq}_\sigma$, where bitvector $\mathtt{freq}_c[1 .. f(c)+N]$ is defined as follows:
\[
\mathtt{freq}_{c} = \mathtt{1} \mathtt{0}^{f(c,1)} \mathtt{1} \mathtt{0}^{f(c,2)} \mathtt{1} \dots \mathtt{1} \mathtt{0}^{f(c,N)}
\]
Note that $\mathtt{freq}$ takes at most $2n+\sigma-1$ bits of space: indeed, every $\mathtt{freq}_{c}$ contains exactly $N$ ones, thus the total number of ones in all bitvectors is $\sigma\lceil n/\sigma\rceil \leq n+\sigma-1$, and the total number of zeros in all bitvectors is $\sum_{c \in [1..\sigma]}f(c)=n$. Note also that a $\mathtt{rank}$ or $\mathtt{select}$ operation on a specific $\mathtt{freq}_c$ can be translated in constant time into a $\mathtt{rank}$ or $\mathtt{select}$ operation on $\mathtt{freq}$. Bitvector $\mathtt{freq}$ can be computed efficiently:

\begin{lemma}\label{lemma:freq}
Given a string $S \in [1..\sigma]^n$, vector $\mathtt{freq}$ can be built in $O(n)$ time and in $o(n)$ bits of working space.
\end{lemma}
\begin{proof}
We use the static allocation strategy described in Section \ref{sec:staticAllocation}: specifically, we first compute $f(c)$ for all $c\in[1..\sigma]$ by scanning $S$ and incrementing corresponding counters. Then, we compute the size of each bitvector $\mathtt{freq}_{c}$ and we allocate a contiguous region of memory for $\mathtt{freq}$. Storing all $f(c)$ counters takes $\sigma\log{n} \leq (\sqrt{n}/\log{n})\log{n} = o(n)$ bits of space. Finally, we scan $S$ once again: whenever we see the beginning of a new block, we append a one to the end of every $\mathtt{freq}_{c}$, and whenever we see an occurrence of character $c$, we append a zero to the end of $\mathtt{freq}_{c}$. The total time taken by this process is $O(n)$, and the pointers to the current end of each $\mathtt{freq}_{c}$ in $\mathtt{freq}$ take $o(n)$ bits of space overall.
\end{proof}

Vector $\mathtt{freq}_{c}$, indexed to support $\mathtt{rank}$ or $\mathtt{select}$ operations in constant time, is all we need to translate in constant time a $\mathtt{rank}$ or $\mathtt{select}$ operation on $S$ into a corresponding operation on a block of $S$: thus, we focus just on supporting $\mathtt{rank}$ and $\mathtt{select}$ operations \emph{inside a block of $S$} in what follows.

For this purpose, let $X^b$ be the string $1^{f(1,b)} 2^{f(2,b)} \cdots \sigma^{f(\sigma,b)}$. $S^b$ can be seen as a permutation of $X^b$: let $\pi_b: [1..\sigma] \mapsto [1..\sigma]$ be the function that maps a position in $S^b$ onto a position in $X^b$, and let $\pi^{-1}_b: [1..\sigma] \mapsto [1..\sigma]$ be the function that maps a position in $X^b$ to a position in $S^b$. A possible choice for such permutation functions is:
\begin{eqnarray}
\pi_{b}(i) & = & C^{b}[S^{b}[i]]+\mathtt{rank}_{S^b}(S^{b}[i],i) \label{eq:pi} \\
\pi^{-1}_{b}(i) & = & \mathtt{select}_{S^b}(i-C^{b}[c],c) \label{eq:piInverse}
\end{eqnarray}
where $c=X^{b}[i]$ is the only character that satisfies $C^{b}[c] < i \leq C^{b}[c+1]$. We store explicitly just one of $\pi_b$ and $\pi^{-1}_b$, so that random access to any element of the stored permutation takes constant time, and we represent the other permutation \emph{implicitly}, as described in the following lemma:

\begin{lemma}[\cite{munro2003succinct}]\label{lemma:permutation}
Given a permutation $\pi[1..n]$ of sequence $1,2,\dots,n$, there is a data structure that takes $(n/k)\log n + n + o(n)$ bits of space in addition to $\pi$ itself, and that supports query $\pi^{-1}[i]$ for any $i \in [1..n]$ in $O(k)$ time, for any integer $k \geq 1$. This data structure can be built in $O(n)$ time and in $o(n)$ bits of working space. The query and the construction algorithms assume constant time access to any element 
$\pi[i]$. 
\end{lemma}
\begin{proof}
A permutation $\pi[1..n]$ of sequence $1,2,\dots,n$ can be seen as a collection of \emph{cycles}, where the number of such cycles ranges between one and $n$. Indeed, consider the following iterated version of the permutation operator:
$$
\pi^{t}[i]=
\left\{
\begin{array}{ll}
i & \mbox{if } t=0 \\
\pi[\pi^{t-1}[i]] & \mbox{if } t>0
\end{array}
\right.
$$
We say that a position $i$ in $\pi$ belongs to a cycle of length $t$, where $t$ is the smallest positive integer such that $\pi^{t}[i]=i$. Note that $\pi$ can be decomposed into cycles in linear time and using $n$ bits of working space, by iterating operator $\pi$ from position one, by marking in a bitvector all the positions that have been touched by such iteration, and by repeating the process from the next position in $\pi$ that has not been marked. 

If a cycle contains a number of arcs $t$ greater than a predefined threshold $k$, it can be subdivided into $\lceil t/k \rceil$ paths containing at most $k$ arcs each. We store in a dictionary the first vertex of each path, and we associate with it a pointer to the first vertex of the path that \emph{precedes} it. 
That is, given a cycle $x,\pi[x],\pi^2[x],\ldots,\pi^{t-1}[x],x$, the dictionary stores the set of pairs:
$$
\left\{ (x,\pi^{k (\lceil t/k \rceil-1)}[x]) \right\} \cup \left\{ (\pi^{ik}[x],\pi^{(i-1)k}[x]) : i \in [1 .. (\lceil t/k \rceil-1)] \right\} \; .
$$ 

Then, we can determine $\pi^{-1}[i]$ for any value $i$ in $O(k)$ time, by successively computing $i,\pi[i],\pi^{2}[i], \dots, \pi^{k}[i]$ and by querying the dictionary for every vertex in such sequence. As soon as the query is successful for some $\pi^{j}[i]$ with $j \in [0..k]$, we get $\pi^{j-k}[i]$ from the dictionary and we compute the sequence $\pi^{j-k}[i],\pi^{j-k+1}[i], \pi^{j-k+2}[i], \dots, \pi^{-1}[i],i$, returning $\pi^{-1}[i]$. The dictionary can be implemented using a table and a bitvector of size $n$ with $\mathtt{rank}$ support, which marks the first element of each path of length $k$ of each cycle. 
\end{proof}

Combining Lemma \ref{lemma:freq} and Lemma \ref{lemma:permutation} with Equations \ref{eq:pi} and \ref{eq:piInverse}, we obtain the key result of this section:

\begin{lemma}\label{lemma:rank_select_access}
Given a string of length $n$ over alphabet $[1..\sigma]$, we can build the following data structures in $O(n)$ time and in $o(n)$ bits of working space:
\begin{itemize}
\item a data structure that takes at most $n\log{\sigma}+4n+o(n)$ bits of space, and that supports $\mathtt{access}$ and $\mathtt{partialRank}$ in constant time;
\item a data structure that takes $n\log{\sigma}(1+1/k) + 5n + o(n)$ bits of space for any positive integer $k$, and that supports either $\mathtt{access}$ and $\mathtt{partialRank}$ in constant time and $\mathtt{select}$ in $O(k)$ time, or $\mathtt{select}$ in constant time and $\mathtt{access}$ and $\mathtt{partialRank}$ in $O(k)$ time.
\end{itemize}
Neither of these data structures requires the original string to support $\mathtt{access}$, $\mathtt{partialRank}$ and $\mathtt{select}$.
\end{lemma}
\begin{proof}
In addition to the data structures built in Lemma \ref{lemma:freq}, we store $\pi_b$ explicitly for every $S^b$, spending overall $n\log{\sigma}$ bits of space. Note that $\pi_b$ can be computed from $S^b$ in linear time for all $b \in [1..N]$. We also store $C^b$ for every $S^b$ as a bitvector of $2\sigma$ bits that coincides with a unary encoding of $X^b$ (that is we store $10^{f(1,b)}10^{f(2,b)} \cdots 10^{f(\sigma,b)}$): given a position $i$ in block $S^b$, we can determine the character $c$ that satisfies $C^b[c] < \pi_{b}[i] \leq C^b[c+1]$ using a $\mathtt{select}$ and a $\mathtt{rank}$ query on such bitvector, thus implementing $\mathtt{access}$ to $S^b[i]$ in constant time. In turn, this allows one to implement $\mathtt{partialRank}_{S^b}(i)$ in constant time using Equation \ref{eq:pi}.

A $\mathtt{select}$ query on $C^b$, combined with the implicit representation of $\pi^{-1}_b$ described in Lemma \ref{lemma:permutation}, allows one to implement $\mathtt{select}$ on $S^b$ in $O(k)$ time, at the cost of $(\sigma/k) \log{\sigma} + \sigma + o(\sigma)$ bits of additional space per block. The complexity of $\mathtt{select}_{S^b}$ can be exchanged with that of $\mathtt{access}_{S^b}$ and $\mathtt{partialRank}_{S^b}$, by storing explicitly $\pi_{b}^{-1}$ rather than $\pi_{b}$.

Note that the individual lower-order terms $o(\sigma)$ needed to support rank and select queries on the bitvectors that encode $C^b$, and in the structures implemented by Lemma \ref{lemma:permutation}, do not necessarily add up to $o(n)$. Thus, for each of the two cases, we concatenate all the individual bitvectors, we index them for rank and/or select queries, and we simulate operations on each individual bitvector using operations on the bitvectors that result from the concatenation. 
\end{proof}


In some parts of the paper we will need an implementation of $\mathtt{rank}$ rather than of $\mathtt{partialRank}$, and to support this operation efficiently we will use predecessor queries. Given a set of sorted integers, a \emph{predecessor query} returns the index of the largest integer in the set that is smaller than or equal to a given integer provided in input. It is known that predecessor queries can be implemented efficiently, for example with the following data structure: 

\begin{lemma}[\cite{Wi83,hagerup2001deterministic}]\label{lemma:predecessor}
Given a sorted sequence of $n$ integers $x^1 < x^2 < \cdots < x^n$, where $x^i$ is encoded in $\log{U}$ bits for all $i \in [1..n]$, we can build in $O(n)$ time and in $O(n\log{U})$ bits of working space a data structure that takes $O(n\log{U})$ bits of space, and that answers predecessor queries in $O(\log{\log{U}})$ time. This data structure does not require the original sequence of integers to answer queries.
\end{lemma}

The original predecessor data structure described in~\cite{Wi83} (called \emph{$y$-fast trie}) 
has an expected linear time construction algorithm. Construction time is randomized, since the data structure uses a hash table. To obtain deterministic linear construction time, one can replace the hash table with a deterministic dictionary~\cite{hagerup2001deterministic}. 

Implementing rank queries amounts to plugging Lemma \ref{lemma:predecessor} into the block partitioning scheme of Lemma \ref{lemma:rank_select_access}:

\begin{lemma}\label{lemma:rank_select_access2}
Given a string of length $n$ over alphabet $[1..\sigma]$ and an integer $c>1$, we can build a data structure that takes $n\log{\sigma}(1+1/k) + 6n + O(n/\log^{c-1}{\sigma}) + o(n)$ bits of space for any positive integer $k$, and that supports:
\begin{itemize}
\item either $\mathtt{access}$ and $\mathtt{partialRank}$ in constant time, $\mathtt{select}$ in $O(k)$ time, and $\mathtt{rank}$ in $O(kc\log{\log{\sigma}})$ time;
\item or $\mathtt{access}$ and $\mathtt{partialRank}$ in $O(k)$ time, $\mathtt{select}$ in constant time, and $\mathtt{rank}$ in $O(c\log{\log{\sigma}})$ time.
\end{itemize}
This data structure can be built in $O(n)$ time and in $o(n)$ bits of working space, and it does not require the original string to support $\mathtt{access}$, $\mathtt{rank}$ and $\mathtt{select}$.
\end{lemma}
\begin{proof}
As described in Lemma \ref{lemma:rank_select_access}, we divide the string $T$ into blocks of size $\sigma$ and we build bitvectors $\mathtt{freq}_a$ for every $a \in [1..\sigma]$. We support $\mathtt{rank}_{a}(i)$ as follows. Let $b$ be the block that contains position $i$, where blocks are indexed from zero. First, we determine the number of zeros in $\mathtt{freq}_a$ that precede the $b$-th one, by computing $\mathtt{select}_{\mathtt{freq}_a}(b,1)-b$. Then, if character $a$ occurs at most $\log^{c}{\sigma}$ times inside block $b$, we binary-search the list of zeros in block $b$ of $\mathtt{freq}_a$, using at each step a select query to convert the position of a zero inside block $b$ of $\mathtt{freq}_a$ into an occurrence of character $a$ in string $T$. This process takes $O(\tau c \log{\log{\sigma}})$ time, where $\tau$ is the time to perform a select query on $T$.

If character $a$ occurs more than $\log^{c}{\sigma}$ times inside block $b$ of $\mathtt{freq}_a$, we use a sampling strategy similar to the one described in \cite{Wi83}. Specifically, we sample the relative positions at which $a$ occurs inside block $b$, every $\log^{c}{\sigma}$ occurrences of a zero in $\mathtt{freq}_a$, and we encode such positions in the data structure described in Lemma \ref{lemma:predecessor}. 
Let us call the sampled positions of a block \emph{red positions}, and all other positions \emph{blue positions}. Since positions are relative to a block, the size of the universe is $\sigma$, thus the data structure of every block takes $O(m\log{\sigma})$ bits of space and it answers queries in time $O(\log{\log{\sigma}})$, where $m$ is the number of red positions of the block. We use the data structure of Lemma \ref{lemma:predecessor} to find the index $j$ of the red position of $a$ that immediately precedes position $i$ inside block $b$: this takes $O(\log{\log{\sigma}})$ time. Since we sampled red positions every $\log^{c}{\sigma}$ occurrences of $a$ in block $b$, we know that there are exactly $(j+1)\log^{c}{\sigma}-1$ zeros inside block $b$ before the $j$-th red position. Finally, we find the blue position that immediately precedes position $i$ inside block $b$ by binary-searching the set of $\log^{c}{\sigma}-1$ blue positions between two consecutive red positions, as described above, in time $O(\tau c \log{\log{\sigma}})$. 

With this strategy we build $O(n/\log^{c}{\sigma})$ data structures of Lemma \ref{lemma:predecessor}, containing in total $O(n/\log^{c}{\sigma})$ elements, thus the total space taken by all such data structures is $O(n/\log^{c-1}{\sigma})$ bits. Note also that all such data structures can be built using just $O(\sigma/\log^{c-1}{\sigma})$ bits of working space. For every character $a \in [1..\sigma]$, we store all data structures consecutively in memory, and we encode their starting positions in the prefix-sum data structure described in Lemma \ref{lemma:prefixSums}. All such prefix-sum data structures take overall $O(n\log{\log{\sigma}}/\log^{c}{\sigma})$ bits of space, and they can be built in $O(\log{n})$ bits of working space. We use also a bitvector $\mathtt{which}_a$ of size $\lceil n/\sigma \rceil$ to mark the blocks of $\mathtt{freq}_a$ for which we built a data structure of Lemma \ref{lemma:predecessor}. To locate the starting position of the data structure of a given block and character $a$, we use a rank query on $\mathtt{which}_a$ and we query the prefix-sum data structure in constant time. The bitvectors for all characters take overall $n+o(n)$ bits of space.
\end{proof}

In the space complexity of Lemma \ref{lemma:rank_select_access2}, we can achieve $5n$ rather than $6n$ by replacing the plain bitvectors $\mathtt{which}_a$ with the compressed bitvector representation described in \cite{Pa08}, which supports constant-time rank queries using $(c\log{\log{\sigma}}/\log^{c}{\sigma})n + O(n/\mbox{polylog}(n))$ bits. Lemma \ref{lemma:rank_select_access2} can be further improved by replacing binary searches with queries to the following data structure:

\begin{lemma}[\cite{grossi2009more}]\label{lemma:sbTree}
Given a sorted sequence of $n$ integers $x^1 < x^2 < \cdots < x^n$, where $x^i$ is encoded in $\log{U}$ bits for all $i \in [1..n]$, and given a constant $\epsilon<1$ and a lookup table of size $O(U^{\epsilon})$ bits, we can build in $O(n)$ time and in $O(n\log{U})$ bits of working space, a data structure that takes $O(n\log{\log{U}})$ bits of space, and that answers predecessor queries in $O(t/\epsilon)$ time, where $t$ is the time to access an element of the sorted sequence of integers. The lookup table can be built in polynomial time on its size.
\end{lemma}

\begin{lemma}\label{lemma:rank_select_access3}
Given a string of length $n$ over alphabet $[1..\sigma]$, we can build a data structure that takes $n\log{\sigma}(1+1/k) + O(n\log{\log{\sigma}})$ bits of space for any positive integer $k$, and that supports:
\begin{itemize}
\item either $\mathtt{access}$ and $\mathtt{partialRank}$ in constant time, $\mathtt{select}$ in $O(k)$ time, and $\mathtt{rank}$ in $O(\log{\log{\sigma}}+k)$ time;
\item or $\mathtt{access}$ and $\mathtt{partialRank}$ in $O(k)$ time, $\mathtt{select}$ in constant time, and $\mathtt{rank}$ in $O(\log{\log{\sigma}})$ time.
\end{itemize}
This data structure can be built in $O(n)$ time and in $o(n)$ bits of working space, and it does not require the original string to support $\mathtt{access}$, $\mathtt{rank}$ and $\mathtt{select}$.
\end{lemma}
\begin{proof}
We proceed as in Lemma \ref{lemma:rank_select_access2}, but we build the data structure of Lemma \ref{lemma:sbTree} on every sequence of consecutive $\log^{c}{\sigma}-1$ blue occurrences of $a$ inside the same block $b$. Every such data structure uses $O(\log^{c}{\sigma} \cdot \log{\log{\sigma}})$ bits of space, and a $O(\tau/\epsilon)$-time predecessor query to such a data structure replaces the binary search over the blue positions performed in Lemma \ref{lemma:sbTree}, where $\tau$ is the time to perform a select query on $T$. The total time to build all the data structures of Lemma \ref{lemma:sbTree} for all blocks and for all characters is $O(n)$. All such data structures take overall $O(n\log{\log{\sigma}})$ bits of space, and they all share the same lookup table of size $o(\sigma)$ bits, which can be built in $o(\sigma)$ time by choosing $\epsilon$ small enough. We also build, in $O(n)$ time, the prefix-sum data structure of Lemma \ref{lemma:prefixSums}, which allows constant-time access to each data structure of Lemma \ref{lemma:sbTree}.
\end{proof}

\subsection{Representing the topology of suffix trees}\label{sec:topology}

It is well known that the topology of an ordered tree $T$ with $n$ nodes can be represented using $2n+o(n)$ bits, as a sequence of $2n$ balanced parentheses built by opening a parenthesis, by recurring on every child of the current node in order, and by closing a parenthesis~\cite{munro2001succinct}. To support tree operations on such representation, we will repeatedly use the following data structure:

\begin{lemma}[\cite{SNsoda10,NStalg14}] \label{lemma:balancedParentheses}
Let $T$ be an ordered tree with $n$ nodes, and let $\mathtt{id}(v)$ be the rank of a node $v$ in the preorder traversal of $T$. Given the balanced parentheses representation of $T$ encoded in $2n+o(n)$ bits, we can build a data structure that takes $2n+o(n)$ bits, and that supports the following operations in constant time:
\begin{itemize}
\item $\mathtt{child}(\mathtt{id}(v),i)$: returns $\mathtt{id}(w)$, where $w$ is the $i$th child of node $v$ ($i \geq 1$), or $\emptyset$ if $v$ has less than $i$ children;
\item $\mathtt{parent}(\mathtt{id}(v))$: returns $\mathtt{id}(u)$, where $u$ is the parent of $v$, or $\emptyset$ if $v$ is the root of $T$;
\item $\mathtt{lca}(\mathtt{id}(v),\mathtt{id}(w))$: returns $\mathtt{id}(u)$, where $u$ is the lowest common ancestor of nodes $v$ and $w$;
\item $\mathtt{leftmostLeaf}(\mathtt{id}(v))$, $\mathtt{rightmostLeaf}(\mathtt{id}(v))$: returns one plus the number of leaves that, in the preorder traversal of $T$, are visited before the first (respectively, the last) leaf that belongs to the subtree of $T$ rooted at $v$;
\item $\mathtt{selectLeaf}(i)$: returns $\mathtt{id}(v)$, where $v$ is the $i$-th leaf visited in the preorder traversal of $T$;
\item $\mathtt{depth}(\mathtt{id}(v))$, $\mathtt{height}(\mathtt{id}(v))$: returns the distance of $v$ from the root or from its deepest descendant, respectively;
\item $\mathtt{ancestor}(\mathtt{id}(v),d)$: returns $\mathtt{id}(u)$, where $u$ is the ancestor of $v$ at depth $d$;
\end{itemize}
This data structure can be built in $O(n)$ time and in $O(n)$ bits of working space.
\end{lemma}

Note that the operations supported by Lemma \ref{lemma:balancedParentheses} are enough to implement a preorder traversal of $T$ in small space, as described by the following folklore lemma:

\begin{lemma} \label{lemma:spaceEfficientPreorder}
Let $T$ be an ordered tree with $n$ nodes, and let $\mathtt{id}(v)$ be the rank of a node $v$ in the preorder traversal of $T$. Assume that we have a representation of $T$ that supports the following operations:
\begin{itemize}
\item $\mathtt{firstChild}(\mathtt{id}(v))$: returns the identifier of the first child of node $v$ in the order of $T$;
\item $\mathtt{nextSibling}(\mathtt{id}(v))$: returns the identifier of the child of the parent of node $v$ that follows $v$ in the order of $T$;
\item $\mathtt{parent}(\mathtt{id}(v))$: returns the identifier of the parent of node $v$.
\end{itemize}
Then, a preorder traversal of $T$ can be implemented using $O(\log{n})$ bits of working space.
\end{lemma}
\begin{proof}
During a preorder traversal of $T$ we visit every leaf exactly once, and every internal node exactly twice. Specifically, we visit a node $v$ from its parent, from its previous sibling, or from its last child in the order of $T$. If we visit $v$ from its parent or from its previous sibling, in the next step of the traversal we will visit the first child of $v$ from its parent -- or, if $v$ has no child, we will visit the next sibling of $v$ from its previous sibling if $v$ is not the last child of its parent, otherwise we will visit the parent of $v$ from its last child. If we visit $v$ from its last child, in the next step of the traversal we will visit the next sibling of $v$ from its previous sibling -- or, if $v$ has no next sibling, we will visit the parent of $v$ from its last child. Thus, at each step of the traversal we need to store just $\mathtt{id}(v)$ and a single bit that encodes the direction in which we visited $v$.
\end{proof}

In this paper we will repeatedly traverse trees in preorder. Not surprisingly, the trees we will be interested in are suffix trees or \emph{contractions} of suffix trees, induced by selecting a subset of the nodes of a suffix tree and by connecting such nodes using their ancestry relationship (the parent of a node in the contracted tree is the nearest selected ancestor in the original tree). We will thus repeatedly need the following space-efficient algorithm for building the balanced parentheses representation of a suffix tree:

\begin{lemma}\label{lemma:tree_topology}
Let $S \in [1..\sigma]^{n-1}$ be a string. Assume that we are given an algorithm that enumerates all the intervals of $\SA_{S\#}$ that correspond to an internal node of $\ST_{S\#}$, in $t$ time per interval. Then, we can build the balanced parentheses representation of the topology of $\ST_{S\#}$ in $O(nt)$ time and in $O(n)$ bits of working space.
\end{lemma}
\begin{proof}
We assume without loss of generality that $\log n$ is a power of two. We associate two counters to every position $i \in [1..n]$, one containing the number of open parentheses and the other containing the number of closed parentheses at $i$. We implement such counters with two arrays $C_o[1..n]$ and $C_c[1..n]$. Given the interval $[i..j]$ of an internal node of $\ST_{S\#}$, we just increment $C_o[i]$ and $C_c[j]$. Once all such intervals have been enumerated, we scan $C_o$ and $C_c$ synchronously, and for each $i \in [1..n]$ we write $C_o[i]+1$ open parentheses followed by $C_c[i]+1$ closed parentheses. The total number of parentheses in the output is at most $2(2n-1)$.

A naive implementation of this algorithm would use $O(n\log n)$ bits of working space: we achieve $O(n)$ bits using the static allocation strategy described in Section \ref{sec:staticAllocation}. Specifically, we partition $C_o[1..n]$ into $\lceil n/b \rceil$ blocks containing $b>1$ positions each (except possibly for the last block, which might be smaller), and we assign to each block a counter of $c$ bits. Then, we enumerate the intervals of all internal nodes of the suffix tree, incrementing counter $\lceil i/b \rceil$ every time we want to increment position $i$ in $C_o$. If a counter reaches its maximum value $2^c - 1$, we stop incrementing it and we call \emph{saturated} the corresponding block. The space used by all such counters is $\lceil n/b \rceil \cdot c$ bits, which is $O(n)$ if $c$ is a constant multiple of $b$. At the end of this process, we allocate a memory area of size $b \log n$ bits to each saturated block, so that every position $i$ in a saturated block has $\log n$ bits available to store $C_{o}[i]$. Note that there can be at most $(n-1)/(2^c-1)$ saturated blocks, so the total memory allocated to saturated blocks is at most $nb\log n / (2^c-1)$ bits: this quantity is $o(n)$ if $2^c$ grows faster than $b\log n$.

To every non-saturated block we assign a memory area in which we will store the counters for all the $b$ positions inside the block. Specifically, we will use Elias gamma coding to store a counter value $x \geq 0$ in exactly $1 + 2 \lceil \log(x+1) \rceil \leq 3+2\log (x+1)$ bits \cite{El75}, and we will concatenate the encodings of all the counters in the same block. The space taken by the memory area of a non-saturated block $j$ whose counter has value $t < 2^{c}-1$ is at most:
\begin{eqnarray}
& & \sum_{i=(j-1)b+1}^{jb} \big( 3+2\log(C_{o}[i]+1) \big) \nonumber \\
& \leq & 3b + 2b\log\left( \frac{\sum_{i=(j-1)b+1}^{jb}(C_{o}[i]+1)}{b} \right) \label{eq:js}\\
& = & 3b + 2b\log\left( \frac{t+b}{b} \right) \label{eq:bound} \\
& \leq & 5b + 2t \label{eq:loglinear}
\end{eqnarray}
where Equation \ref{eq:js} derives from applying Jensen's inequality to the logarithm, and Equation \ref{eq:loglinear} comes from the fact that $\log x \leq x$ for all $x \geq 1$. Since $\sum_{i=1}^{\lceil n/b \rceil}t \leq n-1$, it follows that the total number of bits allocated to non-saturated blocks is at most $7n$ for any choice of $b$ (tighter bounds might be possible, but for clarity we don't consider them here). We concatenate the memory areas of all blocks, and we store a prefix-sum data structure that takes $o(n)$ bits and that returns in constant time the starting position of the memory area allocated to any given block (see Lemma \ref{lemma:prefixSums}). We also store a bitvector $\mathtt{isSaturated}[1..\lceil n/b \rceil]$ that marks every saturated block with a one and index it for rank queries.

Once memory allocation is complete, we enumerate again the intervals of all internal nodes of the suffix tree, and for every such interval $[i..j]$ we increment $C_{o}[i]$, as follows. First, we compute the block that contains position $i$, we use $\mathtt{isSaturated}$ to determine whether the block is saturated or not, and we use the prefix-sum data structure to retrieve in constant time the starting position of the region of memory assigned to the block. If the block is saturated, we increment the counter that corresponds to position $i$ directly. Otherwise, we access a precomputed table $T_{s}[1..2^s,1..b]$ such that $T_{s}[i,j]$ stores, for every possible configuration $i$ of $s$ bits interpreted as the concatenation of the Elias gamma coding of $b$ counter values $x^1,x^2,\dots,x^b$, the configuration of $s$ bits that represents the concatenation of the Elias gamma coding of counter values $x^1,x^2,\dots,x^{j-1},x^{j}+1,x^{j+1},\dots,x^b$. The total number of bits used by all such tables is at most $\sum_{s=1}^{y} 2^s b s = 2b+b(y-1)2^{y+1}$, where $y=3b + 2b\log( (t+b)/b )$ with $t=2^{c}-2$ is from Equation \ref{eq:bound}. Thus, we need to choose $b$ and $c$ so that $b y 2^y \in o(n)$: setting $b=\log\log n$ and $c=db$ for any constant $d \geq 1$ makes $y \in O((\log\log n)^2)$, and thus $b y 2^y \in o(n)$. The same choice of $b$ and $c$ guarantees that a cell of $T_s$ can be read in constant time for any $s$, and that the space for the counters in the memory allocation phase of the algorithm is $O(n)$. Finally, setting $d \geq 2$ guarantees that $2^c$ grows faster than $b\log n$, thus putting the total memory allocated to \emph{saturated} blocks in $o(n)$. Tables $T_s$ for all $s \in [1..y]$ can be precomputed in time linear in their size.

Array $C_c$ of closed parentheses can be handled in the same way as array $C_o$.
\end{proof}

We will also need to be able to follow the suffix link that starts from any node of a suffix tree. Specifically, let operation $\mathtt{suffixLink}(\mathtt{id}(v))$ return the identifier of the destination $w$ of a suffix link from node $v$ of $\ST_{S\#}$. The topology of $\ST_{S\#}$ can be augmented to support operation $\mathtt{suffixLink}$, using just $\BWT_{S\#}$:

\begin{lemma}[\cite{Sa07a}] \label{lemma:suffixLink}
Let $S \in [1..\sigma]^{n-1}$ be a string. Assume that we are given the representation of the topology of $\ST_{S\#}$ described in Lemma \ref{lemma:balancedParentheses}, the BWT of $S\#$ indexed to support $\mathtt{select}$ operations in time $t$, and the $C$ array of $S$. Then, we can implement function $\mathtt{suffixLink}(\mathtt{id}(v))$ for any node $v$ of $\ST_{S\#}$ (possibly a leaf) in $O(t)$ time.
\end{lemma}
\begin{proof}
Let $w$ be the destination of the suffix link from $v$, let $[i..j]$ be the interval of node $v$ in $\BWT_{S\#}$, and let $\ell(v)=aW$ and $\ell(w)=W$, where $a \in [0..\sigma]$ and $W \in [0..\sigma]^*$. We convert $\mathtt{id}(v)$ to $[i..j]$ using operations $\mathtt{leftmostLeaf}$ and $\mathtt{rightmostLeaf}$ of the topology. Let $aWX$ and $aWY$ be the suffixes of $S\#$ that correspond to positions $i$ and $j$ in $\BWT_{S\#}$, respectively, where $X$ and $Y$ are strings on alphabet $[0..\sigma]$. Note that the position $i'$ of $WX$ in $\BWT_{S\#}$ is $\mathtt{select}_{a}(\BWT_{S\#}, i-C[a] )$, the position $j'$ of $WY$ in $\BWT_{S\#}$ is $\mathtt{select}_{a}(\BWT_{S\#}, j-C[a] )$, and $W$ is the longest prefix of the suffixes that correspond to positions $i'$ and $j'$ in $\BWT_{S\#}$, which also labels a node of $\ST_{S\#}$. We use operation $\mathtt{selectLeaf}$ provided by the topology of $\ST_{S\#}$ to convert $i'$ and $j'$ to identifiers of leaves in $\ST_{S\#}$, and we compute $\mathtt{id}(w)$ using operation $\mathtt{lca}$ on such leaves.
\end{proof}

Note that, if $aW$ is neither a suffix nor a right-maximal substring of $S\#$, i.e. if $aW$ is always followed by the same character $b \in [1..\sigma]$, the algorithm in Lemma \ref{lemma:suffixLink} maps the locus of $aW$ to the locus of $WbX$ in $\ST_{S\#}$, where $X \in [1..\sigma]^*$ and $aWbX$ is the (unique) shortest right-extension of $aW$ that is right-maximal. The locus of $WbX$ might not be the same as the locus of $W$. As we will see in Section \ref{sec:bwtIndexesConstruction}, this is the reason why the \emph{bidirectional BWT index} of Definition \ref{def:biBWTindex} (on page \pageref{def:biBWTindex}) does not support operation $\mathtt{contractLeft}$ (respectively, $\mathtt{contractRight}$) for strings that are neither suffixes nor right-maximal substrings of $S\#$ (respectively, of $\REV{S}\#$).

\subsection{Data structures for monotone minimal perfect hash functions}\label{sec:build_mmphf}

Given a set $\mathcal{S} \subseteq [1..U]$ of size $n$, a \emph{monotone minimal perfect hash function} (denoted by MMPHF in what follows) is a function $f: [1..U] \mapsto [1..n]$ such that $x<y$ implies $f(x)<f(y)$ for every $x,y \in \mathcal{S}$. In other words, if the set of elements of $\mathcal{S}$ is $x^1<x^2<\dots<x^n$, then $f(x^i)=i$, i.e. the function returns the rank inside $\mathcal{S}$ of the element it takes as an argument. The function is allowed to return an arbitrary value for any $x \in [1..U] \setminus \mathcal{S}$.

To build efficient implementations of MMPHFs, we will repeatedly take advantage of the following lemma:

\begin{lemma}[\cite{BBPV09}]\label{lemma:build_mmphf1}
Let $\mathcal{S} \subseteq [1..U]$ be a set represented in sorted order by the sequence $x^1<x^2<\cdots<x^n$, where $x^i$ is encoded in $\log{U}$ bits for all $i \in [1..n]$ and $\log U<n$. There is an implementation of a MMPHF on $\mathcal{S}$ that takes $O(n\log\log U)$ bits of space, that evaluates $f(x)$ in constant time for any $x \in [1..U]$, and that can be built in \emph{randomized} $O(n)$ time and in $O(n\log U)$ bits of working space.
\end{lemma}
\begin{proof}
We use a technique known as \emph{most-significant-bit bucketing} \cite{belazzougui2011theory}. Specifically, we partition the sequence that represents $\mathcal{S}$ into $\lceil n/b \rceil$ blocks of consecutive elements, where each block $B^i = x^{(i-1)b+1}, \ldots, x^{ib}$ contains exactly $b=\log n$ elements (except possibly for the last block $x^{(i-1)b+1}, \ldots, x^n$, which might be smaller). Then, we compute the length of the longest common prefix $p^i$ of the elements in every $B^i$, starting from the most significant bit. To do so, it suffices to compute the longest prefix that is common to the first and to the last element of $B^i$: this can be done in constant time using the $\mathtt{mostSignificantBit}$ operation, which can be implemented using a constant number of multiplications~\cite{brodnik1993computation}. The length of the longest common prefix of a block is at most $\log U - \log\log n \in O(\log U)$.

Then, we build an implementation of a \emph{minimal perfect} hash function $F$ that maps every element in $\mathcal{S}$ onto a number in $[1..n]$. This can be done in $O(n\log U)$ bits of working space and in \emph{randomized} $O(n)$ time: see \cite{HT01}. We also use a table $\mathtt{lcp}[1..n]$ that stores at index $F(x^i)$ the length of the longest common prefix of the block to which $x^i$ belongs, and a table $\mathtt{pos}[1..n]$ that stores at index $F(x^i)$ the relative position of $x^i$ inside its block. Formally:
\begin{eqnarray*}
\mathtt{lcp}[F(x^i)] & = & |p^{\lfloor (i-1)/b \rfloor + 1}| \\
\mathtt{pos}[F(x^i)] & = & i-b\cdot\lfloor (i-1)/b \rfloor
\end{eqnarray*}
The implementation of $F$ takes $O(n+\log\log U)$ bits of space, $\mathtt{lcp}$ takes $O(n\log\log U)$ bits, and $\mathtt{pos}$ takes $O(n\log\log n)$ bits.

It is folklore that all $p^i$ values are distinct, thus each $p^i$ identifies block $i$ uniquely. We build an implementation of a \emph{minimal perfect} hash function $G$ on set $p^1,p^2,\dots,p^{\lceil n/b \rceil}$, and an inversion table $\mathtt{lcp2block}[1..\lceil n/b\rceil]$ that stores value $i$ at index $G(p^i)$. The implementation of $G$ takes $O(n/\log{n} + \log{\log{U}})$ bits of space, and it can be built in $O((n/\log{n})\log{U})$ bits of working space and in randomized $O(n/\log{n})$ time. Table $\mathtt{lcp2block}$ takes $O((n/\log{n})  \cdot \log(n/\log{n})) = O(n)$ bits. 
With this setup of data structures, we can return in constant time the rank $i$ in $\mathcal{S}$ of any $x^i$, by issuing:
\[
i = b \cdot \mathtt{lcp2block} \Big[ G(x^{i}\big[ 1..\mathtt{lcp}[F(x^i)] \big] \Big] + \mathtt{pos}[F(x^i)]
\]
where $x^{i}[g..h]$ denotes the substring of the binary representation of $x^i$ in $\log{U}$ bits that starts at position $g$ and ends at position $h$.
\end{proof}

We will mostly use Lemma \ref{lemma:build_mmphf1} inside the following construction, which is based on partitioning the universe rather than the set of numbers:

\begin{lemma}\label{lemma:build_mmphf2}
Let $\mathcal{S} \subseteq [1..U]$ be a set represented in sorted order by the sequence $x^1<x^2<\cdots<x^n$, where $x^i$ is encoded in $\log{U}$ bits for all $i \in [1..n]$. There is an implementation of a MMPHF on $\mathcal{S}$ that takes $O(n \log\log b) + \lceil U/b \rceil(2+\lceil\log(nb/U)\rceil)+o(U/b)$ bits of space, that evaluates $f(x)$ in constant time for any $x \in [1..U]$, and that can be built in \emph{randomized} $O(n)$ time and in $O(b\log b)$ bits of working space, for any choice of $b$.
\end{lemma}
\begin{proof}
We will make use of a partitioning technique known as \emph{quotienting}~\cite{Pagh01}). We partition interval $[1..U]$ into $n' \leq n$ blocks of size $b$ each, except for the last block which might be smaller. Note that the most significant $\log U - \log b$ bits are identical in all elements of $\mathcal{S}$ that belong to the same block. 
For each block $i$ that contains more than one element of $\mathcal{S}$, we build an implementation of a monotone minimal perfect hash function $f^i$ on the elements inside the block, as described in Lemma \ref{lemma:build_mmphf1}, \emph{restricted to their least significant $\log{b}$ bits}: all such implementations take $O(n\log\log b)$ bits of space in total, and constructing each of them takes $O(b \log b)$ bits of working space. Then, we use Lemma \ref{lemma:prefixSums} to build a prefix-sum data structure that encodes in $\lceil U/b \rceil(2+\lceil \log(nb/U) \rceil)+o(U/b)$ bits of space the number of elements in every block. Given an element $x \in [1..U]$, we first find the block it belongs to, by computing $i=\lceil x/b \rceil$, then we use the prefix-sum data structure to compute the number $r$ of elements in $\mathcal{S}$ that belong to blocks smaller than $i$, and finally we return $r+f^{i}(x[\log{U}-\log{b}+1..\log{U}])$, where $x[g..h]$ denotes the substring of the binary representation of $x$ in $\log{U}$ bits that starts at position $g$ and ends at position $h$.
\end{proof}

The construction used in Lemma \ref{lemma:build_mmphf2} is a slight generalization of one initially described in~\cite{BBPV09}. 
Setting $b=\lceil U/n \rceil$ in Lemma \ref{lemma:build_mmphf2} makes the MMPHF implementation fit in $O(n \log\log(U/n))$ bits of space.

\subsection{Data structures for range-minimum and range-distinct queries}\label{sec:build_range_color_rep}

Given an array of integers $A[1..n]$, let function $\mathtt{rmq}(i,j)$ return an index $k \in [i..j]$ such that $A[k]=\min\{A[x] : x \in [i..j]\}$, with ties broken arbitrarily. We call this function a \emph{range minimum query} (RMQ) over $A$. It is known that range-minimum queries can be answered by a data structure that is small and efficient to compute:

\begin{lemma}[\cite{Fi10}]\label{lemma:rmq}
Assume that we have a representation of an array of integers $A[1..n]$ that supports accessing the value $A[i]$ stored at any position $i \in [1..n]$ in time $t$. Then, we can build a data structure that takes $2n+o(n)$ bits of space, and that answers $\mathtt{rmq}(i,j)$ for any pair of integers $i<j$ in $[1..n]$ in constant time, \emph{without accessing the representation of $A$}. This data structure can be built in $O(nt)$ time and in $n+o(n)$ bits of working space.
\end{lemma}

Assume now that the elements of array $A[1..n]$ belong to alphabet $[1..\sigma]$, and let $\Sigma_{i,j}$ be the set of distinct characters that occur inside subarray $A[i..j]$. Let function $\mathtt{rangeDistinct}(i,j)$ return the set of tuples $\{ (c,\mathtt{rank}_{A}(c,p_c),\mathtt{rank}_{A}(c,q_c)) : c \in \Sigma_{i,j} \}$ \emph{in any order}, where $p_c$ and $q_c$ are the first and the last occurrence of $c$ in $A[i..j]$, respectively. The frequency of any $c \in \Sigma_{i,j}$ inside $A[i..j]$ is $\mathtt{rank}_{A}(c,q_c)-\mathtt{rank}_{A}(c,p_c)+1$. It is well known that $\mathtt{rangeDistinct}$ queries can be implemented using $\mathtt{rmq}$ queries on a specific array, as described in the following lemma:

\begin{lemma}[\cite{Mu02,Sa07b,BNV13}]\label{lemma:rangedistinct}
Given a string $A \in [1..\sigma]^n$, we can build a data structure of size $n\log{\sigma} + 8n + o(n)$ bits that answers $\mathtt{rangeDistinct}(i,j)$ for any pair of integers $i<j$ in $[1..n]$ in $O(\mathtt{occ})$ time and in $\sigma\log(n+1)$ bits of temporary space, where $\mathtt{occ}=|\Sigma_{i,j}|$. This data structure can be built in $O(kn)$ time and in $(n/k)\log{\sigma}+2n+o(n)$ bits of working space, for any positive integer $k$, and it does not require $A$ to answer $\mathtt{rangeDistinct}$ queries.
\end{lemma}
\begin{proof}
To return just the distinct characters in $\Sigma_{i,j}$ it suffices to build a data structure that supports RMQs on an auxiliary array $P[1..n]$, where $P[i]$ stores the position of the \emph{previous occurrence} of character $A[i]$ in $A$. Since $\mathtt{rmq}(i,j)$ is the leftmost occurrence of character $A[\mathtt{rmq}(i,j)]$ in $A[i..j]$, it is well known that $\Sigma_{i,j}$ can be built by issuing $O(\mathtt{occ})$ $\mathtt{rmq}$ queries on $P$ and $O(\mathtt{occ})$ accesses to $A$, using a stack of $O(\mathtt{occ} \cdot \log n)$ bits and a bitvector of size $\sigma$. This is achieved by setting $k=\mathtt{rmq}(i,j)$, by recurring on subintervals $[i..k-1]$ and $[k+1..j]$, and by using the bitvector to mark the distinct characters observed during the recursion and to stop the process if $A[k]$ is already marked \cite{Mu02}. Random access to array $P$ can be simulated in constant time using $\mathtt{partialRank}$ and $\mathtt{select}$ operations on $A$, which can be implemented as described in Lemma \ref{lemma:rank_select_access} setting $k$ to a constant. We use the data structures of Lemma \ref{lemma:rank_select_access} also to simulate access to $A$ without storing $A$ itself. We build the RMQ data structure using Lemma \ref{lemma:rmq}. After construction, we will never need to answer $\mathtt{select}$ queries on $A$, thus we do not output the $(n/k)\log{\sigma}+n+o(n)$ bits that encode the inverse permutation in Lemma \ref{lemma:rank_select_access}.

To report partial ranks in addition to characters, we adapt this construction as follows. We build a data structure that supports RMQs on an auxiliary array $N[1..n]$, where $N[i]$ stores the position of the \emph{next} occurrence of character $A[i]$ in $A$. Given an interval $[i..j]$, we first use the RMQ data structure on $P$ and a vector $\mathtt{chars}[1..\sigma]$ of $\sigma\log(n+1)$ bits to store the first occurrence $p_c$ of every $c \in \Sigma_{i,j}$. Then, we use the RMQ data structure on $N$ to detect the last occurrence $q_c$ of every $c \in \Sigma_{i,j}$, and we access $\mathtt{chars}[c]$ both to retrieve the corresponding $p_c$ and to clean up cell $\mathtt{chars}[c]$ for the next query. Finally, we compute $\mathtt{rank}_{A}(c,p_c)$ and $\mathtt{rank}_{A}(c,q_c)$ using the $\mathtt{partialRank}$ data structure of Lemma \ref{lemma:rank_select_access}. To build the data structure that supports RMQs on $N$, we can use the same memory area of $n+o(n)$ bits used to build the data structure that supports RMQs on $P$.
\end{proof}

The temporary space used to answer a $\mathtt{rangeDistinct}$ query can be reduced to $\sigma$ bits by more involved arguments \cite{BNV13}. Rather than using $\mathtt{partialRank}$, $\mathtt{select}$, and $\mathtt{access}$, we can implement the $\mathtt{rangeDistinct}$ operation using MMPHFs: the following lemma details this approach, since the rest of the paper will repeatedly use its main technique.

\begin{lemma}[\cite{BNV13}] \label{lemma:rangedistinct2}
We can augment a string $A \in [1..\sigma]^n$ with a data structure of size $O(n\log{\log{\sigma}})$ bits that answers $\mathtt{rangeDistinct}(i,j)$ for any pair of integers $i<j$ in $[1..n]$ in $O(\mathtt{occ})$ time and in $\sigma\log(n+1)$ bits of temporary space, where $\mathtt{occ}=|\Sigma_{i,j}|$. This data structure can be built in $O(n)$ \emph{randomized} time and in $O(n\log{\sigma})$ bits of working space.
\end{lemma}
\begin{proof}
We build the set of sequences $\{ P_c : c \in [1..\sigma] \}$, such that $P_c$ contains all the positions $p_1,p_2,\dots,p_k$ of character $c$ in $A$ in increasing order. We encode $P_c$ as a bitvector such that position $p_i$ for $i>1$ is represented by the Elias gamma coding of $p_i-p_{i-1}$. The total space taken by all such sequences is $O(n\log{\sigma})$ bits, by applying Jensen's inequality twice. Let $|P_c|$ be the number of bits in $P_c$: we compute $|P_c|$ and we allocate a corresponding region of memory using the static allocation strategy described in Section \ref{sec:staticAllocation}. We also mark in an additional bitvector $\mathtt{start}_{c}[1..|P_c|]$ the first bit of every representation of a $p_i$ in $P_c$, and we index $\mathtt{start}_{c}$ to support $\mathtt{select}$ queries.

Then, we build an implementation of an MMPHF for every $P_c$, using Lemma \ref{lemma:build_mmphf2} with $U=n$ and $b=\sigma^k$ for some positive integer $k$. Specifically, for every $c$, we perform a single scan of sequence $P_c$, decoding all the positions that fall inside the same block of $A$ of size $\sigma^k$, and building an implementation of an MMPHF for the positions inside the block. Once all such MMPHFs have been built, we discard all $P_c$ sequences. The total space used by all MMPHF implementations is at most $O(n(\log{\log{\sigma}}+\log{k})) + (nk/\sigma^k)\log{\sigma} + 2n/\sigma^k + o(n/\sigma^k)$ bits: any $k \geq 1$ makes such space fit in $O(n \log{\log{\sigma}})$ bits, and it makes the working space of the construction fit in $O(\sigma^{k}\log{\sigma})$ bits. Since we assumed $\sigma \in o(\sqrt{n}/\log{n})$, setting $k \in \{1,2\}$ makes this additional space fit in $O(n\log{\sigma})$ bits.

Finally, we proceed as in Lemma \ref{lemma:rangedistinct}. Given a position $i$, we can compute $\mathtt{rank}_{A}(A[i],i)$ by querying the MMPHF data structure of character $A[i]$, and we can simulate random access to $P[i]$ by querying the MMPHF data structure of character $A[i]$ and by accessing $p_i-P[i]$ using a $\mathtt{select}$ operation on $\mathtt{start}_{A[i]}$.
\end{proof}

Lemma \ref{lemma:rangedistinct} builds an internal representation of $A$, and the original representation of $A$ provided in the input can be discarded. On the other hand, Lemma \ref{lemma:rangedistinct2} uses the input representation of $A$ to answer queries, thus it can be combined with any representation of $A$ that allows constant-time access -- for example with those that represent $A$ up to its $k$th order empirical entropy for $k \in o(\log_{\sigma}n)$ \cite{ferragina2007simple}.


\section{Enumerating all right-maximal substrings}\label{sec:enumeration}

The following problem lies at the core of our construction and, as we will see in Section \ref{sec:stringAnalysis}, it captures the requirements of a number of fundamental string analysis algorithms:

\begin{problem} \label{problem:iterator}
Given a string $T \in [1..\sigma]^{n-1}\#$, return the following information for all right-maximal substrings $W$ of $T$:
\begin{itemize}
\item $|W|$ and $\INTERVAL{W}$ in $\SA_{T}$;
\item the sorted sequence $b_1 < b_2 < \dots < b_k$ of all the distinct characters in $[0..\sigma]$ such that $W b_i$ is a substring of $T$;
\item the sequence of intervals $\INTERVAL{Wb_1},\dots,\INTERVAL{Wb_k}$;
\item a sequence $a_1, a_2, \dots, a_h$ that lists all the $h$ distinct characters in $[0..\sigma]$ such that $a_i W$ is a prefix of a rotation of $T$; the sequence $a_1, a_2, \dots, a_h$ is not necessarily in lexicographic order;
\item the sequence of intervals $\INTERVAL{a_1 W},\dots,\INTERVAL{a_h W}$.
\end{itemize}
\end{problem}

Problem \ref{problem:iterator} does not specify the order in which the right-maximal substrings of $T$ (or equivalently, the internal nodes of $\ST_T$) must be enumerated, nor the order in which the left-extensions $a_i W$ of a right-maximal substring $W$ must be returned. It does, however, specify the order in which the \emph{right-extensions} $W b_i$ of $W$ must be returned.

The first step for solving Problem \ref{problem:iterator} consists in devising a suitable representation for a right-maximal substring $W$ of $T$. Let $\gamma(a,W)$ be the number of distinct strings $Wb$ such that $aWb$ is a prefix of a rotation of $T$, where $a \in [0..\sigma]$ and $b \in \{b_1,\dots,b_k\}$. Note that there are precisely $\gamma(a,W)$ distinct characters to the right of $aW$ when it is a prefix of a rotation of $T$: thus, if $\gamma(a,W)=0$, then $aW$ is not a prefix of any rotation of $T$; if $\gamma(a,W)=1$ (for example when $a=\#$), then $aW$ is not a right-maximal substring of $T$; and if $\gamma(a,W) \geq 2$, then $aW$ is a right-maximal substring of $T$. This suggests to represent a substring $W$ of $T$ with the following pair:
$$
\REPR{W} = (\mathtt{chars}[1..k],\mathtt{first}[1..k+1])
$$
where $\mathtt{chars}[i]=b_i$ and $\INTERVAL{Wb_i}=[\mathtt{first}[i]..\mathtt{first}[i+1]-1]$ for $i \in [1..k]$. Note that $\INTERVAL{W} = \big[ \mathtt{first}[1]..\mathtt{first}[k+1]-1 \big]$, since it coincides with the concatenation of the intervals of the right-extensions of $W$ in lexicographic order. If $W$ is not right-maximal, array $\mathtt{chars}$ and $\mathtt{first}$ in $\REPR{W}$ have length one and two, respectively.

Given $\REPR{W}$, $\REPR{a_i W}$ can be precomputed for all $i \in [1..h]$, as follows:

\begin{lemma} \label{lemma:extendLeft}
Assume the notation of Problem \ref{problem:iterator}. Given a data structure that supports $\mathtt{rangeDistinct}$ queries on $\BWT_T$, given the $C$ array of $T$, and given $\REPR{W}=(\mathtt{chars}[1..k],\mathtt{first}[1..k+1])$ for a substring $W$ of $T$, we can compute the sequence $a_1,\dots,a_h$ and the corresponding sequence $\REPR{a_1 W},\dots,\REPR{a_h W}$, in $O(t \cdot \mathtt{occ})$ time and in $O(\sigma^{2}\log{n})$ bits of temporary space, where $t$ is the time taken by the $\mathtt{rangeDistinct}$ operation per element in its output, and $\mathtt{occ}$ is the number of distinct strings $a_{i}Wb_j$ that are the prefix of a rotation of $T$, where $i \in [1..h]$ and $j \in [1..k]$.
\end{lemma}
\begin{proof}
Let $\mathtt{leftExtensions}[1 \ltdots \sigma+1]$ be a vector of characters given in input to the algorithm and initialized to all zeros, and let $h$ be the number of nonempty cells in this vector. We will store in vector $\mathtt{leftExtensions}$ all characters $a_1,a_2,\dots,a_h$, not necessarily in lexicographic order. Consider also matrices $A[0 \ltdots \sigma, 1 \ltdots \sigma+1]$, $F[0 \ltdots \sigma, 1 \ltdots \sigma+1]$ and $L[0 \ltdots \sigma, 1 \ltdots \sigma+1]$, given in input to the algorithm and initialized to all zeros, whose rows correspond to possible left-extensions of $W$. We will store character $b_j$ in cell $A[a_i,p]$, for increasing values of $p$ starting from one, iff $a_i W b_j$ is the prefix of a rotation of $T$: in this case, we will also set $F[a_i,p]=\SP{a_i W b_j}$ and $L[a_i,p]=\EP{a_i W b_j}$. In other words, every triplet $(A[a_i,p],F[a_i,p],L[a_i,p])$ identifies the right-extension $Wb_j$ of $W$ associated with character $b_j = A[a_i,p]$, and it specifies the interval of $a_i W b_j$ in $\BWT_T$ (see Figure \ref{fig:unidirectionalBWTIndex}). We use array $\mathtt{gamma}[0 \ltdots \sigma]$, given in input to the algorithm and initialized to all zeros, to maintain, for every $a \in [0 \ltdots \sigma]$, the number of distinct characters $b \in \{b_1,\dots,b_k\}$ such that $aWb$ is the prefix of a rotation of $T$, or equivalently the number of nonempty cells in row $a$ of matrices $A$, $F$ and $L$. In other words, $\mathtt{gamma}[a]=\gamma(a,W)$.

For every $j \in [1..k]$, we enumerate all the distinct characters that occur inside the interval $\BWT_{T}[\mathtt{first}[j]..\mathtt{first}[j+1]-1]$ of string $Wb_j=W \cdot \mathtt{chars}[j]$, along with the corresponding partial ranks, using operation $\mathtt{rangeDistinct}$. Recall that $\mathtt{rangeDistinct}$ does not necessarily return such characters in lexicographic order. For every character $a$ returned by $\mathtt{rangeDistinct}$, we compute $\INTERVAL{aWb_j}$ in constant time using the $C$ array and the partial ranks, we increment counter $\mathtt{gamma}[a]$ by one, and we set: 
\begin{eqnarray*}
A\big[a,\mathtt{gamma}[a]\big] & = & \mathtt{chars}[j] \\
F\big[a,\mathtt{gamma}[a]\big] & = & \SP{aWb_j} \\
L\big[a,\mathtt{gamma}[a]\big] & = & \EP{aWb_j}
\end{eqnarray*}
See Figure \ref{fig:unidirectionalBWTIndex} for an example. If $\mathtt{gamma}[a]$ transitioned from zero to one, we increment $h$ by one and we set $\mathtt{leftExtensions}[h]=a$. At the end of this process, $\mathtt{leftExtensions}[i]=a_i$ for $i \in [1..h]$ (note again that the characters in $\mathtt{leftExtensions}[1..h]$ are not necessarily sorted lexicographically), the nonempty rows in $A$, $F$ and $L$ correspond to such characters, the characters that appear in row $a_i$ of matrix $A$ are sorted lexicographically, and the corresponding intervals $\big[ F[a_i,p]..L[a_i,p] \big]$ are precisely the intervals of string $a_i W \cdot A[a_i,p]$ in $\BWT_T$. It follows that such intervals are adjacent in $\BWT_T$, thus:
$$
\REPR{a_i W} = \left(A\big[a_i,1..\mathtt{gamma}[a_i]\big],F\big[a_i,1..\mathtt{gamma}[a_i]\big] \bullet \big(L\big[a_i,\mathtt{gamma}[a_i]\big]+1\big) \right)
$$
where $X \bullet y$ denotes appending number $y$ to the end of array $X$. We can restore all matrices and vectors to their original state within the claimed time budget, by scanning over all cells of $\mathtt{leftExtensions}$, using their value to address matrices $A$, $F$ and $L$, and using array $\mathtt{gamma}$ to determine how many cells must be cleaned in each row of such matrices.
\end{proof}

\begin{figure}[t!]
\begin{center}
\includegraphics[width = 0.8\textwidth]{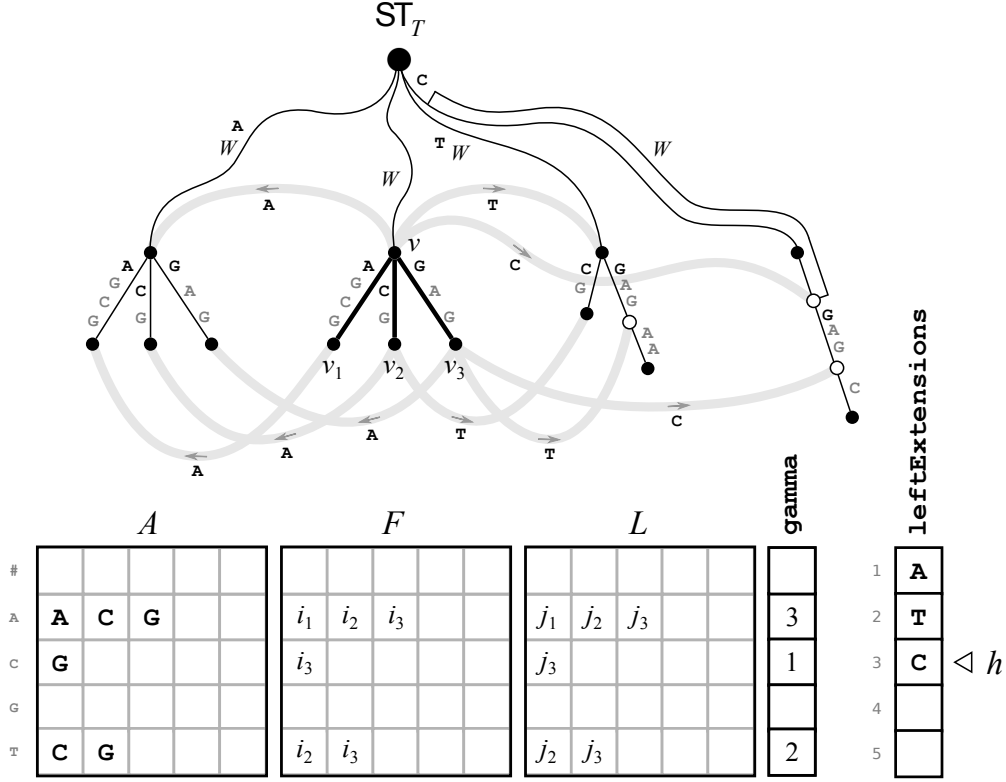}
\caption{Lemma \ref{lemma:extendLeft} applied to the right-maximal substring $W=\ell(v)$. Gray directed arcs represent implicit and explicit Weiner links. White dots represent the destinations of implicit Weiner links. Child $v_k$ of node $v$ in $\ST_T$ has interval $[i_k..j_k]$ in $\BWT_T$, where $k \in [1..3]$. Among all strings prefixed by string $W$, only those prefixed by $W\mathtt{GAG}$ are preceded by $\mathtt{C}$: it follows that $\mathtt{C}W$ is always followed by $\mathtt{G}$ and it is not right-maximal, thus the Weiner link from $v$ labeled by $\mathtt{C}$ is implicit. Conversely, $W\mathtt{AGCG}$, $W\mathtt{CG}$ and $W\mathtt{GAG}$ are all preceded by an $\mathtt{A}$, so $\mathtt{A}W$ is right-maximal and the Weiner link from $v$ labeled by $\mathtt{A}$ is explicit.
}
\label{fig:unidirectionalBWTIndex}
\end{center}
\end{figure}

Iterated applications of Lemma \ref{lemma:extendLeft} are almost all we need to solve Problem \ref{problem:iterator} efficiently, as described in the following lemma:

\begin{lemma} \label{lemma:iterator}
Given a data structure that supports $\mathtt{rangeDistinct}$ queries on the BWT of a string $T \in [1..\sigma]^{n-1}\#$, and given the $C$ array of $T$, there is an algorithm that solves Problem \ref{problem:iterator} in $O(nt)$ time and in $O(\sigma^2\log^{2}{n})$ bits of working space, where $t$ is the time taken by the $\mathtt{rangeDistinct}$ operation per element in its output.
\end{lemma}
\begin{proof}
We use again the notation of Problem \ref{problem:iterator}. Assume by induction that we know $\REPR{W}=(\mathtt{chars}[1..k],\mathtt{first}[1..k+1])$ and $|W|$ for some right-maximal substring $W$ of $T$. Using Lemma \ref{lemma:extendLeft}, we compute $a_i$ and $\REPR{a_i W}=(\mathtt{chars}_{i}[1..k_i],\mathtt{first}_{i}[1..k_i+1])$ for all $i \in [1..h]$, and we determine whether $a_i W$ is right-maximal by checking whether $|\mathtt{chars}_i|>1$, or equivalently whether $\mathtt{gamma}[a_i]>1$ in Lemma \ref{lemma:extendLeft}: if this is the case, we push pair $(\REPR{a_i W},|W|+1)$ to a stack $S$. In the next iteration, we pop the representation of a string from the stack and we repeat the process, until the stack becomes empty. Note that this is equivalent to following all the explicit Weiner links from (or equivalently, all the reverse suffix links to) the node $v$ of $\ST_T$ with $\ell(v)=W$, not necessarily in lexicographic order. Thus, running the algorithm from a stack initialized with $\REPR{\varepsilon}$ is equivalent to a depth-first traversal of the \emph{suffix-link tree} of $T$ (not necessarily following the lexicographic order of Weiner link labels): recall from Section \ref{sec:suffixTree} that a traversal of $\SLT_T$ guarantees to enumerate all the right-maximal substrings of $T$. Triplet $\REPR{\varepsilon}$ can be easily built from the $C$ array of $T$.

Every $\mathtt{rangeDistinct}$ query performed by the algorithm can be charged to a distinct node of $\ST_T$, and every tuple in the output of all such $\mathtt{rangeDistinct}$ queries can be charged to a distinct (explicit or implicit) Weiner link. It follows from Observation \ref{obs:suffixtree} that the algorithm runs in $O(nt)$ time. Since the algorithm performs a depth-first traversal of the suffix-link tree of $T$, the depth of the stack is bounded by the length of a longest right-maximal substring of $T$. More precisely, since we always pop the element at the top of the stack, the depth of the stack is bounded by quantity $\mu_T$ defined in \Sec{sect:definitions_strings}, i.e. by the largest number of (not necessarily proper) suffixes of a maximal repeat that are themselves maximal repeats. Even more precisely, since we push just right-maximal substrings, the depth of the stack is bounded by quantity $\lambda_T$ defined in \Sec{sect:definitions_strings}. Unfortunately, $\lambda_T$ might be $O(n)$. We reduce this depth to $O(\log n)$ by pushing at every iteration the pair $(\REPR{a_i W},|a_i W|)$ with largest $\INTERVAL{a_i W}$ first (a technique already described in \cite{hoare1962quicksort}): the interval of every other $aW$ is necessarily at most half of $\INTERVAL{W}$, thus stack $S$ contains at any time pairs from $O(\log n)$ suffix-link tree levels. Every such level contains $O(\sigma)$ pairs, and every pair takes $O(\sigma\log{n})$ bits, thus the total space used by the stack is $O(\sigma^2 \log^{2}{n})$ bits.
\end{proof}

Algorithm \ref{algo:iterator} summarizes Lemma \ref{lemma:iterator} in pseudocode. 
Combining Lemma \ref{lemma:iterator} with the $\mathtt{rangeDistinct}$ data structure of Lemma \ref{lemma:rangedistinct} we obtain the following result:

\begin{theorem} \label{thm:iterator}
Given the BWT of a string $T \in [1..\sigma]^{n-1}\#$, we can solve Problem \ref{problem:iterator} in $O(nk)$ time, and in $n\log{\sigma}(1+1/k)+10n+(2\sigma+1)\log(n+1)+o(n)=n\log{\sigma}(1+1/k)+O(n)+O(\sigma\log{n})$ bits of working space, for any positive integer $k$.
\end{theorem}
\begin{proof}
Lemma \ref{lemma:iterator} needs just the $C$ array, which takes $(\sigma+1)\log n$ bits, and a $\mathtt{rangeDistinct}$ data structure: the one in Lemma \ref{lemma:rangedistinct} takes $n\log{\sigma}+8n+o(n)$ bits of space, and it answers queries in time linear in the size of their output and in $\sigma\log{(n+1)}$ bits of space in addition to the output. 
Building the $C$ array from $\BWT_T$ takes $O(n)$ time, and building the $\mathtt{rangeDistinct}$ data structure of Lemma \ref{lemma:rangedistinct} takes $O(nk)$ time and $(n/k)\log{\sigma}+2n+o(n)$ bits of working space, for any positive integer $k$.
\end{proof}

Note that replacing Lemma \ref{lemma:rangedistinct} in Theorem \ref{thm:iterator} with the alternative construction of Lemma \ref{lemma:rangedistinct2} introduces randomization and it does not improve space complexity.

As we saw in Section \ref{sec:topology}, having an efficient algorithm to enumerate all intervals in $\BWT_T$ of right-maximal substrings of $T$ has an immediate effect on the construction of the balanced parentheses representation of $\ST_T$. The following result derives immediately from plugging Theorem \ref{thm:iterator} in Lemma \ref{lemma:tree_topology}:

\begin{theorem}\label{thm:topology}
Given the BWT of a string $T \in [1..\sigma]^{n-1}\#$, we can build the balanced parentheses representation of the topology of $\ST_T$ in $O(nk)$ time and in $n\log{\sigma}(1+1/k)+O(n)$ bits of working space, for any positive integer $k$.
\end{theorem}

In this paper we will also need to enumerate all the right-maximal substrings of a \emph{concatenation} $T=T^1 T^2 \cdots T^m$ of $m$ strings $T^1,T^2,\dots,T^m$, where $T^i \in [1..\sigma]^{n_{i}-1}\#_i$ for $i \in [1..m]$. Recall that the right-maximal substrings of $T$ correspond to the internal nodes of the \emph{generalized suffix tree} of $T^1,T^2,\dots,T^m$, thus we can solve Problem \ref{problem:generalizedIterator} by applying Lemma \ref{lemma:iterator} to the BWT of $T$. If we are just given the BWT of each $T^i$ separately, however, we can represent a substring $W$ as a pair \emph{of sets of arrays} $\REPRPRIME{W}=(\{\mathtt{chars}^1,\ldots,\mathtt{chars}^m\},\{\mathtt{first}^1\ldots\mathtt{first}^m\})$, where $\mathtt{chars}^i$ collects all the distinct characters $b$ such that $Wb$ is observed in string $T^i$, in lexicographic order, and the interval of string $W \cdot \mathtt{chars}^i[j]$ in $\BWT_{T^i}$ is $\big[ \mathtt{first}^i[j]..\mathtt{first}^i[j+1]-1 \big]$. If $W$ does not occur in $T^i$, we assume that $|\mathtt{chars}^i|=0$ and that $\mathtt{first}^{i}[1]$ equals one plus the number of suffixes of $T^i$ that are lexicographically smaller than $W$. If necessary, this representation can be converted in $O(m\sigma)$ time into a representation based on intervals of $\BWT_T$. We can thus adapt the approach of Lemma \ref{lemma:iterator} to solve the following generalization of Problem \ref{problem:iterator}, as described in Lemma \ref{lemma:generalizedIterator} below:

\begin{problem} \label{problem:generalizedIterator}
Given strings $T^1,T^2,\dots,T^m$ with $T^i \in [1..\sigma]^{n_{i}-1}\#_i$ for $i \in [1..m]$, return the following information for all right-maximal substrings $W$ of $T=T^1 T^2 \cdots T^m$:
\begin{itemize}
\item $|W|$ and $\INTERVAL{W}$ in $\SA_T$;
\item the sorted sequence $b_1 < b_2 < \dots < b_k$ of all the distinct characters in $[-m+1..\sigma]$ such that $W b_i$ is a substring of $T$;
\item the sequence of intervals $\INTERVAL{Wb_1},\dots,\INTERVAL{Wb_k}$;
\item a sequence $a_1, a_2, \dots, a_h$ that lists all the $h$ distinct characters in $[-m+1..\sigma]$ such that $a_i W$ is the prefix of a rotation of $T$; the sequence $a_1, a_2, \dots, a_h$ is not necessarily in lexicographic order;
\item the sequence of intervals $\INTERVAL{a_1 W},\dots,\INTERVAL{a_h W}$.
\end{itemize}
\end{problem}

\begin{lemma} \label{lemma:generalizedIterator}
Assume that we are given a data structure that supports $\mathtt{rangeDistinct}$ queries on the BWT of a string $T^i$, and the $C$ array of $T^i$, for all strings in a set $\{T^1,T^2,\dots,T^m\}$, where $T^i \in [1..\sigma]^{n_{i}-1}\#_i$ for $i \in [1..m]$. There is an algorithm that solves Problem \ref{problem:generalizedIterator} in $O(mnt)$ time and in $O(m \sigma^2 \log^{2}n)$ bits of working space, where $t$ is the time taken by the $\mathtt{rangeDistinct}$ operation per element in its output, and $n=\sum_{i=1}^{m}n_i$.
\end{lemma}
\begin{proof}
To keep the presentation as simple as possible we omit the details on how to handle strings that occur in some $T^i$ but that do not occur in some $T^j$ with $j \neq i$. We use the same algorithm as in Lemma \ref{lemma:iterator}, but this time with the following data structures:
\begin{itemize}
\item $m$ distinct arrays $\mathtt{gamma}^1,\mathtt{gamma}^2,\dots,\mathtt{gamma}^m$;
\item $m$ distinct matrices $A^1,A^2,\dots,A^m$, $F^1,F^2,\dots,F^m$, and $L^1,L^2,\dots,L^m$;
\item a single stack, in which we push $\REPRPRIME{W}$ tuples;
\item a single array $\mathtt{leftExtensions}[1..\sigma+m]$, which stores all the distinct left-extensions of a string $W$ that are the prefix of a rotation of a string $T^i$, not necessarily in lexicographic order.
\end{itemize}
Given $\REPRPRIME{W}$ for a right-maximal substring $W$ of $T$, we apply Lemma \ref{lemma:extendLeft} to each $T^i$ to compute the corresponding $\REPR{aW}$ for all strings $aW$ that are the prefix of a rotation of $T^i$, updating row $a$ in $A^i$, $F^i$, $L^i$ and $\mathtt{gamma}^i$ accordingly, and adding a character $a$ to the shared array $\mathtt{leftExtensions}$ whenever we see $a$ for the first time in any $T^i$ (see Algorithm \ref{algo:extendLeftPrime}). If $a=\#_i$, we assume it is actually $\#_{i-1}$ if $i>1$, and we assume it is $\#_{m}$ if $i=1$. We push to the shared stack the pair $\REPRPRIME{aW}=(\{\mathtt{chars}^1\ldots \mathtt{chars}^m\},\{\mathtt{first}^1\ldots\mathtt{first}^m\})$ such that $\mathtt{chars}^i=A^i[a,1..\mathtt{gamma}^{i}[a]]$, $\mathtt{first}^i=F^i[a,1..\mathtt{gamma}^{i}[a]] \bullet (L^i[a,\mathtt{gamma}^{i}[a]]+1)$ for all $i \in [1..m]$, if and only if $aW$ is right-maximal in $T$, or equivalently iff there is an $i \in [1..m]$ such that $\mathtt{gamma}^i[a]>1$, or alternatively if there are two integers $i \neq j$ in $[1..m]$ such that $\mathtt{gamma}^i[a]=1$, $\mathtt{gamma}^j[a]=1$, and $A^i[a][1] \neq A^j[a][1]$ (see Algorithm \ref{algo:generalizedIterator}). Note that we never push $\REPRPRIME{aW}$ with $a = \#_i$ in the stack, thus the space taken by the stack is $O(m \sigma^2 \log^{2}n)$ bits. In analogy to Lemma \ref{lemma:iterator}, we push first to the stack the left-extension $aW$ of $W$ that maximizes $\sum_{i=1}^{m}|\INTERVALFUNCTION(aW,T^i)| = \sum_{i=1}^{m}L^{i}[a,\mathtt{gamma}^{i}[a]]-F^{i}[a,1]+1$. The result of this process is a traversal of the suffix-link tree of $T$, not necessarily following the lexicographic order of its Weiner link labels. The total cost of translating every $\REPRPRIME{W}$ into the quantities required by Problem \ref{problem:generalizedIterator} is $O(mn)$.
\end{proof}

Recall that we say that a node (possibly a leaf) of the suffix tree of $T=T^1 T^2 \cdots T^m$ is \emph{pure} if all the leaves in its subtree are suffixes of exactly one string $T^i$, and we call it \emph{impure} otherwise. Lemma \ref{lemma:generalizedIterator} can be adapted to traverse only impure nodes of the generalized suffix tree. This leads to the following algorithm for building the BWT of $T$ from the BWT of $T^1,T^2,\cdots,T^m$:

\begin{lemma} \label{lemma:generalizedIterator_bwtMerge}
Assume that we are given a data structure that supports $\mathtt{rangeDistinct}$ queries on the BWT of a string $T^i$, and the $C$ array of $T^i$, for all strings in a set $\{T^1,T^2,\dots,T^m\}$, where $T^i \in [1..\sigma]^{n_{i}-1}\#_i$ for $i \in [1..m]$. There is an algorithm that builds the BWT of string $T = T^1 T^2 \cdots T^m$ in $O(mnt)$ time and in $O(m \sigma^2 \log^{2}n)$ bits of working space, where $t$ is the time taken by the $\mathtt{rangeDistinct}$ operation per element in its output, and $n=\sum_{i=1}^{m}n_i$.
\end{lemma}
\begin{proof}
The BWT of $T$ can be partitioned into disjoint intervals that correspond to \emph{pure nodes of minimal depth} in $\ST_{T}$, i.e. to pure nodes whose parent is impure. In $\ST_T$, suffix links from impure nodes lead to other impure nodes, so the set of all impure nodes is a subgraph of the suffix-link tree of $T$, it includes the root, and it can be traversed by iteratively taking explicit Weiner links from the root. We modify Algorithm \ref{algo:generalizedIterator} to traverse only impure internal nodes of $\ST_T$, by pushing to the stack $\REPRPRIME{aW}=(\{\mathtt{chars}^i\},\{\mathtt{first}^i\})$, where $\mathtt{chars}^i=A^i[a,1..\mathtt{gamma}^i[a]]$ and $\mathtt{first}^i=F^i[a,1..\mathtt{gamma}^i[a]] \bullet (L^i[a,\mathtt{gamma}^i[a]]+1)$ for all $i \in [1..m]$, iff it represents an internal node of $\ST_T$, and moreover if there are two integers $i \neq j$ in $[1..m]$ such that $\mathtt{gamma}^i[a]>0$ and $\mathtt{gamma}^j[a]>0$.

Assume that we enumerate an impure internal node of $\ST_{T}$ with label $W$, and let $\REPRPRIME{W}=(\{\mathtt{chars}^i\},\{\mathtt{first}^i\})$. We merge in linear time the set of sorted arrays $\{\mathtt{chars}^i\}$. Assume that character $b=\mathtt{chars}^i[j]$ occurs only in $\mathtt{chars}^i$. It follows that the locus of $Wb$ in $\ST_T$ is a pure node of minimal depth, and we can copy $\BWT_{T^i}\big[ \mathtt{first}^i[j] .. \mathtt{first}^i[j+1]-1 \big]$ to $\BWT_{T}\big[ x .. x+\mathtt{first}^i[j+1]-\mathtt{first}^i[j]-1 \big]$, where $x = 1+\sum_{i=1}^{m}\mathtt{smaller}(b,i)$ and
$$
\mathtt{smaller}(b,i) = \left\{
\begin{array}{ll}
\mathtt{first}^i[1]-1 & \mbox{if }(|\mathtt{chars}^i|=0) \:\mbox{or} \:(\mathtt{chars}^i[1] \geq b) \\
\max_{j : \mathtt{chars}^i[j]<b}\{ \mathtt{first}^i[j+1]-1 \} & \mbox{otherwise }
\end{array}
\right.
$$

The value of $x$ can be easily maintained while merging set $\{\mathtt{chars}^i\}$. If character $b$ occurs in more than one $
\mathtt{chars}^i$ array, then the locus of $Wb$ in $\ST_T$ is impure, and it will be enumerated (or it has already been enumerated) by the traversal algorithm.
\end{proof}

In the rest of the paper we will focus on the case $m=2$. The following theorem, which we will use extensively in Section \ref{sec:stringAnalysis}, combines Lemma \ref{lemma:generalizedIterator} for $m=2$ with the $\mathtt{rangeDistinct}$ data structure of Lemma \ref{lemma:rangedistinct}:

\begin{theorem} \label{thm:generalizedIterator}
Given the BWT of a string $S^1 \in [1..\sigma]^{n_1-1}\#_1$ and the BWT of a string $S^2 \in [1..\sigma]^{n_2-1}\#_2$, we can solve Problem \ref{problem:generalizedIterator} in $O(nk)$ time and in $n\log{\sigma}(1+1/k)+10n+o(n)$ bits of working space, for any positive integer $k$, where $n=n_1+n_2$.
\end{theorem}

Finally, in Section \ref{sec:bwtConstruction} we will work on strings that are not terminated by a special character, thus we will need the following version of Lemma \ref{lemma:generalizedIterator_bwtMerge} that works on sets of rotations rather than on sets of suffixes:

\begin{theorem} \label{thm:generalizedIterator_bwtMerge}
Let $S^1 \in [1..\sigma]^{n_1}$ and $S^2 \in [1..\sigma]^{n_2}$ be two strings such that $|\mathcal{R}(S^1)|=n_1$, $|\mathcal{R}(S^2)|=n_2$, and $\mathcal{R}(S^1) \cap \mathcal{R}(S^2) = \emptyset$. Given the BWT of $\mathcal{R}(S^1)$ and the BWT of $\mathcal{R}(S^2)$, we can build the BWT of $\mathcal{R}(S^1) \cup \mathcal{R}(S^2)$ in $O(nk)$ time and in $n\log{\sigma}(1+1/k)+10n+o(n)$ bits of working space, where $n=n_1+n_2$.
\end{theorem}
\begin{proof}
Since all rotations of $S^i$ are lexicographically distinct, the compact trie of all such rotations is well defined, and every leaf of such trie corresponds to a distinct rotation of $S^i$. Since no rotation of $S^1$ is lexicographically identical to a rotation of $S^2$, the generalized compact trie that contains all rotations of $S^1$ and all rotations of $S^2$ is well defined, and every leaf of such trie corresponds to a distinct rotation of $S^1$ or of $S^2$. We can thus traverse such generalized compact trie using $\BWT_{S^1}$ and $\BWT_{S^2}$ as described in Lemma \ref{lemma:generalizedIterator_bwtMerge}, using Lemma \ref{lemma:rangedistinct} to implement $\mathtt{rangeDistinct}$ data structures.
\end{proof}

\begin{algorithm}
\KwIn{$\REPR{W}$ for a substring $W$ of $T$. Support for $\mathtt{rangeDistinct}$ queries on the BWT of $T$. $C$ array of $T$. Empty matrices $A$, $F$, and $L$, empty arrays $\mathtt{gamma}$ and $\mathtt{leftExtensions}$, and a pointer $h$.}
\KwOut{Matrices $A$, $F$, $L$, pointer $h$, arrays $\mathtt{gamma}$ and $\mathtt{leftExtensions}$, filled as described in Lemma \ref{lemma:iterator}.}

	$(\mathtt{chars},\mathtt{first}) \gets \REPR{W}$\;
	$h \gets 0$\;
	\For{$j \in [1..|\mathtt{chars}|]$}{
   		$\mathcal{I} \gets \BWT_{T}.\mathtt{rangeDistinct}(\mathtt{first}[j],\mathtt{first}[j+1]-1)$\;
		\For{$(a,p_a,q_a) \in \mathcal{I}$}{
			\If{$\mathtt{gamma}[a]=0$}{
				$h \gets h+1$\;
				$\mathtt{leftExtensions}[h] \gets a$\;
			}
			$\mathtt{gamma}[a] \gets \mathtt{gamma}[a]+1$\;
			$A[a,\mathtt{gamma}[a]] \gets \mathtt{chars}[j]$\;
			$F[a,\mathtt{gamma}[a]] \gets C[a]+p_a$\;
			$L[a,\mathtt{gamma}[a]] \gets C[a]+q_a$\;
		}
   }

\caption{ \label{algo:extendLeft}
{Building $\REPR{aW}$ from $\REPR{W}$ for all $a \in [0..\sigma]$ such that $aW$ is a prefix of a rotation of $T \in [1..\sigma]^{n-1}\#$.}
}
\end{algorithm}

\begin{algorithm}
\KwIn{BWT transform and $C$ array of $T$. Array $\mathtt{distinctChars}$ of all the distinct characters that occur in $T$, in lexicographic order, and array $\mathtt{start}$ of starting positions of the corresponding intervals in $\BWT_{T}$. Support for $\mathtt{rangeDistinct}$ queries on $\BWT_{T}$, and an implementation of Algorithm \ref{algo:extendLeft} (function $\mathtt{extendLeft}$).}
\KwOut{$(\REPR{W},|W|)$ for all right-maximal substrings $W$ of $T$.}
$S \gets $ empty stack\;
$A \gets \mathtt{zeros}[0..\sigma,1..\sigma+1]$\;
$F \gets \mathtt{zeros}[0..\sigma,1..\sigma+1]$\;
$L \gets \mathtt{zeros}[0..\sigma,1..\sigma+1]$\;
$\mathtt{gamma} \gets \mathtt{zeros}[0..\sigma]$\; 
$\mathtt{leftExtensions} \gets \mathtt{zeros}[1..\sigma+1]$\;
$\REPR{\varepsilon} \gets (\mathtt{distinctChars},\mathtt{start} \bullet (n+1))$\;
$S.\mathtt{push}\big((\REPR{\varepsilon},0)\big)$\;
\While{\Not $S.\mathtt{isEmpty}()$} {
   $(\REPR{W},|W|) \gets S.\mathtt{pop}()$\;
   $h \gets 0$\;
   $\mathtt{extendLeft}(\REPR{W},\BWT_{T},C,A,F,L,\mathtt{gamma},\mathtt{leftExtensions},h)$\;
   \HiLi$\mathtt{callback}(\REPR{W},|W|,\BWT_{T},C,A,F,L,\mathtt{gamma},\mathtt{leftExtensions},h)$\;

   \tcc{Pushing right-maximal left-extensions on the stack}
   $\mathcal{C} \gets \{ c : c=\mathtt{leftExtensions}[i], i \in [1..h], \mathtt{gamma}[c]>1 \}$\;
   \If{$\mathcal{C} \neq \emptyset$}{
		$c \gets \mbox{argmax}\{ L[c,\mathtt{gamma}[c]]-F[c,1] : c \in \mathcal{C} \}$\;
		$\REPR{cW} \gets (A[c,1..\mathtt{gamma}[c]], F[c,1..\mathtt{gamma}[c]] \bullet (L[c,\mathtt{gamma}[c]]+1))$\;
		$S.\mathtt{push}( \REPR{cW}, |W|+1 )$\;
		\For{$a \in \mathcal{C} \setminus \{c\}$}{
			$\REPR{aW} \gets (A[a,1..\mathtt{gamma}[a]], F[a,1..\mathtt{gamma}[a]] \bullet (L[a,\mathtt{gamma}[a]]+1))$\;
			$S.\mathtt{push}( \REPR{aW}, |W|+1 )$\;
		}
   }

   \tcc{Cleaning up for the next iteration}
   \For{$i \in [1..h]$}{
   		$a \gets \mathtt{leftExtensions}[i]$\;
		\For{$j \in [1..\mathtt{gamma}[a]]$}{
			$A[a,j] \gets 0$\;
			$F[a,j] \gets 0$\; 
			$L[a,j] \gets 0$\;
		}
		$\mathtt{gamma}[a] \gets 0$\;
   }
}
\caption{\label{algo:iterator}
{Enumerating all right-maximal substrings of $T \in [1..\sigma]^{n-1}\#$. See Lemma \ref{lemma:extendLeft} for a definition of operator $\bullet$. The callback function $\mathtt{callback}$ highlighted in gray just prints the pair $(\REPR{W},|W|)$ given in input. Section \ref{sec:stringAnalysis} describes other implementations of $\mathtt{callback}$.}
}
\end{algorithm}

\begin{algorithm}
\KwIn{BWT transform and $C$ array of string $T^{i}\#$. Array $\mathtt{distinctChars}^i$ of all the distinct characters that occur in $T^{i}\#$, in lexicographic order, and array $\mathtt{start}^i$ of starting positions of the corresponding intervals in $\BWT_{T^{i}\#}$. Support for $\mathtt{rangeDistinct}$ queries on $\BWT_{T^{i}\#}$, and an implementation of Algorithm \ref{algo:extendLeftPrime} (function $\mathtt{extendLeft}'$).}
\KwOut{$(\REPRPRIME{W},|W|)$ for all right-maximal substrings $W$ of $T$.}
$S \gets $ empty stack\;
\For{$i \in [1..m]$}{
	$A^i \gets \mathtt{zeros}[0..\sigma,1..\sigma+1]$, $F^i \gets \mathtt{zeros}[0..\sigma,1..\sigma+1]$\;
	$L^i \gets \mathtt{zeros}[0..\sigma,1..\sigma+1]$, $\mathtt{gamma}^i \gets \mathtt{zeros}[0..\sigma]$\; 
}
$\mathtt{leftExtensions} \gets \mathtt{zeros}[1..\sigma+m]$\;
$\mathtt{seen} \gets \mathtt{zeros}[1..\sigma]$\;
$\REPRPRIME{\varepsilon} \gets (\{\mathtt{distinctChars}^i\}, \{\mathtt{start}^i \bullet (n_{i}+1)\})$\;
$S.\mathtt{push}((\REPRPRIME{\varepsilon},0))$\;
\While{\Not $S.\mathtt{isEmpty}()$} {
   $(\REPRPRIME{W},|W|) \gets S.\mathtt{pop}()$\;
   $h \gets 0$\;
   $\mathtt{extendLeft}'(\REPRPRIME{W},\{\BWT_{T^i}\},\{C^i\},\{A^i\},\{F^i\},\{L^i\},\{\mathtt{gamma}^i\},\mathtt{leftExtensions},\mathtt{seen},h)$\;
   $\mathtt{callback}(\REPRPRIME{W},|W|,\{\BWT_{T^i}\},\{C^i\},\{A^i\},\{F^i\},\{L^i\},\{\mathtt{gamma}^i\},\mathtt{leftExtensions},h)$\;
   \tcc{Pushing right-maximal left-extensions on the stack}
   \HiLi$\mathcal{C} \gets \{ c>0 : c=\mathtt{leftExtensions}[i], i \in [1..h], (\exists \; p \in [1..m] : \mathtt{gamma}^{p}[c]>1)$ 
   \HiLi\Or $(\exists \; p \neq q : \mathtt{gamma}^{p}[c]=1, \mathtt{gamma}^{q}[c]=1, A^p[c,1] \neq A^q[c,1]) \}$\;
   \If{$\mathcal{C} \neq \emptyset$}{
		$c \gets \mbox{argmax}\left\{ \sum_{i=1}^{m} L^{i}[c,\mathtt{gamma}^i[c]]-F^{i}[c,1] : c \in \mathcal{C} \right\}$\;
		$\REPRPRIME{cW} \gets (\{A^i[c,1..\mathtt{gamma}^i[c]]\}, \{F^i[c,1..\mathtt{gamma}^i[c]] \bullet (L^i[c,\mathtt{gamma}^i[c]]+1)\})$\;
		$S.\mathtt{push}( \REPRPRIME{cW}, |W|+1 )$\;
		\For{$a \in \mathcal{C} \setminus \{c\}$}{
			$\REPRPRIME{aW} \gets ( \{A^i[a,1..\mathtt{gamma}^i[a]]\}, \{F^i[a,1..\mathtt{gamma}^i[a]] \bullet (L^i[a,\mathtt{gamma}^i[a]]+1)\})$\;
			$S.\mathtt{push}( \REPRPRIME{aW}, |W|+1 )$\;
		}
   }
   \tcc{Cleaning up for the next iteration}
   \For{$i \in [1..h]$}{
   		$a \gets \mathtt{leftExtensions}[i]$\;
		\If{$a \leq 0$}{
			$k \gets -a+2 \; (\mbox{mod}_1 \; m)$\;
			$A^k[0,1] \gets 0$, $F^k[0,1] \gets 0$, $L^k[0,1] \gets 0$, $\mathtt{gamma}^k[0] \gets 0$\;
		}
		\Else {
			$\mathtt{seen}[a] \gets 0$\;
			\For{$j \in [1..m]$}{
				\For{$k \in [1..\mathtt{gamma}^j[a]]$}{
					$A^j[a,k] \gets 0$, $F^j[a,k] \gets 0$, $L^j[a,k] \gets 0$\;
				}
				$\mathtt{gamma}^j[a] \gets 0$\;
			}
		}
   }
}
\caption{\label{algo:generalizedIterator}
{Enumerating all right-maximal substrings of $T=T^1 \#_1 T^2 \#_2 \cdots T^m \#_m$, where $T^i \in [1..\sigma]^{n_i-1}$ for $i \in [1..m]$, $m \geq 1$. The key differences from Algorithm \ref{algo:iterator} are highlighted in gray. To iterate over all \emph{impure} right-maximal substrings of $T$, it suffices to replace just the gray lines (see Lemma \ref{lemma:generalizedIterator_bwtMerge}). See Lemma \ref{lemma:extendLeft} for a definition of operator $\bullet$. The callback function $\mathtt{callback}$ just prints its input $(\REPRPRIME{W},|W|)$. For brevity the case in which a string occurs in some $T^i$ but does not occur in some $T^j$ is not handled.
}
}
\end{algorithm}

\begin{algorithm}
\KwIn{$\REPRPRIME{W}$ for a substring $W$ of $T$. Support for $\mathtt{rangeDistinct}$ queries on $\BWT_{T^{i}\#}$, $C$ array of $T^i$, empty matrices $A^i$, $F^i$ and $L^i$, and empty array $\mathtt{gamma}^i$ of string $T^i$, for all $i \in [1..m]$. A single empty array $\mathtt{leftExtensions}$, a single bitvector $\mathtt{seen}$, and a single pointer $h$.}
\KwOut{Matrices $A^i$, $F^i$, $L^i$, pointer $h$, and arrays $\mathtt{gamma}^i$ and $\mathtt{leftExtensions}$, for all $i \in [1..m]$, filled as described in Lemma \ref{lemma:generalizedIterator}.}
$(\{\mathtt{chars}^i\},\{\mathtt{first}^i\}) \gets \REPRPRIME{W}$\;
$h \gets 0$\;
\For{$i \in [1..m]$}{
	\For{$j \in [1..|\mathtt{chars}^i|]$}{
   		$\mathcal{I} \gets \BWT_{T^{i}\#}.\mathtt{rangeDistinct}(\mathtt{first}^{i}[j],\mathtt{first}^{i}[j+1]-1)$\;
		\For{$(a,p_a,q_a) \in \mathcal{I}$}{
			\If{$a=0$}{
				$h \gets h+1$\;
				$\mathtt{leftExtensions}[h] \gets -(i-1 \; (\mbox{mod}_1 \; m))+1$\;
			}
			\Else{
				\HiLi\If{$\mathtt{seen}[a]=0$}{
					\HiLi$\mathtt{seen}[a]=1$\;
					$h \gets h+1$\;
					$\mathtt{leftExtensions}[h] \gets a$\;
				}
			}
			$\mathtt{gamma}^{i}[a] \gets \mathtt{gamma}^{i}[a]+1$\;
			$A^{i}[a,\mathtt{gamma}^{i}[a]] \gets \mathtt{chars}^i[j]$\;
			$F^{i}[a,\mathtt{gamma}^{i}[a]] \gets C^{i}[a]+p_a$\;
			$L^{i}[a,\mathtt{gamma}^{i}[a]] \gets C^{i}[a]+q_a$\;
		}
   }
}
\caption{ \label{algo:extendLeftPrime}
{Building $\REPRPRIME{aW}$ from $\REPRPRIME{W}$ for all $a \in [-m+1..\sigma]$ such that $aW$ is a prefix of a rotation of $T=T^1 \#_1 T^2 \#_2 \cdots T^m \#_m$, where $m \geq 1$ and $T^i \in [1..\sigma]^{n_i-1}$. The lines highlighted in gray are the key differences from Algorithm \ref{algo:extendLeft}.}
}
\end{algorithm}


\section{Building the Burrows-Wheeler transform} \label{sec:bwtConstruction}

It is well-known that the Burrows-Wheeler transform of a string $T\#$ such that $T \in [1 \ltdots \sigma]^n$ and $\#=0 \notin [1 \ltdots \sigma]$, can be built in $O(n\log{\log{\sigma}})$ time and in $O(n\log{\sigma})$ bits of working space \cite{HSS09}. In this section we bring construction time down to $O(n)$ by plugging Theorem \ref{thm:generalizedIterator_bwtMerge} into the recursive algorithm described in \cite{HSS09}, which we summarize here for completeness.

Specifically, we partition $T$ into blocks of equals size $B$. For convenience, we work with a version of $T$ whose length is a multiple of $B$, by appending to the end of $T$ the smallest number of occurrences of character $\#$ such that the length of the resulting padded string is an integer multiple of $B$, and such that the padded string contains at least one occurrence of $\#$. Recall that $B \cdot \lceil x/B \rceil$ is the smallest multiple of $B$ that is at least $x$. Thus, we append $n'-n$ copies of character $\#$ to $T$, where $n'=B \cdot \lceil (n+1)/B \rceil$. To simplify notation we call the resulting string $X$, and we use $n'$ to denote the length of $X$.

We interpret a partitioning of $X$ into blocks as a new string $X_B$ of length $n'/B$, defined on the alphabet $[1 \ltdots (\sigma+1)^B]$ of all strings of length $B$ on alphabet $[0 \ltdots \sigma]$: the ``characters'' of $X_B$ correspond to the blocks of $X$. In other words, $X_{B}[i] = X[(i-1)B+1 \ltdots iB]$. We assume $B$ to be even, and we denote by $\mathtt{left}$ (respectively, $\mathtt{right}$) the function from $[1 \ltdots (\sigma+1)^B]$ to $[1 \ltdots (\sigma+1)^{B/2}]$ such that $\mathtt{left}(W)$ returns the first (respectively, the second) half of block $W$. In other words, if $W=w_1\cdots w_B$, $\mathtt{left}(W)=w_1\cdots w_{B/2}$ and $\mathtt{right}(W)=w_{B/2+1}\cdots w_B$. We also work with circular rotations of $X$ (see Section \ref{sect:definitions_strings}): specifically, we denote by $\overleftarrow{X}$ string $X[B/2+1 \ltdots n'] \cdot X[1 \ltdots B/2]$, or equivalently string $X$ \emph{circularly rotated to the left by $B/2$ positions}, and we denote by $\overleftarrow{X}_B$ the string on alphabet $[1 \ltdots (\sigma+1)^B]$ induced by partitioning $\overleftarrow{X}$ into blocks of size $B$.

Note that the suffix that starts at position $i$ in $\overleftarrow{X}_B$ equals the half-block $P_i=X[B/2+(i-1)B+1 \ltdots iB]$, followed by string $S_i=F_{i+1} \cdot X[1 \ltdots B/2]$, where $F_{i+1}$ is the suffix of $X_B$ that starts at position $i+1$ in $X_B$, if any. 
Thus, it is not surprising that we can derive the BWT of string $\overleftarrow{X}_B$ from the BWT of string $X_B$:

\begin{lemma}[\cite{HSS09}] \label{lemma:bwtOfLeftShift}
The BWT of string $\overleftarrow{X}_B$ can be derived from the BWT of string $X_B$ in $O(n'/B)$ time and $O(\sigma^B \cdot \log(n'/B))$ bits of working space, where $n'=|X|$.
\end{lemma}

The second key observation that we exploit for building the BWT of $X$ is the fact that the suffixes of $X_{B/2}$ which start at odd positions coincide with the suffixes of $X_B$, and the suffixes of $X_{B/2}$ that start at even positions coincide with the suffixes of $\overleftarrow{X}_B$. Thus, we can reconstruct the BWT of $X_{B/2}$ by merging the BWT of $X_B$ with the BWT of $\overleftarrow{X}_B$: this is where Theorem \ref{thm:generalizedIterator_bwtMerge} comes into play.

\begin{lemma} \label{lemma:bwtOfHalfBlock}
Assume that we can read in constant time a block of $B$ characters. Then, the BWT of string $X_{B/2}$ can be derived from the BWT of string $X_B$ and from the BWT of string $\overleftarrow{X}_B$, in $O(n'/B)$ time and $O(n'\log{\sigma})$ bits of working space, where $n'=|X|$.
\end{lemma}
\begin{proof}
All rotations of $X_B$ (respectively, of $\overleftarrow{X}_B$) are lexicographically distinct, and no rotation of $X_B$ is lexicographically identical to a rotation of $\overleftarrow{X}_B$. Thus, we can use Theorem \ref{thm:generalizedIterator_bwtMerge} to build the BWT of $\mathcal{R}(X_B) \cup \mathcal{R}(\overleftarrow{X}_B)$ in $O(n'/B)$ time and in $2n'(1+1/k)\log(\sigma+1)+20n'/B+o(n'/B) \in O(n'\log{\sigma})$ bits of working space. Inside the algorithm of Theorem \ref{thm:generalizedIterator_bwtMerge}, we apply the constant-time operator $\mathtt{right}$ to the characters of the input BWTs. There is a bijection between set $\mathcal{R}(X_B) \cup \mathcal{R}(\overleftarrow{X}_B)$ and set $\mathcal{R}(X_{B/2})$ that preserves lexicographic order, thus the BWT of $\mathcal{R}(X_{B/2})$ coincides with the BWT of $\mathcal{R}(X_B) \cup \mathcal{R}(\overleftarrow{X}_B)$ in which each character is processed with operator $\mathtt{right}$.
\end{proof}

Lemmas \ref{lemma:bwtOfLeftShift} and \ref{lemma:bwtOfHalfBlock} suggest building the BWT of $X$ in $O(\log B)$ steps, where at step $i$ we compute the BWT of string $X_{B/2^i}$, stopping when $B/2^i=1$. Note that the key requirement of Lemma \ref{lemma:bwtOfHalfBlock}, i.e. that all rotations of $X_{B/2^i}$ (respectively, of $\overleftarrow{X}_{B/2^i}$) are lexicographically distinct, and that no rotation of $X_{B/2^i}$ is lexicographically identical to a rotation of $\overleftarrow{X}_{B/2^i}$, holds for all $i$. The time for completing step $i$ is $O(n'/(B/2^i))$, and the Burrows-Wheeler transforms of $X_{B/2^i}$ and of $\overleftarrow{X}_{B/2^i}$ take $O(n' \log{\sigma})$ bits of space for every $i$.

The base case of the recursion is the BWT of string $X_B$ for some initial block size $B$: we build it using any suffix array construction algorithm that works in $O(\sigma^{B}+n'/B)$ time and in $O((n'/B)\log(n'/B))$ bits of space (for example those described in \cite{KA03,KSPP05,KSB06}). We want this first phase to take $O(n')$ time and $O(n'\log{\sigma})$ bits of space, or in other words we want to satisfy the following constraints:

\begin{enumerate}
\item $\sigma^{B} \in O(n')$ \label{constraint1}
\item $(n'/B)\log(n'/B) \in O(n'\log{\sigma})$, or more strictly $(n'/B)\log n' \in O(n'\log{\sigma})$. \label{constraint2}
\end{enumerate}

We also want $B$ to be a power of two. Recall that $2^{\lceil \log x \rceil}$ is the smallest power of two that is at least $x$. Assume thus that we set $B=2^{\lceil\log(\log n' / (c\log{\sigma}))\rceil}$ for some constant $c$. Then $B \geq \log n' / (c \log{\sigma})$, thus Constraint \ref{constraint2} is satisfied by any choice of $c$. Since $\lceil x \rceil < x+1$, we have that $B < (2/c)\log n' / \log{\sigma}$, thus Constraint \ref{constraint1} is satisfied for any $c \geq 2$. For this choice of $B$ the number of steps in the recursion becomes $O(\log\log n')$, and we can read a block of size $B$ in constant time as required by \Lemma{lemma:bwtOfHalfBlock} since the machine word is assumed to be $\Omega(\log n')$. It follows that building the BWT of $X$ takes $O(n'+(n'/B)\sum_{i=1}^{\log B}2^i) = O(n')$ time and $O(n'\log{\sigma})$ bits of working space. Since the BWT of $T\#$ can be derived from the BWT of $X$ at no extra asymptotic cost (see \cite{HSS09}), we have the following result:

\begin{theorem}\label{thm:bwtConstruction}
The BWT of a string $T\#$ such that $T \in [1 \ltdots \sigma]^n$ and $\#=0$ can be built in $O(n)$ time and in $O(n\log{\sigma})$ bits of working space.
\end{theorem}

\section{Building string indexes}\label{sec:bwtIndexesConstruction}


\subsection{Building the compressed suffix array}

The \emph{compressed suffix array} of a string $T \in [1..\sigma]^{n-1}\#$ (abbreviated to CSA in what follows) is a representation of $\SA_T$ that uses just $O((n\log{\sigma})/\epsilon)$ bits space for any given constant $\epsilon$, at the cost of increasing access time to any position of $\SA_T$ to $t = O((\log_{\sigma}{n})^{\epsilon}/\epsilon)$ \cite{GV05}. Without loss of generality, let $B$ be a block size such that $B^i$ divides $n$ for any setting of $i$ that we will consider, and let $T_i$ be the (suitably terminated) string of length $n/B^i$ defined on the alphabet $[1 \ltdots (\sigma+1)^{B^i}]$ of all strings of length $B^i$ on alphabet $[0 \ltdots \sigma]$, and such that the ``characters'' of $T_i$ correspond to the consecutive blocks of size $B^i$ of $T$. In other words, $T_{i}[j] = T[(j-1)B^i+1 \ltdots jB^i]$. Note that $T_0=T$, and $T^i$ with $i>0$ is the string obtained by grouping every consecutive $B$ characters of $T_{i-1}$. The CSA of $T$ with parameter $\epsilon$ consists of the suffix array of $T_{1/\epsilon}$, and of $1/\epsilon$ layers, where layer $i \in [0..1/\epsilon-1]$ is composed  of the following elements: 
\begin{enumerate}
\item A data structure that supports $\mathtt{access}$ and $\mathtt{partialRank}$  
operations\footnote{The CSA was originally defined in terms of the $\psi$ function: in this case, support for $\mathtt{select}$ queries would be needed.} on $\BWT_{T_i}$.
\item The $C$ array of $T_i$, defined on alphabet $[1..(\sigma+1)^{B^i}]$, encoded as a bitvector with $(\sigma+1)^{B^i}$ ones and $n/B^i$ zeros, and indexed to support $\mathtt{select}$ queries.
\item A bitvector $\mathtt{marked}_i$, of size $n/B^i$, that marks every position $j$ such that $\SA_{T_i}[j]$ is a multiple of $B$.
\end{enumerate}
For concreteness, let $\epsilon=2^{-c}$ for some constant $c > 0$.  
Note that layer $i$ contains enough information to support function $\LF$ 
on $T_i$. To compute $\SA_{T_i}[j]$, we first check whether $\mathtt{marked}_{i}[j]=1$: if so, then $\SA_{T_i}[j] = B \cdot \SA_{T_{i+1}}[j']$, where $j'=\mathtt{rank}_{1}(\mathtt{marked}_i,j)$. Otherwise, we iteratively set $j$ to $\LF(j)$ in constant time 
and we test whether $\mathtt{marked}_{i}[j]=1$. If it takes $t$ iterations to reach a $j^*$ such that $\mathtt{marked}_{i}[j^*]=1$, then $\SA_{T_i}[j] = B \cdot \SA_{T_{i+1}}[\mathtt{rank}_1(\mathtt{marked}_i,j^*)] + t$. 
Since $t \leq B-1$ at any layer, the time spent in a layer is $O(B)$, and the time to traverse all layers is $O(B/\epsilon)$. Setting $B=(\log_{\sigma}{n})^{\epsilon}$ achieves the claimed time complexity, and assuming without loss of generality that 
$\sigma$ is a power of two and $n=\sigma^{2^{2^a}}$ for some integer $a \geq c$ ensures that $B^i$ for any $i \in [1..1/\epsilon]$ is an integer that divides $n$. 
Using Lemma \ref{lemma:rank_select_access}, every layer takes $O(n\log{\sigma})$ bits of space, irrespective of $B$, so the whole data structure takes $O((n\log{\sigma})/\epsilon)$ bits of space. Counting the number $\mathtt{occ}$ of occurrences of a pattern $P$ in $T$ can be performed in a number of ways with the CSA. A simple, $O(|P|\log{n} \cdot (\log_{\sigma}{n})^{\epsilon}/\epsilon))$ time solution, consists in performing binary searches on the suffix array: this allows one to locate all such occurrences in $O(|P|(\log{n}+\mathtt{occ}) \cdot (\log_{\sigma}{n})^{\epsilon}/\epsilon))$ time. 
Alternatively, count queries could be implemented with backward steps as in the BWT index, in overall $O(|P|\log{\log{\sigma}})$ time, using Lemma \ref{lemma:rank_select_access3}.

The CSA takes in general at least $n\log{\sigma}+o(n)$ bits, or even $nH_k+o(n)$ bits for $k=o(\log_{\sigma}{n})$ \cite{GGV03}. The CSA has a number of variants, the fastest of which can be built in $O(n\log{\log{\sigma}})$ time using $O(n\log{\sigma})$ bits of working space~\cite{HSS09}. 
%
%
%
%
Combining the setup of data structures described above with Theorem \ref{thm:bwtConstruction} 
allows one to build the CSA more efficiently:

\begin{theorem} \label{thm:csaConstruction}
Given a string $T = [1..\sigma]^{n}$, we can build the compressed suffix array in $O(n)$ time and in $O(n\log{\sigma})$ bits of working space.
\end{theorem}

Note that the BWT of all strings $T_i$ in the CSA of $T$, as well as all bitvectors $\mathtt{marked}_i$, can be built \emph{in a single invocation} of Theorem \ref{thm:bwtConstruction}, rather than by invoking Theorem \ref{thm:bwtConstruction} $1/\epsilon$ times. 
Note also that combining Theorem \ref{thm:bwtConstruction} with Lemma \ref{lemma:buildingSuccinctSA} and with the first data structure of Lemma \ref{lemma:rank_select_access3}, yields immediately a BWT index and a succinct suffix array that can be built in deterministic linear time: 


\begin{theorem} \label{thm:bwtIndexConstruction0}
Given a string $T = [1..\sigma]^{n}$, we can build the following data structures:
\begin{itemize}
\item A BWT index that takes $n\log{\sigma}(1+1/k)+O(n\log{\log{\sigma}})$ bits of space for any positive integer $k$, and that implements operation $\LF(i)$ in constant time for any $i \in [1..n]$, and operation $\mathtt{count}(P)$ in $O(m(\log{\log{\sigma}}+k))$ time for any $P \in [1..\sigma]^m$. The index can be built in $O(n)$ time and in $O(n\log{\sigma})$ bits of working space.

\item A succinct suffix array that takes $n\log{\sigma}(1+1/k)+O(n\log{\log{\sigma}})+O((n/r)\log{n})$ bits of space for any positive integers $k$ and $r$, and that implements operation $\mathtt{count}(P)$ in $O(m(\log{\log{\sigma}}+k))$ time for any $P \in [1..\sigma]^m$, operation $\mathtt{locate}(i)$ in $O(r)$ time, and operation $\mathtt{substring}(i,j)$ in $O(j-i+r)$ time for any $i<j$ in $[1..n]$. The index can be built in $O(n)$ time and in $O(n\log{\sigma})$ bits of working space.

\end{itemize}
\end{theorem}

\subsection{Building BWT indexes}

To reduce the time complexity of a backward step to a constant, however, we need to augment the representation of the topology of $\ST_{T}$ described in Lemma \ref{lemma:balancedParentheses} with an additional operation, where $T = [1..\sigma]^{n-1}\#$. Recall that the identifier $\mathtt{id}(v)$ of a node $v$ of $\ST_T$ is the rank of $v$ in the preorder traversal of $\ST_T$. Given a node $v$ of $\ST_T$ and a character $a \in [0..\sigma]$, let operation $\mathtt{weinerLink}(\mathtt{id}(v),a)$ return zero if string $a\ell(v)$ is not the prefix of a rotation of $T$, and return $\mathtt{id}(w)$ otherwise, where $w$ is the locus of string $a\ell(v)$ in $\ST_T$. The following lemma describes how to answer $\mathtt{weinerLink}$ queries efficiently:

\begin{lemma} \label{lemma:weiner_link_support}
Assume that we are given  a data structure that supports $\mathtt{access}$ queries on the BWT of a string $T = [1..\sigma]^{n-1}\#$ in constant time, a data structure that supports $\mathtt{rangeDistinct}$ queries on $\BWT_T$ in constant time per element in the output, and a data structure that supports $\mathtt{select}$ queries on $\BWT_T$ in time $t$. Assume also that we are given the representation of the topology of $\ST_T$ described in Lemma \ref{lemma:balancedParentheses}. Then, we can build a data structure that takes $O(n\log{\log{\sigma}})$ bits of space and that supports operation $\mathtt{weinerLink}(\mathtt{id}(v),a)$ in $O(t)$ time for any node $v$ of $\ST_T$ (including leaves) and for any character $a \in [0..\sigma]$. This data structure can be built in \emph{randomized} $O(nt)$ time and in $O(n\log{\sigma})$ bits of working space.
\end{lemma}
\begin{proof}
We show how to build efficiently the data structure described in \cite{BNtalg14}, which we summarize here for completeness. We use the suffix tree topology to convert in constant time $\mathtt{id}(v)$ to $\INTERVAL{v}$ (using operations $\mathtt{leftmostLeaf}$ and $\mathtt{rightmostLeaf}$), and vice versa (using operations $\mathtt{selectLeaf}$ and $\mathtt{lca}$). We traverse $\ST_T$ in preorder using the suffix tree topology, as described in Lemma \ref{lemma:spaceEfficientPreorder}. For every internal node $v$ of $\ST_T$, we use a $\mathtt{rangeDistinct}$ query to compute all the $h$ distinct characters $a_1,\dots,a_h$ that appear in $\BWT_{T}[\INTERVAL{v}]$, and for every such character the interval of $a_i \ell(v)$ in $\BWT_{T}$, in overall $O(h)$ time. Note that the sequence $a_1,\dots,a_h$ returned by a $\mathtt{rangeDistinct}$ query is not necessarily sorted in lexicographic order. We determine whether $a_i \ell(v)$ is the label of a node $w$ of $\ST_T$ by taking a suffix link from the locus of $a_i \ell(v)$ in $O(t)$ time, using Lemma \ref{lemma:suffixLink}, and by checking whether the destination of such link is indeed $v$.

For every character $c \in [0..\sigma]$, we use vector $\mathtt{sources}^c$ to store all nodes $v$ of $\ST_T$ (including leaves) that are the source of an implicit or explicit Weiner link labeled by $c$, in the order induced by the preorder traversal of $\ST_T$. We encode the difference between the preorder ranks of two consecutive nodes in the same $\mathtt{sources}^c$ using Elias delta or gamma coding~\cite{El75}. We also store a bitvector $\mathtt{explicit}^c$ that marks with a one every explicit Weiner link in $\mathtt{sources}^c$ (recall that Weiner links from leaves are explicit). Bitvectors $\mathtt{sources}^c$ and $\mathtt{explicit}^c$ can be filled during the preorder traversal of $\ST_T$. Once $\mathtt{explicit}^c$ has been filled, we index it to answer $\mathtt{rank}$ queries. The space used by such indexed bivectors $\mathtt{explicit}^c$ for all $c \in [0..\sigma]$ is $O(n)$ bits by Observation \ref{obs:suffixtree}, and the space used by vectors $\mathtt{sources}^c$ for all $c \in [0..\sigma]$ is $O(n\log{\sigma})$ bits, by applying Jensen's inequality twice as in Lemma \ref{lemma:rangedistinct2}. We follow the static allocation strategy described in Section \ref{sec:staticAllocation}: specifically, we compute the number of bits needed by $\mathtt{sources}^c$ and $\mathtt{explicit}^c$ during a preliminary pass over $\ST_T$, in which we increment the size of the arrays by keeping the preorder position of the last internal node with a Weiner link labeled by $c$, for all $c \in [0..\sigma]$. This preprocessing takes $O(n)$ time and $O(\sigma\log{n}) \in o(n)$ bits of space. Once such sizes are known, we allocate a large enough contiguous region of memory.

Finally, we build an array $C'[1..\sigma]$ where $C'[a]$ is the number of nodes $v$ in $\ST_T$ (including leaves) such that $\ell(v)$ starts with a character strictly smaller than $a$. We also build an implementation of an MMPHF $f^c$ for every $\mathtt{sources}^c$ using the technique described in the proof of Lemma \ref{lemma:rangedistinct2}, and we discard $\mathtt{sources}^c$. All such MMPHF implementations take $O(n\log{\log{\sigma}})$ bits of space, and they can be built in overall $O(n)$ randomized time and in $O(\sigma^k \log{\sigma})$ bits of working space, for any integer $k>1$. Note that $C'$ takes $O(\sigma\log{n}) \in o(n)$ bits of space, since $\sigma \in o(\sqrt{n}/\log{n})$, and it can be built with a linear-time preorder traversal of $\ST_T$.

Given a node $v$ of $\ST_T$ and a character $c \in [0..\sigma]$, we determine whether the Weiner link from $v$ labeled by $c$ is explicit or implicit by accessing $\mathtt{explicit}^{c}(f^{c}(\mathtt{id}(v)))$, and we compute the identifier of the locus $w$ of the destination of the Weiner link (which might be a leaf) by computing:
$$
C'[c]+\mathtt{rank}_{1}(\mathtt{explicit}^{c},f^{c}(\mathtt{id}(v))-1)+1
$$
If there is no Weiner link from $v$ labeled by $c$, then $v$ does not belong to $\mathtt{sources}^c$, but $f^{c}(\mathtt{id}(v))$ still returns a valid pointer in $\mathtt{sources}^c$: to check whether this pointer corresponds to $v$, we convert $v$ and $w$ to intervals in $\BWT_T$ using the suffix tree topology, and we check whether $\mathtt{select}_{c}(\BWT_{T},\SP{w}-C[c]) \in \INTERVAL{v}$.

The output of this construction consists in arrays $C'$, $\mathtt{explicit}^c$, and in the implementation of $f^c$, for all $c \in [0..\sigma]$.
\end{proof}

Since operation $\mathtt{weinerLink}(\mathtt{id}(v),a)$ coincides with a backward step with character $a$ from $\INTERVAL{v}$ in $\BWT_T$, Lemma \ref{lemma:weiner_link_support} enables the construction of space-efficient BWT indexes with constant-time $\LF$:

\begin{theorem}\label{thm:bwtIndexConstruction}
Given a string $T = [1..\sigma]^{n-1}\#$, we can build any of the following data structures in randomized $O(n)$ time and in $O(n\log{\sigma})$ bits of working space:
\begin{itemize}
\item A BWT index that takes $n\log{\sigma}(1+1/k) + O(n\log{\log{\sigma}})$ bits of space for any positive integer $k$, and that implements operation $\LF(i)$ in constant time for any $i \in [1..n]$, and operation $\mathtt{count}(P)$ in $O(mk)$ time for any $P \in [1..\sigma]^m$.
\item A succinct suffix array that takes $n\log{\sigma}(1+1/k) + O(n\log{\log{\sigma}}) + O((n/r)\log{n})$ bits of space for any positive integers $k$ and $r$, and that implements operation $\mathtt{count}(P)$ in $O(mk)$ time for any $P \in [1..\sigma]^m$, operation $\mathtt{locate}(i)$ in $O(r)$ time, and operation $\mathtt{substring}(i,j)$ in $O(j-i+r)$ time for any $i<j$ in $[1..n]$.
\end{itemize}
Alternatively, for the same construction space and time, we can build analogous data structures that support $\LF(i)$ in $O(k)$ time, $\mathtt{count}(P)$ in $O(m)$ time, $\mathtt{locate}(i)$ in $O(r)$ time, and $\mathtt{substring}(i,j)$ in $O(j-i+r)$ time: such data structures take the same space as those described above.
\end{theorem}
\begin{proof}
In this proof we combine a number of results described earlier in the paper: see Figure \ref{fig:architecture} for a summary of their mutual dependencies.

We use Theorem \ref{thm:bwtConstruction} to build $\BWT_T$ from $T$, and Lemma \ref{lemma:rank_select_access} to build a data structure that supports $\mathtt{access}$, $\mathtt{partialRank}$ and $\mathtt{select}$ queries on $\BWT_T$. Then, we discard $\BWT_T$. Together with the $C$ array of $T$, this is already enough to implement function $\LF$ and to build arrays $\mathtt{samples}$ and $\mathtt{pos2rank}$ for the succinct suffix array, using Lemma \ref{lemma:buildingSuccinctSA}. We either use the data structure of Lemma \ref{lemma:rank_select_access} that supports select queries in $O(k)$ time (in which case we implement locate and substring queries with function $\LF$), or the data structure that supports select queries in constant time (in which case we implement locate and substring queries with function $\psi$).

To implement backward steps we need support for $\mathtt{weinerLink}$ operations from internal nodes of $\ST_T$. We use Lemma \ref{lemma:rangedistinct} to build a $\mathtt{rangeDistinct}$ data structure on $\BWT_T$ from the $\mathtt{access}$, $\mathtt{partialRank}$ and $\mathtt{select}$ data structure built by Lemma \ref{lemma:rank_select_access}. We use $\mathtt{rangeDistinct}$ queries inside the algorithm to enumerate the BWT intervals of all internal nodes of $\ST_T$ described in Theorem \ref{thm:iterator}, and we use such algorithm to build the balanced parentheses representation of $\ST_T$ as described in Theorem \ref{thm:topology}. To support operations on the topology of $\ST_T$, we feed the balanced parentheses representation of $\ST_T$ to Lemma \ref{lemma:balancedParentheses}. Finally, we use the $\mathtt{rangeDistinct}$ data structure, the tree topology, and the support for $\mathtt{access}$, $\mathtt{partialRank}$ and $\mathtt{select}$ queries on $\BWT_T$, to build the data structures that support $\mathtt{weinerLink}$ operations described in Lemma \ref{lemma:weiner_link_support}. At the end of this process, we discard the $\mathtt{rangeDistinct}$ data structure.

The output of this construction consists of the data structures that support $\mathtt{access}$, $\mathtt{partialRank}$, and $\mathtt{select}$ on $\BWT_T$, and $\mathtt{weinerLink}$ on $\ST_T$.
\end{proof}

\subsection{Building the bidirectional BWT index \label{sect:biBWT}}

The BWT index can be made \emph{bidirectional}, in the sense that it can be adapted to support both left and right extension by a single character \cite{Bidirectional_search_in_a_string_with_wavelet_trees,Bidirectional_search_in_a_string_with_wavelet_trees_and_bidirectional_matching_statistics}. In addition to having a number of applications in high-throughput sequencing (see e.g.  \cite{LLTWWY09,SOAP2_An_improved_ultrafast_tool_for_short_read_alignment}), this index can be used to implement a number of string analysis algorithms, and as an intermediate step for building the compressed suffix tree.

Given a string $T=t_1 t_2 \cdots t_{n-1}$ on alphabet $[1 \ltdots \sigma]$, consider two BWT transforms, one built on $T\#$ and one built on $\REV{T}\# = t_n t_{n-1} \cdots t_1\#$. Let $\INTERVALFUNCTION(W,T)$ be the function that returns the interval in $\BWT_{T\#}$ of the suffixes of $T\#$ that are prefixed by string $W \in [1 \ltdots \sigma]^+$. Note that interval $\INTERVALFUNCTION(W,T)$ in the \emph{suffix array} of $T\#$ contains all the starting positions of string $W$ in $T$. Symmetrically, interval $\INTERVALFUNCTION(\REV{W},\REV{T})$ in the suffix array of $\REV{T}\#$ contains all those positions $i$ such that $n-i+1$ is an \emph{ending position} of string $W$ in $T$.

\begin{definition}
\label{def:biBWTindex}
Given a string $T \in [1 \ltdots \sigma]^{n-1}$, a \emph{bidirectional BWT index} on $T$ is a data structure that supports the following operations on pairs of integers $1 \leq i \leq j \leq n$ and on substrings $W$ of $T$:
\begin{itemize}
\item $\mathtt{isLeftMaximal}(i,j)$: returns $1$ if substring $\BWT_{T\#}[i \ltdots j]$ contains at least two distinct characters, and $0$ otherwise.
\item $\mathtt{isRightMaximal}(i,j)$: returns $1$ if substring $\BWT_{\REV{T}\#}[i \ltdots j]$ contains at least two distinct characters, and $0$ otherwise.
\item $\mathtt{enumerateLeft}(i,j)$: returns all the distinct characters that appear in substring $\BWT_{T\#}[i \ltdots j]$, \emph{in lexicographic order}.
\item $\mathtt{enumerateRight}(i,j)$: returns all the distinct characters that appear in $\BWT_{\REV{T}\#}[i \ltdots j]$, \emph{in lexicographic order}.
\item $\mathtt{extendLeft}\left(c,\INTERVALFUNCTION(W,T),\INTERVALFUNCTION(\REV{W},\REV{T})\right)$: returns pair $\left(\INTERVALFUNCTION(cW,T),\INTERVALFUNCTION(\REV{W}c,\REV{T})\right)$ for $c \in [0 \ltdots \sigma]$.
\item $\mathtt{extendRight}\left(c,\INTERVALFUNCTION(W,T),\INTERVALFUNCTION(\REV{W},\REV{T})\right)$: returns $\left(\INTERVALFUNCTION(Wc,T),\INTERVALFUNCTION(c\REV{W},\REV{T})\right)$ for $c \in [0 \ltdots \sigma]$.
\item $\mathtt{contractLeft}\left(\INTERVALFUNCTION(aW,T),\INTERVALFUNCTION(\REV{W}a,\REV{T})\right)$, where $a \in [1..\sigma]$ and $aW$ is right-maximal: returns pair $\left(\INTERVALFUNCTION(W,T),\INTERVALFUNCTION(\REV{W},\REV{T})\right)$;
\item $\mathtt{contractRight}\left(\INTERVALFUNCTION(Wb,T),\INTERVALFUNCTION(b\REV{W},\REV{T})\right)$, where $b \in [1..\sigma]$ and $Wb$ is left-maximal: returns pair $\left(\INTERVALFUNCTION(W,T),\INTERVALFUNCTION(\REV{W},\REV{T})\right)$.
\end{itemize}
\end{definition}

Operations $\mathtt{extendLeft}$ and $\mathtt{extendRight}$ are analogous to a standard backward step in $\BWT_{T\#}$ or $\BWT_{\REV{T}\#}$, but they keep the interval of a string $W$ in one BWT \emph{synchronized} with the interval of its reverse $\REV{W}$ in the other BWT.
 
In order to build a bidirectional BWT index on string $T$, we also need to support operation $\mathtt{countSmaller}(\INTERVAL{v},c)$, which returns the number of occurrences of characters smaller than $c$ in $\BWT_{T\#}[\INTERVAL{v}]$, where $v$ is a node of $\ST_{T\#}$ and $c$ is the label of an explicit or implicit Weiner link from $v$. Note that, when $v$ is a leaf of $\ST_{T\#}$, $\mathtt{countSmaller}(\INTERVAL{v},\BWT_{T\#}[x])=0$, where $x=\SP{v}=\EP{v}$. The construction of Lemma \ref{lemma:weiner_link_support} can be extended to support constant-time $\mathtt{countSmaller}$ queries, as described in the following lemma:

\begin{lemma} \label{lemma:bidirectional_weiner_link_support}
Assume that we are given a data structure that supports $\mathtt{access}$ queries on the BWT of a string $T = [1..\sigma]^{n-1}\#$ in constant time, a data structure that supports $\mathtt{rangeDistinct}$ queries on $\BWT_T$ in constant time per element in the output, and a data structure that supports $\mathtt{select}$ queries on $\BWT_T$ in time $t$. Assume also that we are given the representation of the topology of $\ST_T$ described in Lemma \ref{lemma:balancedParentheses}. Then, we can build a data structure that takes $3n\log{\sigma}+O(n\log{\log{\sigma}})$ bits of space, and that supports operation $\mathtt{weinerLink}(\mathtt{id}(v),a)$ in $O(t)$ time for any node $v$ of $\ST_T$ (including leaves) and for any character $a \in [0..\sigma]$, and operation $\mathtt{countSmaller}(\INTERVAL{v},c)$ in constant time for any internal node $v$ of $\ST_T$ and for any character $a \in [1..\sigma]$ that labels a Weiner link from $v$. This data structure can be built in \emph{randomized} $O(nt)$ time and in $O(n\log{\sigma})$ bits of working space.
\end{lemma}
\begin{proof}
We run the algorithm described in the proof of Lemma \ref{lemma:weiner_link_support}. Specifically, we traverse $\ST_T$ in preorder, we print arrays $\mathtt{sources}^c$ for all $c \in [0..\sigma]$, and we build the implementation of an MMPHF $f^c$ for every $\mathtt{sources}^c$. Before discarding $\mathtt{sources}^c$, we build the prefix-sum data structure of Lemma \ref{lemma:prefixSums} on every $\mathtt{sources}^c$ with $c>0$: by Jensen's inequality and Observation \ref{obs:suffixtree}, all such data structures take at most $3n\log{\sigma}+6n+o(n)$ bits of space in total.

Let $\ST^c = (V^c,E^c)$ be the contraction of $\ST_T$ induced by all the $n^c$ nodes (including leaves) that have an explicit or implicit Weiner link labeled by character $c \in [1..\sigma]$. During the preorder traversal of $\ST_T$, we also concatenate to a bitvector $\mathtt{parentheses}^c$ an open parenthesis every time we visit an internal node $v$ with a Weiner link labeled by character $c$ from its parent, and a closed parenthesis every time we visit $v$ from its last child. Note that $\mathtt{parentheses}^c$ represents the topology of $\ST^c$. By Observation \ref{obs:suffixtree}, building all bitvectors $\mathtt{parentheses}^c$ takes $O(n)$ time and $6n+o(n)$ bits of space in total, since every pair of corresponding parentheses can be charged to an explicit or implicit Weiner link of $\ST_T$. We feed $\mathtt{parentheses}^c$ to Lemma \ref{lemma:balancedParentheses} to obtain support for tree operations, and we discard $\mathtt{parentheses}^c$. Following the strategy described in Section \ref{sec:staticAllocation}, we preallocate the space required by $\mathtt{sources}^c$ and $\mathtt{parentheses}^c$ for all $c \in [1..\sigma]$ during a preliminary pass over $\ST_T$. Note that the preorder rank in $\ST^c$ of a node $v$, that we denote by $\mathtt{id}^{c}(v)$, equals its position in array $\mathtt{sources}^c$. Note also that the set of $\mathtt{id}^{c}(w)$ values for all the descendants $w$ of $v$ in $\ST^c$, including $v$ itself, forms a contiguous range.

We allocate $\sigma$ empty arrays $\mathtt{diff}^{c}[1..n^c]$ which, at the end of the algorithm, will contain the following information:
\begin{eqnarray*}
\mathtt{diff}^{c}[\mathtt{id}^{c}(v)] & = & \mathtt{countSmaller}\big( \INTERVAL{v},c \big) \\
& & - \sum_{(v,w) \in E_c}\mathtt{countSmaller}\big( \INTERVAL{w},c \big)
\end{eqnarray*}
i.e. $\mathtt{diff}^{c}[k]$ will encode the difference between the number of characters smaller than $c$ in the BWT interval of the node $v$ of $\ST_T$ that is mapped to position $k$ in $\mathtt{sources}^c$, and the number of characters smaller than $c$ in the BWT intervals of all the descendants of $v$ in the contracted suffix tree $\ST^c$. To compute $\mathtt{countSmaller}(\INTERVAL{v},c)$ for some internal node $v$ of $\ST_T$, we proceed as follows. We use the implementation of the MMPHF $f^c$ built on $\mathtt{sources}^c$ to compute $\mathtt{id}^{c}(v)$, we retrieve the smallest and the largest $\mathtt{id}^{c}(w)$ value assumed by a descendant $w$ of $v$ in $\ST^c$ using operations $\mathtt{leftmostLeaf}$ and $\mathtt{rightmostLeaf}$ provided by the topology of $\ST^c$, and we sum $\mathtt{diff}^{c}[k]$ for all $k$ in this range. We compute this sum in constant time by encoding $\mathtt{diff}^{c}$ with the prefix-sum data structure described in Lemma \ref{lemma:prefixSums}. Since $\sum_{k=1}^{n^c}\mathtt{diff}^{c}[k] \leq n$, the total space taken by all such prefix-sum data structures is at most $3n\log{\sigma}+6n+o(n)$ bits, by Observation \ref{obs:suffixtree} and Jensen's inequality.

To build the $\mathtt{diff}^{c}$ arrays, we scan the sequence of all characters $c_1 < c_2 < \cdots < c_k$ such that $c_i \in [1..\sigma]$ and $\ST^{c_i}$ has at least one node, for all $i \in [1..k]$. We use a temporary vector $\mathtt{lastChar}$ with one element per node of $\ST_T$: after having processed character $c_i$, $\mathtt{lastChar}[\mathtt{id}(v)]$ stores the largest $c_j \leq c_i$ that labels a Weiner link from $v$. We also assume to be able to answer $\mathtt{countSmaller}(\INTERVAL{v},c)$ queries in constant time. Note that $\mathtt{lastChar}$ takes at most $(2n-1)\log{\sigma}$ bits of space. We process character $c_i$ as follows. We traverse $\ST^{c_i}$ in preorder using its topology, as described in Lemma \ref{lemma:spaceEfficientPreorder}. For each node $v$ of $\ST^{c_i}$, we use $\mathtt{id}^{c_i}(v)$ and the prefix-sum data structure on $\mathtt{sources}^{c_i}$ to compute $\mathtt{id}(v)$. If $v$ is an internal node of $\ST_T$, we use $\mathtt{id}(v)$ to access $b=\mathtt{lastChar}[\mathtt{id}(v)]$. We compute the number of occurrences of character $b$ in $\INTERVAL{v}$ using the $O(t)$-time operation $\mathtt{weinerLink}(\mathtt{id}(v),b)$, and we compute the number of occurrences of characters smaller than $b$ in $\INTERVAL{v}$ using the constant-time operation $\mathtt{countSmaller}(\INTERVAL{v},b)$. We do the same for all children of $v$ in $\ST^{c_i}$, which we can access using the topology of $\ST^{c_i}$. Finally, we sum the values of all children and we subtract this sum from the value of $v$, appending the result to the end of $\mathtt{diff}^{c_i}$ using Elias delta or gamma coding. Finally, we set $\mathtt{lastChar}[\mathtt{id}(v)] = c_i$. The total number of accesses to a node $v$ of $\ST_T$ is a constant multiplied by the number of Weiner links from $v$, thus the algorithm runs in $O(nt)$ time.

The output of the construction consists in the topology of $\ST^{c_i}$ for all $i \in [1..k]$, in arrays $C'$ and $\mathtt{explicit}^c$ of Lemma \ref{lemma:weiner_link_support} for all $c \in [0..\sigma]$, in the implementation of $f^c$ for all $c \in [0..\sigma]$, and in the prefix-sum data structure on $\mathtt{diff}^{c_i}$ for all $i \in [1..k]$.
\end{proof}

Lemma \ref{lemma:bidirectional_weiner_link_support} immediately yields the following result:

\begin{theorem} \label{thm:bidirectionalIndexConstruction}
Given a string $T = [1..\sigma]^n$, we can build in randomized $O(n)$ time and in $O(n\log{\sigma})$ bits of working space a bidirectional BWT index that takes $O(n\log{\sigma})$ bits of space and that implements every operation in time linear in the size of its output.
\end{theorem}
\begin{proof}
Let $W$ be a substring of $T$ such that $\INTERVALFUNCTION(W,\BWT_{T\#})=[i..j]$ and $\INTERVALFUNCTION(\REV{W},\BWT_{\REV{T}\#})=[i'..j']$, let $v$ be the node of $\ST_{T\#}$ such that $\INTERVALFUNCTION(v,\BWT_{T\#})=[i..j]$ and let $v'$ be the node of $\ST_{\REV{T}\#}$ such that $\INTERVALFUNCTION(v',\BWT_{\REV{T}\#})=[i'..j']$. We plug the $\mathtt{countSmaller}$ support provided by Lemma \ref{lemma:bidirectional_weiner_link_support} in the construction of the BWT index described in Theorem \ref{thm:bwtIndexConstruction}, and we build the corresponding data structures on both $\BWT_{T\#}$ and $\BWT_{\REV{T}\#}$.

Operation $\mathtt{extendLeft}\big( a,(i,j),(i',j') \big) = \big( (p,q),(p',q') \big)$ can be implemented as follows: we compute $(p,q)$ using $\mathtt{weinerLink}(\mathtt{id}(v),a)$, and we set $(p',q') = \big( i'+ \mathtt{countSmaller}(i,j,a), i'+ \mathtt{countSmaller}(i,j,a) + q-p \big)$.

To support $\mathtt{isLeftMaximal}$ we build a bitvector $\mathtt{runs}[2 \ltdots n+1]$ such that $\mathtt{runs}[i]=1$ if and only if $\BWT_{T\#}[i] \neq \BWT_{T\#}[i-1]$. We build this vector by a linear scan of $\BWT_{T\#}$, and we index it to support $\mathtt{rank}$ queries in constant time. We implement $\mathtt{isLeftMaximal}(i,j)$ by checking whether there is a one in $\mathtt{runs}[i+1 \ltdots j]$, i.e. whether $\mathtt{rank}_{1}(\mathtt{runs},j)-\mathtt{rank}_{1}(\mathtt{runs},i) \geq 1$. This technique was already described in e.g. \cite{kulekci2012efficient,okanohara2009text}.

Assuming that $W$ is right-maximal, we support $\mathtt{contractLeft}\big( (i,j),(i',j') \big) = \big( (p,q),(p',q') \big)$ as follows. Let $W=aV$ for some $a \in [0..\sigma]$ and $V \in [1..\sigma]^*$. We compute $(p,q)=\INTERVALFUNCTION(V,\BWT_{T\#})$ using operation $\mathtt{suffixLink}(\mathtt{id}(v))$ described in Lemma \ref{lemma:suffixLink}, and we check the result of operation $\mathtt{isLeftMaximal}(p,q)$: if $V$ is not left-maximal, then $(p',q')=(i',j')$, otherwise $\REV{V}$ is the label of an internal node of $\ST_{\REV{T}\#}$, and this node is the parent of $v'$.

To implement $\mathtt{enumerateLeft}(i,j)$, we first check whether $\mathtt{isLeftMaximal}(i,j)$ returns true: otherwise, there is just character $\BWT_{T\#}[i]$ to the left of $W$ in $T\#$. Recall that operation $\mathtt{rangeDistinct}(i,j)$ on $\BWT_{T\#}$ returns the distinct characters that occur in $\BWT_{T\#}[i..j]$ as a sequence $a_1,\dots,a_h$ which is not necessarily sorted lexicographically. Note that characters $a_1,\dots,a_h$ are precisely the distinct right-extensions of string $\REV{W}$ in $\REV{T}\#$: since $W$ is left-maximal, we have that $\REV{W}=\ell(v')$, and $a_1,\dots,a_h$ are the labels associated with the children of $v'$ in $\ST_{\REV{T}\#}$. Thus, if we had an MMPHF $f^{v'}$ that maps $a_1,\dots,a_h$ to their rank among the children of $v'$ in $\ST_{\REV{T}\#}$, we could sort the output of $\mathtt{rangeDistinct}(i,j)$ in linear time. We can build the implementation of $f^{v'}$ for all internal nodes $v'$ of $\ST_{\REV{T}\#}$ using the enumeration algorithm described in Theorem \ref{thm:iterator}, and by applying to array $\mathtt{chars}$ of $\mathtt{repr}(\ell(v'))$ the implementation of the MMPHF described in Lemma \ref{lemma:build_mmphf1}. Since every character in every $\mathtt{chars}$ array can be charged to a distinct node of $\ST_{\REV{T}\#}$, the set of all such MMPHF implementations takes $O(n\log{\log{\sigma}})$ bits of space, and building it takes randomized $O(n)$ time and $O(\sigma\log{\sigma})$ bits of working space. Operation $\mathtt{enumerateLeft}$ can be combined with $\mathtt{extendLeft}$ to return intervals in addition to distinct characters.

We support $\mathtt{enumerateRight}$, $\mathtt{isRightMaximal}$, $\mathtt{contractRight}$ and $\mathtt{extendRight}$ symmetrically.
\end{proof}

Before describing the construction of other indexes, we note that the constant-time $\mathtt{countSmaller}$ support of \Lemma{lemma:bidirectional_weiner_link_support}, combined with the enumeration algorithm of \Lemma{lemma:iterator}, enables an efficient way of building $\BWT_{\REV{T}\#}$ from $\BWT_{T\#}$:

\begin{lemma} \label{lemma:revBWT}
Let $T \in [1..\sigma]^n$ be a string. Given $\BWT_{T\#}$, indexed to support $\mathtt{rangeDistinct}$ queries in constant time per element in their output, and $\mathtt{countSmaller}$ queries in constant time, we can build $\BWT_{\REV{T}\#}$ in $O(n)$ time and in $O(\sigma^{2}\log^{2}{n})$ bits of working space, and we can build $\BWT_{\REV{T}\#}$ \emph{from left to right}, in $O(n)$ time and in $O(\lambda_{T} \cdot \sigma^{2}\log{n})$ bits of working space, where $\lambda_T$ is defined in \Sec{sect:definitions_strings}.
\end{lemma}
\begin{proof}
We use \Lemma{lemma:iterator} to iterate over all right-maximal substrings $W$ of $T$, and we use $\mathtt{countSmaller}$ queries to keep at every step, in addition to $\REPR{W}$, the interval of $\REV{W}$ in $\BWT_{\REV{T}\#}$, as described in \Thm{thm:bidirectionalIndexConstruction}.

Let $a \in [1..\sigma]$, let $\INTERVALFUNCTION(\REV{W},\BWT_{\REV{T}\#}) = [i..j]$, and let $\INTERVALFUNCTION(\REV{aW},\BWT_{\REV{T}\#}) = [i'..j']$. Recall that $[i'..j'] \subseteq [i..j]$, and that we can test whether $aW$ is right-maximal by checking whether $\mathtt{gamma}[a]>1$ in \Lemma{lemma:extendLeft}. If $aW$ is not right-maximal, i.e. if the Weiner link labelled by $a$ from the locus of $W$ in $\ST_{T\#}$ is implicit, then $\BWT_{\REV{T}\#}[i'..j']$ is a run of character $A[a][1]$, where $A$ is the matrix used in \Lemma{lemma:extendLeft}. If $aW$ is right-maximal, then it will be processed in the same way as $W$ during the iteration, and its corresponding interval $[i'..j']$ in $\BWT_{\REV{T}\#}$ will be recursively filled.

To build $\BWT_{\REV{T}\#}$ from left to right, it suffices to replace the traversal strategy of \Lemma{lemma:iterator}, based on the logarithmic stack technique, with a traversal based on the lexicographic order of the left-extensions of every right-maximal substring. This makes the depth of the traversal stack of \Lemma{lemma:iterator} become $O(\lambda_T)$.
\end{proof}

Contrary to the algorithm described in \cite{OhlebuschBA14}, Lemma \ref{lemma:revBWT} does not need $T$ and $\SA_{T\#}$ in addition to $\BWT_{T\#}$. 

We also note that a fast bidirectional BWT index, such as the one in \Thm{thm:bidirectionalIndexConstruction}, enables a number of applications, which we will describe in more detail in Section \ref{sec:stringAnalysis}. For example, we can enumerate all the right-maximal substrings of $T$ as in Section \ref{sec:enumeration}, but with the additional advantage of providing access to their left extensions in lexicographic order:

\begin{lemma}
Given the bidirectional BWT index of $T \in [1..\sigma]^n$ described in Theorem \ref{thm:bidirectionalIndexConstruction}, there is an algorithm that solves Problem \ref{problem:iterator} in $O(n)$ time, and in $O(\sigma\log^{2}{n})$ bits of working space and $O(\sigma^{2}\log{n})$ bits of temporary space, where the sequence $a_1,\dots,a_h$ of left-extensions of every right-maximal string $W$ is in lexicographic order.
\end{lemma}
\begin{proof}
By adapting Lemma \ref{lemma:iterator} to use operations $\mathtt{enumerateLeft}$, $\mathtt{extendLeft}$ and $\mathtt{isRightMaximal}$ provided by the bidirectional BWT index. The smaller working space with respect to Lemma \ref{lemma:iterator} derives from the fact that the representation of a string $W$ is now the constant-space pair of intervals $\big( \INTERVALFUNCTION(W,T\#),\INTERVALFUNCTION(\REV{W},\REV{T}\#) \big)$.
\end{proof}

\subsection{Building the permuted LCP array}

We can use the bidirectional BWT index to compute the permuted LCP array as well:

\begin{lemma}\label{lemma:plcpConstruction}
Given the bidirectional BWT index of $T \in [1..\sigma]^n$ described in Theorem \ref{thm:bidirectionalIndexConstruction}, we can build $\PLCP_{T\#}$ in $O(n)$ time and in $O(\log n)$ bits of working space.
\end{lemma}
\begin{proof}
We scan $T'=T\#$ from left to right. By inverting $\BWT_{\REV{T}\#}$, we know in constant time the position $r_i$ in $\BWT_{T\#}$ that corresponds to every position $i$ in $T\#$. Assume that we know $\PLCP[i]$ and the interval of $aW=T[i..i+\PLCP[i]-1]$ in $\BWT_{T\#}$ and in $\BWT_{\REV{T}\#}$, where $a \in [1..\sigma]$. Note that $aW$ is right-maximal, thus we can take the suffix link from the internal node of the suffix tree of $T\#$ labeled by $aW$ to the internal node labeled by $W$, using operation $\mathtt{contractLeft}$. Let $([x..y],[x'..y'])$ be the intervals of $W$ in $\BWT_{T\#}$ and in $\BWT_{\REV{T}\#}$, respectively. If $i=0$ or $\PLCP[i]=0$, rather than taking the suffix link from $aW$, we set $W=\varepsilon$, $x=x'=1$ and $y=y'=n+1$. Since $\PLCP[i+1] \geq \PLCP[i]-1$, we set $\PLCP[i+1]$ to its lower bound $|W|$. Then, we issue:
$$
([x..y],[x'..y']) \gets \mathtt{extendRight}(T'[i+\PLCP[i]],[x..y],[x'..y'])
$$
and we check whether $x=r_{i+1}$: if this is the case we stop, since neither $W \cdot T'[i+\PLCP[i]]$ nor any of its right-extensions are prefixes of the suffix at position $r_{i+1}-1$ in $\BWT_{T\#}$. Otherwise, we increment $\PLCP[i+1]$ by one and we continue issuing $\mathtt{extendRight}$ operations with the following character of $T'$. At the end of this process we know the interval of $T'[i+1..i+1+\PLCP[i+1]-1]$ in $\BWT_{T\#}$ and $\BWT_{\REV{T}\#}$, thus we can repeat the algorithm from position $i+2$.
\end{proof}

This algorithm can be easily adapted to compute the \emph{distinguishing statistics array} of a string $T$ given its bidirectional BWT index, and to compute the \emph{matching statistics array} of a string $T^2$ with respect to a string $T^1$, given the bidirectional BWT index of $T^1 \#_1 T^2 \#_2$: see Section \ref{sect:matchingStatistics}.

\subsection{Building the compressed suffix tree} \label{sec:buildingCST}

The \emph{compressed suffix tree} of a string $T \in [1..\sigma]^{n-1}$ \cite{Sa07a}, abbreviated to CST in what follows, is an index that consists of the following elements:
\begin{enumerate}
\item The compressed suffix array of $T\#$.
\item The topology of the suffix tree of $T\#$. This takes $4n+o(n)$ bits of space, but it can be reduced to $2.54n+o(n)$ bits~\cite{Fi11}. 
\item The permuted LCP array of $T\#$, which takes $2n+o(n)$ bits of space \cite{Sa07a}. 
\end{enumerate}
The CST is designed to support the same set of operations as the suffix tree. Specifically, all operations that involve just the suffix tree topology can be supported in constant time, including taking the parent of a node and the lowest common ancestor of two nodes. Most of the remaining operations are instead supported in time $t$, i.e. in the time required for accessing the value stored at a given suffix array position. Some operations are supported by augmenting the CST with other data structures: for example, following the edge that connects a node to its child with a given label (and returning an error if no such edge exists) needs additional $O(n\log{\log{\sigma}})$ bits, and runs in $t$ time. Some operations take even more time: for example, string level ancestor queries (defined in Section \ref{sec:buildingCST}) need additional $o(n)$ bits of space, and are supported in $O(t\log{\log{n}})$ time.

By just combining Lemma \ref{lemma:plcpConstruction} with Theorems \ref{thm:bidirectionalIndexConstruction}, \ref{thm:csaConstruction}, \ref{thm:topology} and \ref{thm:bwtConstruction}, we can prove the key result of Section \ref{sec:bwtIndexesConstruction}:

\begin{theorem}
Given a string $T = [1..\sigma]^n$, we can build the three main components of the compressed suffix tree (i.e. the compressed suffix array, the suffix tree topology, and the permuted LCP array) in randomized $O(n)$ time and in $O(n\log{\sigma})$ bits of working space. Such components take overall $O(n\log{\sigma})$ bits of space.
\end{theorem}



A number of applications of the suffix tree depend on the following additional operations: $\mathtt{stringDepth}(\mathtt{id}(v))$, which returns the length of the label of a node $v$ of the suffix tree; $\mathtt{blindChild}(\mathtt{id}(v),a)$, which returns the identifier of the child $w$ of a node $v$ of the suffix tree such that $\ell(v,w)=aW$ for some $W \in \Sigma^*$ and $a \in \Sigma$, and whose output is undefined if $v$ has no outgoing edge whose label starts with $a$; $\mathtt{child}(\mathtt{id}(v),a)$, which is analogous to $\mathtt{blindChild}$ but returns $\emptyset$ if $v$ has no outgoing edge whose label starts with $a$; and $\mathtt{stringAncestor}(\mathtt{id}(v),d)$, which returns the locus of the prefix of length $d$ of $\ell(v)$. The latter operation is called \emph{string level ancestor query}. Operation $\mathtt{stringDepth}$ can be supported in $O((\log_{\sigma}^{\epsilon}{n})/\epsilon)$ 
time using just the three main components of the compressed suffix tree. To support $\mathtt{blindChild}$ and $\mathtt{child}$ we need the following additional structure:

\begin{lemma}
Given a string $T = [1..\sigma]^n$, we can build in randomized $O(n)$ time and in $O(n\log{\sigma})$ bits of working space, a data structure that allows a compressed suffix tree to support operation $\mathtt{blindChild}$ in constant time, and operation $\mathtt{child}$ in $O((\log_{\sigma}^{\epsilon}{n})/\epsilon)$ 
time. Such data structure takes $O(n\log{\log{\sigma}})$ bits of space.
\end{lemma}
\begin{proof}
We build the following data structures, described in \cite{BN11,BNtalg14}. We use an array $\mathtt{nChildren}[1..2n-1]$, of $(2n-1)\log{\sigma}$ bits, to store the number of children of every suffix tree node, in preorder, and we use an array $\mathtt{labels}[1..2n-2]$ of $(2n-2)\log{\sigma}$ bits to store the sorted labels of the children of every node, in preorder. We enumerate the BWT intervals of every right-maximal substring $W$ of $T$, as well as the number $k$ of distinct right-extensions of $W$, using Theorem \ref{thm:iterator}. We convert $\INTERVAL{W}$ into the preorder identifier $i$ of the corresponding suffix tree node using the tree topology, and we set $\mathtt{nChildren}[i]=k$. Then, we build the prefix-sum data structure of Lemma \ref{lemma:prefixSums} on array $\mathtt{nChildren}$ (recall that this structure takes $O(n)$ bits of space), and we enumerate again the BWT interval of every right-maximal substring $W$ of $T$, along with its right-extensions $b_1,b_2,\dots,b_k$, using Theorem \ref{thm:iterator}. For every such $W$, we set $\mathtt{labels}[i+j]=b_j$ for all $j \in [0..k-1]$, where  $i$ is computed from the prefix-sum data structure. Finally, we scan $\mathtt{nChildren}$ and $\mathtt{labels}$ using pointers $i$ and $j$, respectively, both initialized to one, we iteratively build a monotone minimal perfect hash function on $\mathtt{labels}[j..j+\mathtt{nChildren}[i]-1]$ using Lemma \ref{lemma:build_mmphf2}, and we set $i$ to $i+1$ and $j$ to $j+\mathtt{nChildren}[i]$. All such MMPHFs fit in $O(n\log{\log{\sigma}})$ bits of space and they can be built in randomized $O(n)$ time.
\end{proof}

The output of $\mathtt{blindChild}$ can be checked in $O((\log_{\sigma}^{\epsilon}{n})/\epsilon)$ 
time using the compressed suffix array and the $\mathtt{stringDepth}$ operation, assuming we store the original string. 
Finally, the data structures that support string level ancestor queries can be built in deterministic linear time and in $O(n\log{\sigma})$ bits of space:

\begin{lemma}\label{lemma:stringAncestor}
Given the compressed suffix tree of a string $T = [1..\sigma]^n$, we can build in  $O(n)$ time and in $O(n\log{\sigma})$ bits of working space, a data structure that allows the compressed suffix tree to answer $\mathtt{stringAncestor}$ queries in $O( ((\log_{\sigma}^{\epsilon}{n})/\epsilon)\log^{\epsilon}{n} \cdot \log{\log{n}})$ time. Such data structure takes $o(n)$ bits of space.

\end{lemma}
\begin{proof}
We call \emph{depth} of a node the number of edges in the path that connects it to the root, and we call \emph{height} of an internal node $v$ the difference between the depth of the deepest leaf in the subtree rooted at $v$ and the depth of $v$. To build the data structure, we first sample a node every $b$ in the suffix tree. Specifically, we sample a node iff its depth is multiple of $b$ and its height is at least $b$. Note that the number of sampled nodes is at most $n/b$, since we can associate at least $b-1$ non-sampled nodes to every sampled node. Specifically, let $v$ be a sampled node at depth $ib$ for some $i$. If no descendant of $v$ is sampled, we can assign to $v$ all the at least $b$ nodes in the path from $v$ to its deepest leaf. If at least one descendant of $v$ is sampled, then $v$ has at least one sampled descendant $w$ at depth $(i+1)b$, and we can assign to $v$ all the $b-1$ non-sampled nodes in the path from $v$ to $w$.

We perform the sampling using just operations supported by the balanced parentheses representation of the topology of the suffix tree (see Lemma \ref{lemma:balancedParentheses}). Specifically, we perform a preorder traversal of the suffix tree topology using Lemma \ref{lemma:spaceEfficientPreorder}, we compute the depth and the height of every node $v$ using operations $\mathtt{depth}$ and $\mathtt{height}$ provided by the balanced parentheses representation, and, if $v$ has to be sampled, we append pair $(\mathtt{id}(v),\mathtt{stringDepth}(\mathtt{id}(v)))$ to a temporary list $\mathtt{pairs}$. Building $\mathtt{pairs}$ takes $O((n/b) \cdot (\log_{\sigma}^{\epsilon}{n})/\epsilon 
)$ time, and $\mathtt{pairs}$ itself takes $O((n/b)\log{n})$ bits of space.  During the traversal we also build a sequence of balanced parentheses $S$, that encodes the topology of the subgraph of the suffix tree induced by sampled nodes: every time we traverse a sampled node from its parent we append to $S$ an opening parenthesis, and every time we traverse a sampled node from its last child we append to $S$ a closing parentheses. At the end of this process, we build a weighted level ancestor data structure (WLA, see e.g. \cite{amir2007dynamic}) on the set of sampled nodes, where the weight assigned to a node equals its string depth. To do so, we build the data structure of Lemma \ref{lemma:balancedParentheses} on $S$, and we feed $S$ and $\mathtt{pairs}$ to the algorithm described in \cite{amir2007dynamic}. The WLA data structure takes $O(n/b)$ space and it can be built in $O(n/b)$ time. Finally, we build a dictionary $D$ that stores the identifiers of all sampled nodes: the size of this dictionary is $O((n/b)\log{n})$ bits. The dictionary and the WLA data structure are the output of our construction.

We now describe how to answer $\mathtt{stringAncestor}(\mathtt{id}(v),d)$, waving details on corner cases for brevity. We first check whether $w$, the lowest ancestor of $v$ at depth $ib$ for some $i$, is sampled: to do so, we compute $e=\mathtt{depth}(\mathtt{id}(v))$, we issue $\mathtt{ancestor}(\mathtt{id}(v),ib)$, where $i=\lfloor e/b \rfloor$, using the suffix tree topology, and we query $D$ with $\mathtt{id}(w)$. If $w$ is not sampled, we replace $w$ with its ancestor at depth $(i-1)b$, which is necessarily sampled. If the string depth of $w$ equals $d$, we return $\mathtt{id}(w)$. Otherwise, if the string depth of $w$ is less than $d$, we perform a binary search over the range of tree depths between the depth of $w$ plus one and the depth of $v$, using operations $\mathtt{ancestor}$ and $\mathtt{stringDepth}$. Otherwise, we query the WLA data structure to determine $u$, the deepest sampled ancestor of $w$ whose depth is less than $d$, and we perform a binary search over the range of depths between the depth of $u$ plus one and the depth of $w$, using operations $\mathtt{ancestor}$ and $\mathtt{stringDepth}$. Note that the range explored by the binary search is of size at most $2b$, thus the search takes $O(\log{b})$ steps and $O(\log{b} \cdot ((\log_{\sigma}^{\epsilon}{n})/\epsilon)    
)$ time. Setting $b=\log^{2}{n}$ makes the query time $O(\log{\log{n}} \cdot ((\log_{\sigma}^{\epsilon}{n})/\epsilon) 
)$, the time to build $\mathtt{pairs}$ $O(n)$, and the space taken by $\mathtt{pairs}$ and by the WLA data structure $o(n)$ bits.
\end{proof}


\section{String analysis} \label{sec:stringAnalysis}

In this section we use the enumerators of right-maximal substrings described in Theorems \ref{thm:iterator} and \ref{thm:generalizedIterator} to solve a number of fundamental string analysis problems in optimal deterministic time and small space. Specifically, we show that all such problems can be solved efficiently by just implementing function $\mathtt{callback}$ invoked by Algorithms \ref{algo:iterator} and \ref{algo:generalizedIterator}. We also show how to compute \emph{matching statistics} and \emph{distinguishing statistics} (defined below) using a bidirectional BWT index.

\subsection{Matching statistics} \label{sect:matchingStatistics}

\begin{definition}[\cite{weiner1973file,Wei73}]
Given two strings $S \in [1..\sigma]^n$ and $T \in [1..\sigma]^m$, and an integer threshold $\tau>0$, the \emph{matching statistics} $\MS_{T,S,\tau}$ of $T$ with respect to $S$ is a vector of length $m$ that stores at index $i \in [1..m]$ the length of the longest prefix of $T[i..m]$ that occurs at least $\tau$ times in $S$.
\end{definition}

\begin{definition}[\cite{weiner1973file,Wei73}]
Given $S \in [1..\sigma]^n$ and an integer threshold $\tau>0$, the \emph{distinguishing statistics} $\DS_{S,\tau}$ of $S$ is a vector of length $|S|$ that stores at index $i \in [1..|S|]$ the length of the shortest prefix of $S[i..|S|]$ that occurs at most $\tau$ times in $S$.
\end{definition}

We drop a subscript from $\DS_{S,\tau}$ whenever it is clear from the context. Note that $\DS_{S,\tau}[i] \geq 1$ for all $i$ and $\tau$. The key additional property of $\DS_{S,\tau}$, which is shared by $\PLCP_{S\#}$, is called \emph{$\delta$-monotonicity}:

\begin{definition}[\cite{robertson1968generalization}]
Let $a = a_0 a_1 \dots a_n$ and $\delta = \delta_1 \delta_2 \dots \delta_n$ be two sequences of nonnegative
integers. Sequence $a$ is said to be \emph{$\delta$-monotone} if $a_{i}-a_{i-1} \geq -\delta_{i}$ for all $i \in [1..n]$.
\end{definition}

Specifically, $\MS_{T,S,\tau}[i]-\MS_{T,S,\tau}[i-1] \geq -1$ for all $i \in [2..m]$, $\DS_{T,\tau}[i]-\DS_{T,\tau}[i-1] \geq -1$ and $\PLCP_{T\#}[i]-\PLCP_{T\#}[i-1] \geq -1$ for all $i \in [2..m+1]$. This property allows all three of these vectors to be encoded in $2x$ bits, where $x$ is the length of the corresponding input string \cite{Sa02,belazzougui2014indexed}.

The matching statistics array and the distinguishing statistics array of a string can be built in linear time from the bidirectional BWT index of \Thm{thm:bidirectionalIndexConstruction}:

\begin{lemma}\label{lemma:ds}
Given a bidirectional BWT index of $T \in [1..\sigma]^n$ that supports every operation in time linear in the size of its output, we can build $\DS_{T,\tau}$ in $O(n)$ time and in $O(\log{n})$ bits of working space. 
\end{lemma}
\begin{proof}
We proceed as in the proof of Lemma \ref{lemma:plcpConstruction}, scanning $T'=T\#$ from left to right. Assume that we are at position $i$ of $T'$, and assume that we know $\DS[i]$. Then, $aW=T'[i..i+\DS[i]-2]$ occurs more than $\tau$ times in $T'$ and it is a right-maximal substring of $T'$. To compute $\DS[i+1]$, we take the suffix link from the node of the suffix tree of $T'$ that corresponds to $aW$, using operation $\mathtt{contractLeft}$, and we issue $\mathtt{extendRight}$ operations on string $W$ using characters $T'[i+\DS[i]-1]$, $T'[i+\DS[i]]$, etc., until the frequency of the right-extension of $W$ drops again below $\tau+1$.
\end{proof}

\begin{lemma}\label{lemma:ms}
Let $S \in [1..\sigma]^n$ and $T \in [1..\sigma]^m$ be two strings. Given a bidirectional BWT index of their concatenation $S \#_1 T \#_2$ that supports every operation in time linear in the size of its output, we can build $\MS_{T,S,\tau}$ in $O(n+m)$ time and in $n+m+o(n+m)$ bits of working space.
\end{lemma}
\begin{proof}
We use the same algorithm as in Lemma \ref{lemma:ds}, scanning $T$ from left to right and checking at each step the frequency of the current string in $S$. This can be done in constant time using a bitvector $\mathtt{which}[1..n+m+2]$ indexed to support rank operations, such that $\mathtt{which}[i]=1$ iff the suffix of $S \#_1 T \#_2$ with lexicographic rank $i$ starts inside $S$.
\end{proof}

By plugging Theorem \ref{thm:bidirectionalIndexConstruction} into Lemmas \ref{lemma:ds} and \ref{lemma:ms}, we immediately get the following result:

\begin{theorem}
Let $S \in [1..\sigma]^n$ and $T \in [1..\sigma]^m$ be two strings. We can build $\DS_{T,\tau}$ in randomized $O(m)$ time and in $O(m\log{\sigma})$ bits of working space, and we can build $\MS_{T,S,\tau}$ in randomized $O(n+m)$ time and in $O((n+m)\log{\sigma})$ bits of working space.
\end{theorem}

Moreover, using Algorithm \ref{algo:generalizedIterator}, we can achieve the same bounds in deterministic linear time:

\begin{theorem}
Let $S \in [1..\sigma]^n$ and $T \in [1..\sigma]^m$ be two strings. We can build $\DS_{T,\tau}$ in $O(m)$ time and in $O(m \log{\sigma})$ bits of working space, and we can build $\MS_{T,S,\tau}$ in $O(n+m)$ time and in $O((n+m)\log{\sigma})$ bits of working space.
\end{theorem}
\begin{proof}
For simplicity we describe just how to compute $\MS_{T,S,1}$. Note that array $\MS_{T,S}$ can be built in linear time from two bitvectors $\mathtt{start}$ and $\mathtt{end}$, of size $|T|$ each, where $\mathtt{start}[i]=1$ iff $\MS_{T,S}[i]>\MS_{T,S}[i-1]-1$, and where $\mathtt{end}[j]=1$ iff there is an $i$ such that $j=i+\MS_{T,S}[i]-1$.

To build $\mathtt{start}$, we use an auxiliary bitvector $\mathtt{start}'$ of size $|T|+1$, initialized to zeros, and we run Algorithm \ref{algo:generalizedIterator} to iterate over all right-maximal substrings $W$ of $S\#_{1}T\#_2$ that occur both in $S$ and in $T$. Let $\mathtt{repr}'(W)=(\{\mathtt{chars}^S,\mathtt{chars}^T\},\{\mathtt{first}^S,\mathtt{first}^T\})$. If $\mathtt{chars}^T \setminus \mathtt{chars}^S = \emptyset$, we don't process $W$ further and we continue the iteration. Otherwise, for every character $b \in \mathtt{chars}^T \setminus \mathtt{chars}^S$, we enumerate all the distinct characters $a$ that occur to the left of $Wb$ in $T$, and their corresponding intervals in $\BWT_{T\#}$, as described in Lemma \ref{lemma:extendLeft}. If $aW$ is not a prefix of a rotation of $S$, we set to one all positions in $\mathtt{start}'[i..j]$, where $[i..j]$ is the interval of $aWb$ in $\BWT_{T\#}$. At the end of this process, we invert $\BWT_{T\#}$ and we set $\mathtt{start}[i+1]=1$ for every $i$ such that $\mathtt{start}'[j]=1$ and $j$ is the lexicographic rank of suffix $T[i..|T|]\#$ among all suffixes of $T\#$. Finally, we repeat the entire process using $\BWT_{\REV{T}\#}$, $\BWT_{\REV{S}\#}$ and $\mathtt{end}'$. The claimed complexity comes from Theorems \ref{thm:bwtConstruction} and \ref{thm:generalizedIterator}.
\end{proof}

\subsection{Maximal repeats, maximal unique matches, maximal exact matches.} \label{sect:maximalrepeats}

Recall from \Sec{sect:definitions} that string $W$ is a \emph{maximal repeat} of string $T \in [1..\sigma]^n$ if $W$ is both left-maximal and right-maximal in $T$. Let $\{W^1,W^2,\dots,W^{\mathtt{occ}}\}$ be the set of all $\mathtt{occ}$ distinct maximal repeats of $T$. We encode such set as a list of $\mathtt{occ}$ pairs of words $(p^i,|W^i|)$, where $p^i$ is the starting position of an occurrence of $W^i$ in $T$.

\begin{theorem} \label{thm:maxRepeats}
Given a string $T \in [1..\sigma]^n$, we can compute an encoding of all its $\mathtt{occ}$ distinct maximal repeats in $O(n+\mathtt{occ})$ time and in $O(n\log{\sigma})$ bits of working space.
\end{theorem}
\begin{proof}
Recall from \Sec{sec:enumeration} the representation $\mathtt{repr}(W)$ of a substring $W$ of $T$. \Algo{algo:iterator} invokes function $\mathtt{callback}$ on every right-maximal substring $W$ of $T$: inside such function we can determine the left-maximality of $W$ by checking whether $h>1$, and in the positive case we append pair $(\mathtt{first}[1],|W|)$ to a list $\mathtt{pairs}$, where $\mathtt{first}[1]$ in $\REPR{W}$ is the first position of the interval of $W$ in $\BWT_{T\#}$ (see \Algo{algo:maxRepeats}). After the execution of the whole \Algo{algo:iterator}, we feed $\mathtt{pairs}$ to \Lemma{lemma:batch_extract}, obtaining in output a list of lengths and starting positions in $T$ that uniquely identifies the set of all maximal repeats of $T$. 

We build $\BWT_{T\#}$ from $T$ using \Thm{thm:bwtConstruction}. Then, we use \Lemma{lemma:rank_select_access} to build a data structure that supports $\mathtt{access}$ and $\mathtt{partialRank}$ queries on $\BWT_{T\#}$, and we discard $\BWT_{T\#}$. We use this structure both to implement function $\LF$ in \Lemma{lemma:batch_extract}, and to build the $\mathtt{rangeDistinct}$ data structure of \Lemma{lemma:rangedistinct}. Finally, as described in \Thm{thm:iterator}, we implement \Lemma{lemma:iterator} with this $\mathtt{rangeDistinct}$ data structure. We allocate the space for $\mathtt{pairs}$ and for related data structures in \Lemma{lemma:batch_extract} using the static allocation strategy described in \Sec{sec:staticAllocation}. We charge to the output the space taken by $\mathtt{pairs}$. In \Lemma{lemma:batch_extract}, we charge to the output the space taken by list $\mathtt{translate}$, as well as part of the working space used by radix sort.
\end{proof}

Maximal repeats have been detected from the input string in $O(n\log{\sigma})$ bits of working space before, but not in overall $O(n)$ time. Specifically, it is possible to achieve overall running time $O(n\log{\sigma})$ by combining the BWT construction algorithm described in \cite{HSS09}, which runs in $O(n\log{\log{\sigma}})$ time, with the maximal repeat detection algorithm described in \cite{BBO12}, which runs in $O(n\log{\sigma})$ time. The claim of \Thm{thm:maxRepeats} holds also for an encoding of the maximal repeats that contains, for every maximal repeat, the starting position of \emph{all} its occurrences in $T$. In this case, $\mathtt{occ}$ becomes the number of \emph{occurrences} of all maximal repeats of $T$. Specifically, given the BWT interval of a maximal repeat $W$, it suffices to mark all the positions inside the interval in a bitvector $\mathtt{marked}[1..n]$, to assign a unique identifier to every distinct maximal repeat, and to sort the translated list $\mathtt{pairs}$ by such identifier before returning it in output. Bitvector $\mathtt{marked}$ can be replaced by a smaller bitvector $\mathtt{marked}'$ as described in Lemma \ref{lemma:batch_extract}.

Once we have the encoding $(p^i,|W^i|)$ of every maximal repeat $W^i$, we can return the corresponding \emph{string} $W^i$ by scanning $T$ in blocks of size $\log{n}$, i.e. outputting $\log_{\sigma}{n}$ characters in constant time: this allows us to print the $C$ total characters in the output in overall $C/\log_{\sigma}{n}$ time. Alternatively, we can discard the original string altogether, and maintain instead an auxiliary stack of characters while we traverse the suffix-link tree in Lemma \ref{lemma:iterator}. Once we detect a maximal repeat, we print its string to the output by scanning the auxiliary stack in blocks of size $\log{n}$. Recall from Section \ref{sec:suffixTree} that the leaves of the suffix-link tree are maximal repeats: this implies that the depth $d$ of the auxiliary stack is at most equal to the length of the longest maximal repeat, thus the maximum size $d\log{\sigma}$ of the auxiliary stack can be charged to the output.

A \emph{supermaximal repeat} is a maximal repeat that is not a substring of another maximal repeat, and a \emph{near-supermaximal repeat} is a maximal repeat that has at least one occurrence that is not contained inside an occurrence of another maximal repeat (see e.g. \cite{Gu97}). The proof of \Thm{thm:maxRepeats} can be adapted to detect supermaximal and near-supermaximal repeats within the same bounds: we leave the details to the reader.

\begin{algorithm}[t]
\KwIn{$\REPR{W}$, $|W|$, $\BWT_{T}$, and $C$ array of string $T \in [1..\sigma]^{n-1}\#$. Matrices $A$, $F$, $L$, $\mathtt{gamma}$, $\mathtt{leftExtensions}$, and counter $h$, from \Lemma{lemma:extendLeft}. List $\mathtt{pairs}$.}
\If{$h<2$}{\textbf{return}\;}
$\mathtt{pairs}.\mathtt{append}\big( (\mathtt{first}[1],|W|) \big)$\;
\caption{\label{algo:maxRepeats}
Function $\mathtt{callback}$ for maximal repeats. See \Thm{thm:maxRepeats} and \Algo{algo:iterator}.}
\end{algorithm}

Consider two string $S$ and $T$. For a \emph{maximal unique match (MUM)} $W$ between $S \in [1..\sigma]^{n}$ and $T \in [1..\sigma]^{m}$, it holds that: (1) $W = S[i..i+k-1]$ and $W = T[j..j+k-1]$ for exactly one $i \in [1..n]$ and for exactly one $j \in [1..m]$; (2) if $i-1 \geq 1$ and $j-1 \geq 1$, then $S[i-1] \neq T[j-1]$; (4) if $i+k \leq n$ and $j+k \leq m$, then $S[i+k] \neq T[j+k]$ (see e.g. \cite{Gu97}). To detect all the MUMs of $S$ and $T$, it would suffice to build the suffix tree of the concatenation $C=S\#_{1}T\#_2$ and to traverse its internal nodes, since MUMs are right-maximal substrings of $C$, like maximal repeats. More specifically, only internal nodes $v$ with exactly two leaves as children can be MUMs. Let the two leaves of a node $v$ be associated with suffixes $C[i \ltdots |C|]$ and $C[j \ltdots |C|]$, respectively. Then, $i$ and $j$ must be such that $i \leq |S|$ and $j>|S+1|$, and the left-maximality of $v$ can be checked by accessing $S[i-1 \; (\mbox{mod}_1 \; n)]$ and $T[j-1 \; (\mbox{mod}_1 \; m)]$ in constant time.

This notion extends naturally to a set of strings: a string $W$ is a maximal unique match (MUM) of $d$ strings $T^1,T^2,\dots,T^d$, where $T^i \in [1..\sigma]^{n_{i}}$, if $W$ occurs exactly once in $T^i$ for all $i \in [1..d]$, and if $W$ cannot be extended to the left or to the right without losing one of its occurrences. We encode the set of all maximal unique matches $W$ of $T^1,T^2,\dots,T^d$ as a list of $\mathtt{occ}$ triplets of words $(p^i,|W|,\mathtt{id})$, where $p^i$ is the first position of the occurrence of $W$ in string $T^i$, and $\mathtt{id}$ is a number that uniquely identifies $W$. Note that the maximal unique matches of $T^1,T^2,\dots,T^d$ are maximal repeats of the concatenation $T = T^1 \#_1 T^2 \#_1 \cdots \#_1 T^d \#_2$, thus we can adapt \Thm{thm:maxRepeats} as follows:

\begin{theorem} \label{thm:mums}
Given a set of strings $T^1, T^2, \dots, T^d$ where $d>1$ and $T^i \in [1..\sigma]^{+}$ for all $i \in [1..d]$, we can compute an encoding of all the distinct maximal unique matches of the set in $O(n+\mathtt{occ})$ time and in $O(n\log{\sigma})$ bits of working space, where $n=\sum_{i=1}^{d}|T^i|$ and $\mathtt{occ}$ is the number of words in the encoding.
\end{theorem}
\begin{proof}
We build the same data structures as in \Thm{thm:maxRepeats}, but on string $T = T^1 \#_1 T^2 \#_1 \cdots \#_1 T^d \#_2$, and we enumerate all the maximal repeats of $T$ using \Algo{algo:iterator}. Whenever we find a maximal repeat $W$ with exactly $d$ occurrences in $T$, we set to one in a bitvector $\mathtt{intervals}[1..|T|]$ the first and the last position of the interval of $W$ in $\BWT_T$ (see \Algo{algo:mums1}). Note that the BWT intervals of all the maximal repeats of $T$ with exactly $d$ occurrences are disjoint. Then, we index $\mathtt{intervals}$ to support rank queries in constant time, we allocate another bitvector $\mathtt{documents}[1..|T|]$, and we invert $\BWT_T$. Assume that, at the generic step of the inversion, we are at position $i$ in $T$ and at position $j$ in $\BWT_T$. We decide whether $j$ belongs to the interval of a maximal repeat with $d$ occurrences, by checking whether $\mathtt{rank}_{1}(\mathtt{intervals},j)$ is odd, or, if it is even, whether $\mathtt{intervals}[j]=1$. If $j$ belongs to an interval $[x..y]$ that has been marked in $\mathtt{intervals}$, we compute $x$ using rank queries on $\mathtt{intervals}$, and we set $\mathtt{documents}[x+p-1]$ to one, where $p$ is the identifier of the document that contains position $i$ in $T$. Finally, we scan bitvectors $\mathtt{intervals}$ and $\mathtt{documents}$ synchronously: for each interval $[x..y]$ that has been marked in $\mathtt{intervals}$ and such that $\mathtt{documents}[i]=0$ for some $i \in [x..y]$, we reset $\mathtt{intervals}[x]$ and $\mathtt{intervals}[y]$ to zero. Finally, we iterate again over all the maximal repeats of $T$ with exactly $d$ occurrences, using \Algo{algo:iterator}. Let $W$ be such a maximal repeat, with interval $[x..y]$ in $\BWT_T$: if $\mathtt{intervals}[x]=\mathtt{intervals}[y]=1$, we append to list $\mathtt{pairs}$ of \Thm{thm:maxRepeats} a triplet $(i,|W|,\mathtt{id})$ for all $i \in [x..y]$, where $\mathtt{id}$ is a number that uniquely identifies $W$ (see \Algo{algo:mums2}). Then, we continue as in \Thm{thm:maxRepeats}.
\end{proof}

\begin{algorithm}[t]
\KwIn{$\REPR{W}$, $|W|$, $\BWT_{T}$, and $C$ array of string $T \in [1..\sigma]^{n-1}\#$. Matrices $A$, $F$, $L$, $\mathtt{gamma}$, $\mathtt{leftExtensions}$, and counter $h$, from \Lemma{lemma:extendLeft}. Bitvector $\mathtt{intervals}[1..|T|]$.}
\If{$h<2$ \Or $\mathtt{first}[|\mathtt{first}|]-\mathtt{first}[1] \neq d$}{\textbf{return}\;}
$\mathtt{intervals}[\mathtt{first}[1]] \gets 1$\;
$\mathtt{intervals}[\mathtt{first}[|\mathtt{first}|]-1] \gets 1$\;
\caption{\label{algo:mums1}
First $\mathtt{callback}$ function for maximal unique matches. See \Thm{thm:mums} and \Algo{algo:iterator}.}
\end{algorithm}

\begin{algorithm}[t]
\KwIn{$\REPR{W}$, $|W|$, $\BWT_{T}$, and $C$ array of string $T \in [1..\sigma]^{n-1}\#$. Matrices $A$, $F$, $L$, $\mathtt{gamma}$, $\mathtt{leftExtensions}$, and counter $h$, from \Lemma{lemma:extendLeft}. Bitvector $\mathtt{intervals}[1..|T|]$. List $\mathtt{pairs}$. Integer $\mathtt{id}$.}
\If{$h<2$ \Or $\mathtt{first}[|\mathtt{first}|]-\mathtt{first}[1] \neq d$ \Or $\mathtt{intervals}[\mathtt{first}[1]] \neq 1$ \Or $\mathtt{intervals}[\mathtt{first}[|\mathtt{first}|]-1] \neq 1$}{\textbf{return}\;}
\For{$i \in [\mathtt{first}[1]..\mathtt{first}[|\mathtt{first}|]-1]$}{
	$\mathtt{pairs}.\mathtt{append}\big( (i,|W|,\mathtt{id}) \big)$\;
	$\mathtt{id} \gets \mathtt{id}+1$\;
}
\caption{\label{algo:mums2}
Second $\mathtt{callback}$ function for maximal unique matches. See \Thm{thm:mums} and \Algo{algo:iterator}.}
\end{algorithm}

Maximal exact matches (MEMs) are related to maximal repeats as well. A triplet $(i,j,\ell)$ is a \emph{maximal exact match} (also called \emph{maximal pair}) of two strings $T^1$ and $T^2$ if: (1) $T^{1}[i \ldots i+\ell-1]=T^{2}[j \ldots j+\ell-1]=W$; (2) if $i-1 \geq 1$ and $j-1 \geq 1$, then $T^{1}[i-1] \neq T^{2}[j-1]$; (3) if $i+\ell \leq |T^1|$ and $j+\ell \leq |T^2|$, then $T^{1}[i+\ell] \neq T^{2}[j+\ell]$ (see e.g. \cite{baker1995finding,Gu97}). We encode the set of all maximal exact matches of strings $T^1$ and $T^2$ as a list of $\mathtt{occ}$ such triplets. Since $W$ is a maximal repeat of $T^{1}\#_{1}T^{2}\#_{2}$ that occurs both in $T^1$ and in $T^2$, we can build a detection algorithm on top of the generalized iterator of \Algo{algo:generalizedIterator}, as follows:

\begin{theorem} \label{thm:mems}
Given two strings $T^1$ and $T^2$ in $[1..\sigma]^+$, we can compute an encoding of all their $\mathtt{occ}$ maximal exact matches in $O(|T^1|+|T^2|+\mathtt{occ})$ time and in $O((|T^1|+|T^2|)\log{\sigma})$ bits of working space.
\end{theorem}
\begin{proof}
Recall that \Algo{algo:generalizedIterator} uses a $\mathtt{rangeDistinct}$ data structure built on top of the BWT of $T^1$, and a $\mathtt{rangeDistinct}$ data structure built on top of the BWT of $T^2$, to iterate over all the right-maximal substrings $W$ of $T^{1}\#_{1}T^{2}\#_{2}$. For every such $W$, the algorithm gives to function $\mathtt{callback}$ the intervals of all strings $aWb$ such that $a \in [1..\sigma]$, $b \in [1..\sigma]$, and $aWb$ is a prefix of a rotation of $T^1$, and with the intervals of all strings $cWd$ such that $c \in [1..\sigma]$, $d \in [1..\sigma]$, and $cWd$ is a prefix of a rotation of $T^2$. Recall from \Sec{sec:enumeration} the representation $\mathtt{repr}'(W)$ of a substring $W$ of $T^{1}\#_{1}T^{2}\#_{2}$. Inside function $\mathtt{callback}$, it suffices to determine whether $W$ occurs in both $T^1$ and $T^2$, using arrays $\mathtt{first}^1$ and $\mathtt{first}^2$ of $\REPRPRIME{W}$, and to determine whether $W$ is left-maximal in $T^{1}\#_{1}T^{2}\#_{2}$, by checking whether $h>1$ (see \Algo{algo:mems}). If both such tests succeed, we build a set $X$ that represents all strings $aWb$ that are the prefix of a rotation of $T^1$, and a set $X^2$ that represents all strings $cWd$ that are the prefix of a rotation of $T^2$:
\begin{eqnarray*}
X^1 & = \{ & (a,b,i,j) : a=\mathtt{leftExtensions}[p], \; p \leq h, \; \mathtt{gamma}^{1}[a]>0, \; b=A^{1}[a][q], \\
& & q \leq \mathtt{gamma}^{1}[a], \; i=F^{1}[a][q], \; j=L^{1}[a][q] \; \} \\
X^2 & = \{ & (c,d,i',j') : c=\mathtt{leftExtensions}[p], \; p \leq h, \; \mathtt{gamma}^{2}[c]>0, \; d=A^{2}[c][q], \\
& & q \leq \mathtt{gamma}^{1}[c], \; i'=F^{2}[c][q], \; j'=L^{2}[c][q] \; \}
\end{eqnarray*}
In such sets, $[i..j]$ is the interval of $aWb$ in the BWT of $T^{1}\#$, and $[i'..j']$ is the interval of $cWd$ in the BWT of $T^{2}\#$. Building $X^1$ and $X^2$ for all maximal repeats $W$ of $T^{1}\#_{1}T^{2}\#_2$ takes overall linear time in the size of the input, since every element of $X^1$ (respectively, of $X^2$) can be charged to a distinct edge or implicit Weiner link of the generalized suffix tree of $T^{1}\#_{1}T^{2}\#_2$, and the number of such objects is linear in the size of the input (see \Obs{obs:suffixtree}). Then, we use \Lemma{lemma:product} to compute the set of all quadruplets $(i,j,i',j')$ such that $(a,b,i,j) \in X^1$, $(c,d,i',j') \in X^2$, $a \neq c$ and $b \neq d$, in overall linear time in the size of the input and of the output, and for every such quadruplet we append all triplets $(x,y,|W|)$ to list $\mathtt{pairs}$ of \Thm{thm:maxRepeats}, where $x \in [i..j]$ and $y \in [i'..j']$. Running \Algo{algo:generalizedIterator} and building its input data structures from $T^1$ and $T^2$ takes overall $O(|T^1|+|T^2|)$ time and $O((|T^1|+|T^2|)\log{\sigma})$ bits of working space, by combining \Thm{thm:bwtConstruction} with Lemmas \ref{lemma:rank_select_access}, \ref{lemma:rangedistinct} and \ref{lemma:generalizedIterator}. 

Finally, we translate every $x$ and $y$ in $\mathtt{pairs}$ to a string position, as described in \Thm{thm:maxRepeats}. We allocate the space for $\mathtt{pairs}$ and for related data structures in \Lemma{lemma:batch_extract} using the static allocation strategy described in \Sec{sec:staticAllocation}. We charge to the output the space taken by $\mathtt{pairs}$. In \Lemma{lemma:batch_extract}, we charge to the output the space taken by list $\mathtt{translate}$, as well as part of the working space used by radix sort.
\end{proof}

\begin{algorithm}[t]
\KwIn{$\REPRPRIME{W}$, $|W|$, $\{\BWT_{T^{i}\#}\}$, $\{C^i\}$ arrays. Matrices $\{A^i\}$, $\{F^i\}$, $\{L^i\}$, $\{\mathtt{gamma}^i\}$. Array $\mathtt{leftExtensions}$ and counter $h$. List $\mathtt{pairs}$.}
\If{$h<2$ \Or $|\mathtt{chars}^{1}|=0$ \Or $|\mathtt{chars}^{2}|=0$}{\textbf{return}\;}
$X^1 \gets \emptyset$\;
$X^2 \gets \emptyset$\;
\For{$i \in [1..h]$}{
	$a \gets \mathtt{leftExtensions}[i]$\;
	\If{$\mathtt{gamma}^{1}[a]>0$}{
		\For{$j \in [1..\mathtt{gamma}^{1}[a]]$}{
			$X^1 \gets X^1 \cup \{(a,A^{1}[a][j],F^{1}[a][j],L^{1}[a][j])\}$\;
		}
	}
	\If{$\mathtt{gamma}^{2}[a]>0$}{
		\For{$j \in [1..\mathtt{gamma}^{2}[a]]$}{
			$X^2 \gets X^2 \cup \{(a,A^{2}[a][j],F^{2}[a][j],L^{2}[a][j])\}$\;
		}
	}	
}
$Y \gets X^1 \otimes X^2$\;
\For{$(i,j,i',j') \in Y$}{
	\For{$x \in [i..j]$, $y \in [i'..j']$}{
		$\mathtt{pairs}.\mathtt{append}\big( (x,y,|W|) \big)$\;
	}
}
\caption{\label{algo:mems}
Function $\mathtt{callback}$ for maximal exact matches. See \Thm{thm:mems} and \Algo{algo:generalizedIterator}. Operator $\otimes$ is from \Lemma{lemma:product}.}
\end{algorithm}

\begin{lemma} \label{lemma:product}
Let $\Sigma$ be a set, and let $A$ and $B$ be two subsets of $\Sigma \times \Sigma$. We can compute $A \otimes B = \{(a,b,c,d) \st (a,b) \in A, (c,d) \in B, a \neq c, b \neq d\}$ in $O(|A|+|B|+|A \otimes B|)$ time.
\end{lemma}
\begin{proof}
We assume without loss of generality that $|A|<|B|$. We say that two pairs $(a,b),(c,d)$ are \emph{compatible} if $a \neq c$ and $b \neq d$. Note that, if $(a,b)$ and $(c,d)$ are compatible, then the only elements of $\Sigma \times \Sigma$ that are incompatible with both $(a,b)$ and $(c,d)$ are $(a,d)$ and $(c,b)$. We iteratively select a pair $(a,b) \in A$ and scan $A$ in $O(|A|)=O(|B|)$ time to find another compatible pair $(c,d)$: if we find one, we scan $B$ and report every pair in $B$ that is compatible with either $(a,b)$ or $(c,d)$. The output will be of size $|B|-2$ or larger, thus the time to scan $A$ and $B$ can be charged to the output. Then, we remove $(a,b)$ and $(c,d)$ from $A$ and repeat the process. If $A$ becomes empty we stop. If all the remaining pairs in $A$ are incompatible with our selected pair $(a,b)$, that is, if $c=a$ or $d=b$ for every $(c,d) \in A$, we build subsets $A^a$ and $A^b$ where $A^a=\{(a,x) : x \neq b\} \subseteq A$ and $A^b=\{ (x,b) : x \neq a\} \subseteq A$. Then we scan $B$, and for every pair $(x,y) \in B$ different from $(a,b)$ we do the following. If $x \neq a$ and $y \neq b$, then we report $(a,b,x,y)$, $\{(a,z,x,y) : (a,z) \in A^a, z \neq y\}$ and $\{(z,b,x,y) : (z,b) \in A^b, z \neq x\}$. Pairs $(a,y) \in A^a$ and $(x,b) \in A^b$ are the only ones that do not produce output, thus the cost of scanning $A^a$ and $A^b$ can be charged to printing the result. If $x=a$ and $y \neq b$, then we report $\{(z,b,x,y) : (z,b) \in A^b\}$. If $x \neq a$ and $y=b$, then we report $\{(a,z,x,y) : (a,z) \in A^a\}$.
\end{proof}

\Thm{thm:mems} uses the matrices and arrays of \Lemma{lemma:extendLeft} to access all the left-extensions $aW$ of a string $W$, and for every such left-extension to access all its right-extensions $aWb$. A similar approach can be used to compute all the \emph{minimal absent words} of a string $T$. String $W$ is a \emph{minimal absent word} of a string $T \in \Sigma^+$ if $W$ is not a substring of $T$ and if every proper substring of $W$ is a substring of $T$ (see e.g. \cite{crochemore1998automata}). To decide whether $aWb$ is a minimal absent word of $T$, where $\{a,b\} \subseteq \Sigma$, it suffices to check that $aWb$ does not occur in $T$, and that both $aW$ and $Wb$ occur in $T$. Only a maximal repeat of $T$ can be the infix $W$ of a minimal absent word $aWb$: we can enumerate all the maximal repeats $W$ of $T$ as in \Thm{thm:maxRepeats}. Recall also that $aWb$ is a minimal absent word of $T$ only if both $aW$ and $Wb$ occur in $T$. We can use $\REPR{W}$ to enumerate all strings $Wb$ that occur in $T$, we can use vector $\mathtt{leftExtensions}$ to enumerate all strings $aW$ that occur in $T$, and finally we can use matrix $A$ to discard all strings $aWb$ that occur in $T$. \Algo{algo:maws} uses this approach to output an encoding of all distinct minimal absent words of $T$ as a list of triplets $(i,\ell,b)$, where each triplet encodes minimal absent word $T[i..i+\ell-1] \cdot b$. Every operation of this algorithm can be charged to an element of the output, to an edge of the suffix tree of $T$, or to a Weiner link. The following theorem holds by this observation, and by applying the same steps as in \Thm{thm:maxRepeats}: we leave its proof to the reader.

\begin{theorem} \label{thm:maws}
Given a string $T \in [1..\sigma]^n$, we can compute an encoding of all its $\mathtt{occ}$ minimal absent words in $O(n+\mathtt{occ})$ time and in $O(n\log{\sigma})$ bits of working space.
\end{theorem}

Recall from \Sec{sect:definitions} that $\mathtt{occ}$ can be of size $\Theta(n\sigma)$ in this case. Minimal absent words have been detected in linear time in the length of the input before, but using a suffix array (see \cite{barton2014linear} and references therein).

\begin{algorithm}[t]
\KwIn{$\REPR{W}$, $|W|$, $\BWT_{T}$, and $C$ array of string $T \in [1..\sigma]^{n-1}\#$. Matrices $A$, $F$, $L$, $\mathtt{gamma}$, $\mathtt{leftExtensions}$, and counter $h$, from \Lemma{lemma:extendLeft}. Bitvector $\mathtt{used}[1..\sigma]$ initialized to all zeros. List $\mathtt{pairs}$.}
\If{$h<2$}{\textbf{return}\;}
\For{$i \in [1..|\mathtt{chars}|]$}{
	$\mathtt{used}[\mathtt{chars}[i]] \gets 1$\;
}
\For{$i \in [1..h]$}{
	$a \gets \mathtt{leftExtensions}[i]$\;
	\For{$j \in [1..\mathtt{gamma}[a]]$}{
		$\mathtt{used}[A[a][j]] \gets 0$\;
	}
	\For{$j \in [1..|\mathtt{chars}|]$}{
		$b \gets \mathtt{chars}[j]$\;
		\If{$\mathtt{used}[b]=0$}{
			$\mathtt{pairs}.\mathtt{append}\big( (F[a][1],|W|+1,b) \big)$\;
			$\mathtt{used}[b] \gets 1$\;
		}
	}
}
\For{$i \in [1..|\mathtt{chars}|]$}{
	$\mathtt{used}[\mathtt{chars}[i]] \gets 0$\;
}
\caption{\label{algo:maws}
Function $\mathtt{callback}$ for minimal absent words. See \Thm{thm:maws} and \Algo{algo:iterator}.}
\end{algorithm}

\subsection{String kernels} \label{sect:stringkernels}

Another way of comparing and analyzing strings consists in studying the composition and abundance of all the distinct strings that occur in them. Given two strings $T^1$ and $T^2$, a \emph{string kernel} is a function that simultaneously converts $T^1$ and $T^2$ to \emph{composition vectors} $\{\mathbf{T^1},\mathbf{T^2}\} \subset \mathbb{R}^n$, indexed by a given set of $n>0$ distinct strings, and that computes a similarity or a distance measure between $\mathbf{T^1}$ and $\mathbf{T^2}$ (see e.g. \cite{haussler1999convolution,lodhi2002text}). Value $\mathbf{T^i}[W]$ is typically a function of the number $f_{T^i}(W)$ of (possibly overlapping) occurrences of string $W$ in $T^i$ (for example the estimate $p_{i}(W)=f_{T^i}(W)/(|T^i|-|W|+1)$ of the empirical probability of observing $W$ in $T^i$). In this section, we focus on computing the cosine of the angle between $\mathbf{T^1}$ and $\mathbf{T^2}$, defined as:
\begin{eqnarray*}
\kappa(\mathbf{T^1},\mathbf{T^2}) = \frac{\sum_{W}\mathbf{T^1}[W]\mathbf{T^2}[W]}{\sqrt{\left( \sum_{W}\mathbf{T^1}[W]^2 \right) \left( \sum_{W}\mathbf{T^2}[W]^2 \right)}} \label{eq:cosine}
\end{eqnarray*}
Specifically, we consider the case in which $\mathbf{T^i}$ is indexed by all distinct strings of a given length $k$ (called \emph{$k$-mers}), and the case in which $\mathbf{T^i}$ is indexed by all distinct strings \emph{of any length}:



\begin{definition}
Given a string $T \in [1 \ltdots \sigma]^+$ and a length $k>0$, let vector $\mathbf{T_k}=[1 \ltdots \sigma^k]$ be such that $\mathbf{T_k}[W]=f_{T}(W)$ for every $W \in [1 \ltdots \sigma]^k$. The \emph{$k$-mer complexity} $C(T,k)$ of string $T$ is the number of nonzero components of $\mathbf{T_k}$. The \emph{$k$-mer kernel} between two strings $T^1$ and $T^2$ is $\kappa(\mathbf{T^1_k},\mathbf{T^2_k})$.
\end{definition}

\begin{definition} \label{def:substringKernel}
Given a string $T \in [1 \ltdots \sigma]^+$, consider the infinite-dimensional vector $\mathbf{T_\infty}$, indexed by all distinct substrings $W \in [1 \ltdots \sigma]^+$, such that $\mathbf{T_\infty}[W]=f_{T}(W)$. The \emph{substring complexity} $C(T)$ of string $T$ is the number of nonzero components of $\mathbf{T_\infty}$. The \emph{substring kernel} between two strings $T^1$ and $T^2$ is $\kappa(\mathbf{T^1_\infty},\mathbf{T^2_\infty})$.
\end{definition}

Substring complexity and substring kernels, with or without a constraint on string length, can be computed using the suffix tree of a single string or the generalized suffix tree of two strings, using a \emph{telescoping technique} that works by adding and subtracting terms to and from a sum, and that does not depend on the order in which the nodes of the suffix tree are enumerated \cite{belazzougui2015framework}. We can thus implement all such algorithms as callback functions of Algorithms \ref{algo:iterator} and \ref{algo:generalizedIterator}:

\begin{theorem} \label{thm:kmerComplexity}
Given a string $T \in [1 \ltdots \sigma]^n$, there is an algorithm that computes:
\begin{itemize}
\item the $k$-mer complexity $C(T,k)$ of $T$, in $O(n)$ time and in $O(n\log{\sigma})$ bits of working space, for a given integer $k$;
\item the substring complexity $C(T)$, in $O(n)$ time and in $O(n\log{\sigma})$ bits of working space.
\end{itemize}
Given two strings $T^1$ and $T^2$ in $[1 \ltdots \sigma]^+$, there is an algorithm that computes:
\begin{itemize}
\item the $k$-mer kernel between $T^1$ and $T^2$, in $O(|T^1|+|T^2|)$ time and $O((|T^1|+|T^2|)\log{\sigma})$ bits of working space, for a given integer $k$;
\item the substring kernel between $T^1$ and $T^2$, in $O(|T^1|+|T^2|)$ time and in $O((|T^1|+|T^2|)\log{\sigma})$ bits of working space.
\end{itemize}
\end{theorem}
\begin{proof}
To make the paper self-contained, we just sketch the proof of $k$-mer complexity given in \cite{belazzougui2015framework}: the same telescoping technique can be applied to solve all other problems: see \cite{belazzougui2015framework}. 

A $k$-mer of $T$ is either the label of a node of the suffix tree of $T$, or it ends in the middle of an edge $(u,v)$ of the suffix tree. In the latter case, we assume that the $k$-mer is represented by its locus $v$, which might be a leaf. Let $C(T,k)$ be initialized to $|T|+1-k$, i.e. to the number of leaves that correspond to suffixes of $T\#$ of length at least $k$, excluding suffix $T[|T|-k+2 \ltdots |T|]\#$. We use \Algo{algo:iterator} to enumerate the internal nodes of $\ST_{T\#}$, and every time we enumerate a node $v$ we proceed as follows. Let $\ell(v)=W$. If $|W|<k$ we leave $C(T,k)$ unaltered, otherwise we increment $C(T,k)$ by one and we decrement $C(T,k)$ by the number of children of $v$ in $\ST_{T\#}$, which is equal to $|\mathtt{chars}|$ in $\REPR{W}$. It follows that every node $v$ of $\ST_{T\#}$ that is located at depth at least $k$ and that is not the locus of a $k$-mer is both added to $C(T,k)$ (when the algorithm visits $v$) and subtracted from $C(T,k)$ (when the algorithm visits $\mathtt{parent}(v)$). Leaves at depth at least $k$ are added by the initialization of $C(T,k)$, and subtracted during the enumeration. Conversely, every locus $v$ of a $k$-mer of $T$ (including leaves) is just added to $C(T,k)$, because $|\ell(\mathtt{parent}(v))|<k$. The claimed complexity comes from \Thm{thm:bwtConstruction} and \Thm{thm:iterator}.
\end{proof}

A number of other kernels and complexity measures can be implemented on top of Algorithms \ref{algo:iterator} and \ref{algo:generalizedIterator}: see \cite{belazzougui2015framework} for details. Since such iteration algorithms work on data structures that can be built from the input strings in deterministic linear time, all such kernels and complexity measures can be computed from the input strings in deterministic $O(n)$ time and in $O(n\log{\sigma})$ bits of working space, where $n$ is the total length of the input strings.

\section*{Acknowledgement}
The authors wish to thank Travis Gagie for explaining the data structure built in Lemma \ref{lemma:stringAncestor}, as well as for valuable comments and encouragements, Gonzalo Navarro for explaining the algorithm in Theorem \ref{thm:mums}, Enno Ohlebusch for useful comments and remarks, and Alexandru Tomescu for valuable comments and encouragements.
\newpage
\bibliographystyle{plain}
\small
\bibliography{jacm}
\normalsize

\end{document}